\documentclass[twoside,leqno]{article}

\usepackage[letterpaper]{geometry}

\usepackage{ltexpprt}

\usepackage[utf8]{inputenc}
\usepackage[dvipsnames]{xcolor}
\usepackage{import,todonotes,xspace,mathtools,amsfonts,soul,paralist,complexity}
\usepackage[hidelinks]{hyperref}
\usepackage[capitalise]{cleveref}
\DeclareFontShape{OT1}{cmr}{bx}{sc}{<-> cmbcsc10}{} % for bold \textsc

\newcommand{\classC}{\mathcal{C}}
\newcommand{\classD}{\mathcal{D}}
\newcommand{\Bc}{\mathsf{B}}
\newcommand{\Cc}{\mathsf{C}}
\newcommand{\Dc}{\mathsf{D}}
\newcommand{\Tc}{\mathsf{T}}
\newcommand{\Uc}{\mathsf{U}}
\newcommand{\Vc}{\mathsf{V}}
\newcommand{\Zc}{\mathsf{Z}}
\newcommand{\Ll}{\mathcal{L}}
\newcommand{\Pp}{\mathcal{P}}
\newcommand{\Rr}{\mathcal{R}}
\newcommand{\N}{\mathbb{N}}
\newcommand{\Nat}{\mathbb{N}}
\newcommand{\set}[1]{\{#1\}}
\renewcommand{\class}[1]{[#1]_\approx}
\newcommand{\classk}[2]{[#1]_{#2}}
\newcommand{\interval}[1]{[#1]}
\newcommand{\quotient}[1]{[#1]}
\newcommand{\eats}[1]{\mathord{|#1\rangle}}
\newcommand{\Act}{\mathbb{A}}
\newcommand{\dist}{\textup{\textsc{Dist}}}
\newcommand{\sgrow}{\textup{\textsc{StackGrowth}}}
\newcommand{\Below}{\mathit{Below}}
\newcommand{\tpeq}{\stackrel{\mathsf{tp}}{=}}
\newcommand{\decomp}[1]{\langle\!|#1|\!\rangle}
\newcommand{\rea}[2]{#1^{[#2]}}
\newcommand{\Near}{\textup{\textsc{Near}}}
\newcommand{\Cons}{\textup{\textsc{Cons}}}

\newcommand{\init}{\mathsf{init}}
\newcommand{\initstate}{q_\init X_\init}
\newcommand{\tuple}{\boldsymbol{\tau}}
\newcommand{\tupleS}{\boldsymbol{\sigma}}
\newcommand{\eqlev}{\textup{\textsc{EqLev}}}
\newcommand{\EXPTIME}{\textup{\textsc{ExpTime}}\xspace}
\newcommand{\PTIME}{\textup{\textsc{PTime}}\xspace}
\newcommand{\ACKERMANN}{\textup{\textsc{Ackermann}}\xspace}
\renewcommand{\PSPACE}{\textup{\textsc{PSpace}}\xspace}
\renewcommand{\EXPSPACE}{\textup{\textsc{ExpSpace}}\xspace}

\newcommand{\PDS}{$\varepsilon$-PDS\xspace}
\newcommand{\weak}[1]{\stackrel{#1}{\Longrightarrow}}
\newcommand{\rewrite}[1]{\xhookrightarrow{\phantom{a\!\!}#1\phantom{a\!\!}}}
\newcommand{\pos}{\textup{\textsc{Pos}}}
\newcommand{\Pos}{\pos}
\renewcommand{\M}{\mathcal{M}}
\newcommand{\Configurations}{\mathsf{Conf}}
\newcommand{\ppoly}{\ell}
\newcommand{\MSB}{\mathsf{MSB}}
\newcommand{\Bin}{\alpha}
\newcommand{\Binsucc}{\textup{\textsc{Inc}}}
\newcommand{\Enc}{\beta}
\newcommand{\Source}{\textup{\textsc{Source}}}
\newcommand{\Defender}{\textup{\textsc{Def}}}
\newcommand{\Target}{\textup{\textsc{Target}}}
\newcommand{\Choose}{\textup{\textsc{Choose}}}
\newcommand{\Wait}{\textup{\textsc{Wait}}}
\newcommand{\Succ}{\textup{\textsc{Succ}}}
\newcommand{\Desc}{\textup{\textsc{Desc}}}
\newcommand{\Final}{\textup{\textsc{Final}}}
\newcommand{\Play}{\textup{\textsc{Play}}}
\newcommand{\Temp}{\textup{\textsc{Temp}}}
\newcommand{\Pop}{\textup{\textsc{Pop}}}
\newcommand{\Push}{\textup{\textsc{Push}}}
\newcommand{\Done}{\textup{\textsc{Done}}}
\newcommand{\Init}{\textup{\textsc{Init}}}
\newcommand{\IniChk}{\textup{\textsc{IniChk}}}

\makeatletter
\newcommand{\qed}{\@endproof}
\newenvironment{proofof}[1]{\@beginproof{Proof of #1}}{\@endproof}
\newenvironment{proofnoqed}{\@beginproof{Proof}}{}
\newenvironment{proofnoqedof}[1]{\@beginproof{Proof of #1}}{}
\makeatother
\newcommand{\qedhere}{\qquad\vbox{\hrule height0.6pt\hbox{%
   \vrule height1.3ex width0.6pt\hskip0.8ex
   \vrule width0.6pt}\hrule height0.6pt}}

\newcommand{\problem}{{\sc Weak Bisimulation Finiteness for \PDS}\xspace}

\renewcommand{\epsilon}{\varepsilon}
\renewcommand{\rho}{\varrho}
\renewcommand{\hat}{\widehat}
\newcommand{\gs}{\mathsf{gs}}
\newcommand{\up}{\mathsf{up}}

\newcommand{\problemx}[3]{
\par\noindent\underline{\sc#1}\par\nobreak\vskip.2\baselineskip
\begingroup\clubpenalty10000\widowpenalty10000
\setbox0\hbox{\bf INPUT:\ }\setbox1\hbox{\bf QUESTION:\ }
\dimen0=\wd0\ifnum\wd1>\dimen0\dimen0=\wd1\fi
\vskip-\parskip\noindent
\hbox to\dimen0{\box0\hfil}\hangindent\dimen0\hangafter1\ignorespaces#2\par
\vskip-\parskip\noindent
\hbox to\dimen0{\box1\hfil}\hangindent\dimen0\hangafter1\ignorespaces#3\par
\endgroup}

\newtheorem{claim}{Claim}[section]
\Crefname{@theorem}{Theorem}{Theorems}
\Crefname{Definition}{Definition}{Definitions}
\Crefname{lemma}{Lemma}{Lemmata}
\Crefname{equation}{Line}{Lines}\crefrangelabelformat{equation}{(#3#1#4)--(#5#2#6)}
\Crefname{inequality}{Inequality}{Inequalities}\creflabelformat{inequality}{(#2#1#3)}
\Crefname{equivalence}{Equivalence}{Equivalences}\creflabelformat{equivalence}{(#2#1#3)}
\Crefname{property}{Property}{Properties}\creflabelformat{property}{(#2#1#3)}
\NewEnvironmentCopy{inequalities}{align}\AddToHook{env/inequalities/begin}{\crefalias{equation}{inequality}}
\NewEnvironmentCopy{equivalences}{align}\AddToHook{env/equivalences/begin}{\crefalias{equation}{equivalence}}
\NewEnvironmentCopy{properties}{align}\AddToHook{env/properties/begin}{\crefalias{equation}{property}}

\begin{document}

\newcommand\relatedversion{}

\title{Weak Bisimulation Finiteness of Pushdown Systems
With Deterministic $\varepsilon$-Transitions Is 2-\EXPTIME-Complete}
\author{Stefan G\"oller\thanks{University of Kassel, School of Electrical Engineering and Computer Science.
The author was supported by Agence nationale de la recherche (grant no.\@ ANR-17-CE40-0010).
}
\and Paweł Parys\thanks{University of Warsaw, Institute of Informatics. 
The author was supported by the National Science Centre, Poland (grant no.\@ 2021/41/B/ST6/00535).}}

\date{}

\maketitle

% Copyright Statement
% When submitting your final paper to a SIAM proceedings, it is requested that you include
% the appropriate copyright in the footer of the paper.  The copyright added should be
% consistent with the copyright selected on the copyright form submitted with the paper.
% Please note that "20XX" should be changed to the year of the meeting.

% Default Copyright Statement
%\fancyfoot[R]{\scriptsize{Copyright \textcopyright\ 2023 by SIAM\\
%Unauthorized reproduction of this article is prohibited}}

% Depending on which copyright you agree to when you sign the copyright form, the copyright
% can be changed to one of the following after commenting out the default copyright statement
% above.

%\fancyfoot[R]{\scriptsize{Copyright \textcopyright\ 20XX\\
%Copyright for this paper is retained by authors}}

%\fancyfoot[R]{\scriptsize{Copyright \textcopyright\ 20XX\\
%Copyright retained by principal author's organization}}

%\pagenumbering{arabic}
%\setcounter{page}{1}%Leave this line commented out.

\begin{abstract} \small\baselineskip=9pt 
	We consider the problem of deciding whether a given pushdown system all of
	whose $\epsilon$-transitions are deterministic is weakly bisimulation finite, that is,
	whether it is weakly bisimulation equivalent to a finite system.
	We prove that this problem is 2-\EXPTIME-complete.
	This consists of three elements:
	First, we prove that the smallest finite system that is weakly bisimulation equivalent to a
	fixed pushdown system, if exists,
	has size at most doubly exponential in the description size of the pushdown system.
	Second, we propose a fast algorithm deciding whether a given pushdown system is weakly
	bisimulation equivalent to a finite system of a given size.
	Third, we prove 2-\EXPTIME-hardness of the problem.
	The problem was known to be decidable, but the previous algorithm had Ackermannian complexity
	(6-\EXPSPACE in the easier case of pushdown systems without $\epsilon$-transitions);
	concerning lower bounds, only \EXPTIME-hardness was known.
\end{abstract}

\section{Introduction}

An important decision problem in computer science
is to decide whether a given infinite system is semantically finite, that is,
whether it is semantically equivalent to some finite system.
If so, particular techniques and properties
of finite systems can be exploited in order to verify
such systems.
There are different important notions of equivalences that have been studied
by the computer science community.

In the area of verification {\em bisimulation} equivalence~\cite{Stirling98}
can be seen as the central one.
It can be seen as a two-player game between Attacker and Defender:
given a pair of configurations $(c,d)$ of a system,
Attacker chooses a
transition $c\rightarrow_a c'$ (resp.\@ $d\rightarrow_a d'$) and Defender
must find a reply
$d\rightarrow_a d'$ (resp.\@ $c\rightarrow_a c'$), hereby leading the game to a new
pair of configurations $(c',d')$ ---
Attacker wins in case Defender cannot answer, whereas Defender
wins every infinite
play and every play terminating in a pair of dead ends.
It is worth pointing out that trace equivalence coincides with
bisimulation equivalence (bisimilarity for short) in case the underlying systems are deterministic.
Several central verification logics like modal logic, the modal $\mu$-calculus,
CTL$^*$, and Propositional Dynamic Logic can all be characterized as the bisimulation-invariant
fragment of well-established logics like first-order logic~\cite{JvB1976}, monadic second-order
logic~\cite{JW96}, monadic chain logic~\cite{MR03}, and
weak monadic chain logic~\cite{Carreiro15}, respectively.
In presence of possible $\epsilon$-transitions {\em weak bisimulation equivalence}
generalizes bisimulation equivalence~\cite{vanGlabbeek90} in that Attacker
can make moves of the
form $c\rightarrow_{\varepsilon}^*\circ\rightarrow_a\circ
\rightarrow_{\varepsilon}^*c'$
(or $c\rightarrow_\varepsilon^*c'$) that
can be answered by Defender by moves of the form
$d\rightarrow_\varepsilon^*\circ\rightarrow_a\circ\rightarrow_\varepsilon^*d'$
(resp.\@ $d\rightarrow_\varepsilon^*d'$).
In the context of (weak) bisimulation equivalence, the semantic finiteness question becomes the
{\em (weak) bisimulation finiteness} problem: given an infinite system, does there exist a
{\em finite} system that is (weakly) bisimilar to it:
it is important to emphasize that it is asked whether {\em there exists} a finite system that
is (weakly) bisimilar to the infinite system, so the finite system is not specified in the input to the problem.
Indeed, in case both the infinite system and the finite system are part of the input, their equivalence
can be reduced to the model checking problem of CTL's fragment EF~\cite{KM10}. The (complexity of the) model checking problem for most classes of infinite systems is well-understood.

To date, it is fair to say that decidability --- and in particular the complexity ---
of this principal and easily-stated problem
of (weak) bisimulation finiteness is not well understood.
We refer to Srba's survey~\cite{Srba04}, where it becomes clear that for many classes
of infinite state systems decidability is unknown, and if it is known, oftentimes huge complexity gaps exist.
A central such class is the class of pushdown systems, that is, systems that can be generated
by pushdown automata. Model checking monadic second-order logic is decidable for them~\cite{MS85},
whereas its bisimulation-invariant fragment can be model-checked in $\EXPTIME$~\cite{Walukiewicz01}.
Only a few years ago Jan\v{c}ar proved that weak bisimulation finiteness of pushdown systems
all of whose $\varepsilon$-transitions are deterministic, is decidable with Ackermannian complexity~\cite{Jancar20}
via a reduction to the \ACKERMANN-complete
(weak) bisimulation equivalence problem~\cite{JS19,YL020}.
Recently, G\"oller and Parys~\cite{LICS20} have shown that the Ackermannian bottleneck can avoided
for $\varepsilon$-free pushdown systems: bisimulation finiteness of ($\varepsilon$-free) pushdown
systems is in $6$-\EXPSPACE.
The problem is known to be $\EXPTIME$-hard~\cite{KM02,Srba02}.

\paragraph{Our contribution.}

We prove that weak bisimulation finiteness of pushdown systems
all of whose $\varepsilon$-transitions are deterministic is 2-\EXPTIME-complete.
The problem was known to be decidable~\cite{Jancar20} with Ackermannian complexity
and in 6-\EXPSPACE in the easier case of pushdown systems without $\epsilon$-transitions.
Concerning lower bounds, only \EXPTIME-hardness was known.
Our upper bound also generalizes Valiant's 2-\EXPTIME upper bound~\cite{Valiant75} for the
regularity problem of deterministic pushdown automata from 1975: regularity of deterministic
pushdown automata is (modulo simple adaptations)
the same problem as weak bisimulation finiteness of
pushdown systems
all of whose transitions (not only $\varepsilon$-transitions) are deterministic.
Our contribution consists of three elements.
First, we prove that the smallest finite system that is weakly bisimulation equivalent to a
fixed pushdown system, if exists,
has size at most doubly exponential in the description size of the pushdown system.
Second, we propose a fast algorithm deciding whether a given pushdown system is weakly
bisimulation equivalent to a finite system of a given size.
Third, we prove 2-\EXPTIME-hardness of the problem.

\paragraph{Related and future work.}

As mentioned above, (weak) bisimulation finiteness is
a problem that is not well understood~\cite{Srba04}.
Let us mention a few exceptions.
Over Petri net's subclass of basic parallel processes
bisimulation finiteness is \PSPACE-complete~\cite{Srba02b,Kot05}.
Bisimulation finiteness of one-counter systems (which
are pushdown systems where there is, apart
from the bottom symbol, only one stack symbol)
is \PTIME-complete~\cite{BGJ14}.
For weak bisimulation finiteness a relevant result is
undecidability for Petri nets~\cite{JE96},
whereas decidability seems to be open for most
other central infinite systems~\cite{Srba04}.
We hope that our results can pave the way to eventually
determining the decidability/complexity status
of weak bisimulation finiteness
of pushdown systems whose $\varepsilon$-transitions
are not restricted to be deterministic.

\paragraph{Organization of the paper.}

We introduce basic notation and state our main result in \cref{sec:preliminaries}.
Our 2-\EXPTIME upper bound is sketched in \cref{sec:overview}.
Basics on pushdown systems are subject of \cref{sec:basics}.
In \cref{sec:decompositions} we discuss decompositions of stacks.
In \cref{S Runs down} we analyze runs in pushdown systems
that mainly decrease the stack height.
The core arguments of the upper bound proof are content of \cref{sec:core}.
In \cref{sec:algorithm} our algorithm running in double exponential time 
is presented.
A matching 2-\EXPTIME lower bound is given in \cref{sec:lower-full}.
We conclude in \cref{sec:conclusion}.

\section{Preliminaries}\label{sec:preliminaries}

If $X$ is a set we denote by $2^X$ {\em power set of $X$}, that is, the set of all subsets of $X$.
By $\N=\{0,1,2,\dots,\}$ we denote the set of non-negative integers.
For all finite alphabets $\Sigma$ we denote
by $\Sigma^*$ the set of finite words over $\Sigma$
and, for all $n\in\N$, we denote by $\Sigma^{\leq n}=\{w\in\Sigma^*\mid|w|\leq n\}$ the set of
finite words in $\Sigma^*$ of length at most $n$.
The {\em empty word} is denoted by $\varepsilon$.
By $\Sigma_\varepsilon$ we denote the (disjoint) union $\Sigma\cup\{\varepsilon\}$.

A {\em labeled transition system with $\varepsilon$-transitions ($\varepsilon$-LTS)} is a tuple
$\Ll=(S,\Act,\allowbreak(\rightarrow_a\nobreak)_{a\in\Act_\varepsilon})$,
where $S$ is a (possibly infinite) set of {\em configurations},
$\Act$ is a finite set of {\em action symbols},
$(\rightarrow_a)\subseteq S\times S$ is a binary relation for
all $a\in\Act_\varepsilon$.
We say $\Ll$ is {\em finite} if $S$ is finite.
We define its {\em size} as $|\Ll|=|S|$, thus $|\Ll|\in\Nat$ if $\Ll$ is finite
and $|\Ll|=\omega$ if not (we only consider countable labeled transition systems in this paper).
A {\em pointed $\varepsilon$-LTS} is a pair $(\Ll,c)$ such that $c$ is a configuration of $\Ll$.
We define the relation $(\weak{\varepsilon})=(\rightarrow_\varepsilon^*)$.
The relation is extended to words in $\Act^+$ as follows:
for $a\in\Act$ and $w\in\Act^*$
we define $(\weak{aw})=(\rightarrow_\varepsilon^*)\circ(\rightarrow_a)\circ(\weak{w})$.
Thus, note that $(\weak{a})=(\rightarrow_\varepsilon^*)\circ(\rightarrow_a)\circ(\rightarrow^*_\epsilon)$.
We define the binary relation $(\rightarrow)=\bigcup_{a\in\Act_\varepsilon}(\rightarrow_a)$.
For all $c,d\in S$ we define
$\dist(c,d)=\min\set{|w|\mid c\weak{w}d}\in\Nat\cup\set{\omega}$,
the length of the shortest word of action symbols allowing to reach $d$ from $c$ in $\Ll$.

For an $\varepsilon$-LTS $\Ll$ we say a binary relation $R_2\subseteq S\times S$
{\em weakly covers} a relation $R_1\subseteq S\times S$ if for all
$c,d,c'$ such that $(c,d)\in R_1$ and $c\weak{a}c'$ with $a\in\Act_\epsilon$
there exists $d'$ such that
$d\weak{a}d'$ and $(c',d')\in R_2$.
A {\em weak bisimulation} is a relation $R\subseteq S\times S$ that is symmetric and covers itself.
Observe that the union of two weak bisimulations is again a weak bisimulation.
We write $c\approx d$ if $(c,d)\in R$ for some weak bisimulation relation $R$;
note that $(\approx)\subseteq S\times S$ is the largest weak bisimulation on $S$.
If $c\approx d$, we say that $c$ and $d$ are {\em weakly bisimilar}.
For every configuration $c\in S$ we denote by $\class{c}=\set{d\in S\mid c\approx d}$
the {\em weak bisimulation class of $c$}.
The {\em weak bisimulation quotient} $\quotient{\Ll}$ is the $\varepsilon$-LTS
$\quotient{\Ll}=(\set{\class{c}\mid c\in S},\Act_\varepsilon,(\rightarrow_a')_{a\in\Act_\varepsilon})$,
where for $a\in\Act_\varepsilon$ we have $\classC\rightarrow_a'\classD$ if $c\weak{a} d$ for some $c\in\classC$, $d\in\classD$
(note that then for every $c\in\classC$ there exists some $d\in\classD$ such that $c\weak{a}d$).
When talking about weak bisimulation equivalence of two configurations, we must not necessarily require that
the configurations are from the same $\varepsilon$-LTS: generally
we can write $(\Ll_1,c_1)\approx(\Ll,c_2)$ if
$c_1\approx c_2$ holds in the disjoint union of $\Ll_1$ and $\Ll_2$.
We say that $(\Ll,c)$ (or just $c$ when $\Ll$ is clear from the context) is
{\em weakly bisimulation finite} if
$(\Ll,c)\approx(\Ll',c')$ for some finite pointed $\varepsilon$-LTS
$(\Ll',c')$.

We also define relations $(\approx_k)\subseteq S\times S$ for $k\in\Nat$, by induction:
$\approx_0$ is the full relation $S\times S$, and $\approx_{k+1}$ is the largest symmetric relation that is covered by $\approx_k$.
For configurations $c,c'\in S$ we define
$\classk{c}{k}=\set{d\in S\mid c\approx_k d}$, and
$\eqlev(c,c')=\sup\set{k\in\Nat\mid c\approx_k c'}\in\Nat\cup\set{\omega}$.

If for every configuration $c\in S$ and every $a\in\Act$ there are finitely many configurations $d\in S$ such that $c\weak{a}d$,
then $(\approx)=\bigcap_{k\in\Nat}(\approx_k)$, and $\eqlev(c,c')=\omega$ implies $c\approx c'$
(below we restrict ourselves to $\epsilon$-LTSs that have this property).

A sequence $c_0\rightarrow_{a_1}c_1\rightarrow_{a_2}\dots\rightarrow_{a_n}c_n$ is called a \emph{run} from $c_0$ to $c_n$.
For such a run $\rho$ we write $|\rho|$ for $n$, and $\rho(i)$ for $c_i$.
A composition $\rho\circ\rho'$ of two runs is defined in the expected way, assuming $\rho'(0)=\rho(|\rho|)$.
We say that runs $\rho=(c_0\rightarrow_{a_1}c_1\rightarrow_{a_2}\dots\rightarrow_{a_m}c_m)$
and $\rho'=(d_0\rightarrow_{b_1}d_1\rightarrow_{b_2}\dots\rightarrow_{b_n}d_n)$ are \emph{parallel}
if we can find indices $0=i_0<i_1<\dots<i_k=m$ and $0=i'_0<i'_1<\dots<i'_k=n$
such that for all $j\in\interval{0,k}$ we have $c_{i_j}\approx d_{i'_j}$
and for all $j\in\interval{1,k}$ the two words $a_{i_{j-1}+1}a_{i_{j-1}+2}\dots a_{i_j}$ and $b_{i'_{j-1}+1}b_{i'_{j-1}+2}\dots b_{i'_j}$
(after dropping all $\epsilon$'s) are equal and have length at most $1$.

A {\em pushdown system with deterministic $\varepsilon$-transitions}
(\PDS\ for short) is a tuple
$\Pp=(Q,\Gamma,\allowbreak \Act,\allowbreak\Delta)$, where
$Q$ is a finite set of {\em control states},
$\Gamma$ is a finite {\em stack alphabet},
$\Act$ is finite set of {\em action symbols},
$\Delta\subseteq Q\times\Gamma\times \Act_\varepsilon \times Q\times\Gamma^*$ is a finite {\em rewrite relation}
such that whenever $(p_1,X_1)=(p_2,X_2)$ and $a_1=\varepsilon$ for two tuples $(p_1,X_1,a_1,q_1,\beta_1),(p_2,\allowbreak X_2,\allowbreak a_2,q_2,\beta_2)\in\Delta$,
then $(p_1,X_1,a_1,q_1,\beta_1)=(p_2,X_2,a_2,q_2,\beta_2)$.

In case $(p,X,a,q,\beta)\in\Delta$ we simply write $qX\rewrite{a}q'\beta$ and refer to it
as a {\em rule}: in case $a\in\Act$ we call it a {\em reading rule} in case $a=\varepsilon$
we call it an {\em $\varepsilon$-rule}.
We say a pair $(p,X)\in Q\times\Gamma$ {\em is in $\varepsilon$-mode} if for some $(q,\beta)\in Q\times\Gamma^*$ we have
$(p,X,\varepsilon,q,\beta)\in\Delta$; we remark that there can be at most one such pair $(q,\beta)$.
If $(p,X)$ is in $\varepsilon$-mode, then there is no $(a,q',\beta')\in\Act\times Q\times\Gamma^*$ with
$(p,X,a,q',\beta')\in\Delta$.
We say that an \PDS\ is {\em $\varepsilon$-popping} if for all $(p,X,\varepsilon,q,\beta)\in\Delta$ we have $\beta=\varepsilon$.

The {\em size} of $\Pp$ is defined as $|\Pp|=|Q|+|\Gamma|+|\Act|+|\Delta|$.
A {\em configuration} of an \PDS\ is an element from $Q\Gamma^*$.
An \PDS\ induces a (potentially infinite) $\varepsilon$-LTS $\Ll(\Pp)=(S,
\Act_\varepsilon,\allowbreak(\rightarrow_a\nobreak)_{a\in\Act_\varepsilon})$,
where $S=Q\Gamma^*$ and for all $a\in\Act_\varepsilon$ we have
$(\rightarrow_a)=\{(pX\gamma,q\beta\gamma)\mid \gamma\in\Gamma^*,(p,X,a,q,\beta)\in\Delta\}$.

In this paper, we are interested in the following decision problem:\medskip
\problemx{\problem}{A \PDS $\Pp=(Q,\Gamma,\Act,\Delta)$ and a configuration $\initstate\in Q\Gamma$.}
{Is $(\Ll(\Pp),\initstate)$ weakly bisimulation finite?}
\medskip

We remark that a seemingly more general problem when having an initial configuration from $Q\Gamma^*$ as input can
always be reduced in polynomial time to our restriction, namely where it is from $Q\Gamma$.
It is folklore that, given an \PDS\ $\Pp$ and a configuration $q_\init\alpha_\init\in Q\Gamma^*$,
	one can compute in polynomial time
	an \PDS\ $\Pp'$ that is $\varepsilon$-popping and a configuration $\initstate\in Q\Gamma$ such that $(\Ll(\Pp),q_\init\alpha_\init)\approx(\Ll(\Pp'),\initstate)$.
A formal proof of this reduction can be found
in~\cite[Proposition 11]{JS19}.
Additionally, and this is also folklore, one can restrict the
resulting
$\Pp'$ in such a way that all rules
$(p,X,a,q,\beta)$ that appear in it satisfy
$|\beta|\leq 2$.

Let us state the main result of this paper:

\begin{theorem}
	\problem is 2-\EXPTIME-complete.
\end{theorem}

\section{Overview of the upper-bound proof}\label{sec:overview}

In this section we present overall ideas
for proving that \problem is in $2$-\EXPTIME.
As already mentioned in the introduction, the proof consists of two main parts.
First, in \cref{thm:main}, we prove that if an \PDS $\Pp$ is weakly bisimulation finite, that is, is weakly bisimilar to some finite $\epsilon$-LTS (of an arbitrary, unknown size),
then it is weakly bisimilar to some finite $\epsilon$-LTS of at most doubly exponential size.
This already gives a $2$-\EXPSPACE algorithm for deciding weak bisimulation finiteness by a result by Ku\v{c}era and Mayr~\cite{KM10},
who proved that checking whether a given
\PDS is weakly bisimilar to a given finite
$\varepsilon$-LTS is \PSPACE-complete:
it is enough to enumerate all finite $\epsilon$-LTSs of at most doubly exponential size, and for each of them check whether it is weakly bisimilar to the given \PDS $\Pp$.
We improve this in \cref{thm:algorithm}, where we show how to compute the weak bisimulation quotient of $\Pp$ (i.e., the smallest $\epsilon$-LTS weakly bisimilar to $\Pp$)
in time polynomial in its size (and exponential in $|\Pp|$).

In this description we concentrate on the proof of \cref{thm:main}, which is the essential part.
The second part (the algorithm) is much shorter, and is described in \cref{sec:algorithm}.

For simplicity, in this overview we concentrate on the case of an $\epsilon$-free PDS $\Pp$ (i.e., an \PDS without any $\epsilon$-transition),
which is already non-trivial.
We then briefly explain how ideas from the $\epsilon$-free case can be extended to the general case, when deterministic $\epsilon$-transitions are present.
In particular, when there are no $\epsilon$-transitions, the notions of bisimulation and of weak bisimulation coincide, so we drop the word ``weak'' in this description.
Moreover, we formulate some statements here in a loose way, neglecting some details. Full details can be found in
the following sections.

Let us assume that our initial configuration
$\initstate$ of our $\epsilon$-free PDS $\Pp=(Q,\Gamma,\Act,\Delta)$ is bisimulation finite.
Consider an arbitrary configuration $q\delta$ reachable from
$\initstate$.
Our first step is to locate some possibilities of pumping in the stack content $\delta$.
Namely, we represent $\delta$ as $\alpha\beta\gamma$, so that configurations $q\alpha\beta^i\gamma$, obtained by repeating the
$\beta$ part (we require $\beta\not=\varepsilon$),
are ``similar'' to the original configuration $q\delta$
in the following sense:
first, all these configurations are also reachable
(from $\initstate$) and, second, the set of
control states reachable after popping
the topmost part $\alpha\beta^j$ is the same for
every $j$.
A standard application of the pigeonhole
principle allows us to find such $\alpha$ and $\beta$ of at most exponential size.
This gives us at most a doubly exponential number of different pairs $\alpha,\beta$.
It is thus enough to show that for each fixed pair $\alpha,\beta$, configurations of the form $q\alpha\beta\gamma$ constitute only a doubly exponential number of bisimulation classes.

Next, we shift the study of configurations of the form
$q\alpha\beta\gamma$ to the study of configurations of the form $r\gamma$,
but where we still assume that $r\gamma$ can be reached from reachable configurations of the form $q\alpha\beta^i\gamma$, for all $i\geq 1$.
It is enough to prove that the number of classes of these configurations $r\gamma$ is small,
because the class of $q\alpha\beta\gamma$ is determined by the classes of configurations $r\gamma$, obtainable by popping the fixed stack prefix $\alpha\beta$ (cf.~\cref{equiv-because-equiv-below}).

We now exploit our assumption that our initial
configuration $\initstate$ is bisimulation finite;
it can reach $F$ classes, for
some $F\in\Nat$ (at this point we do not have yet any bound on $F$, we only know that it is finite).
An interesting property of finite $\epsilon$-LTSs is that the $\approx_F$ relation coincides with $\approx$.
In other words, if two configurations are not bisimilar, then this can be detected in the first $F$ steps, that is,
while reading at most $F$ action symbols.
In the $\epsilon$-free case, reading at most $F$ action symbols is equivalent to performing at most $F$ transitions.
Bearing in mind $\beta\not=\varepsilon$, it follows that configurations of the form $q\alpha\beta^i\gamma$ for $i\geq F$ are all bisimilar
(here $q$, $\alpha$, and $\beta$ are fixed, but the number $i\geq F$ and the stack content $\gamma$ are arbitrary):
during the first $F$ transitions, $\Pp$ can pop at most $F$ topmost stack symbols, and they are identical in all these configurations.
If we take an even larger bound, say $F+k$, then not only $q\alpha\beta^i\gamma\approx q\alpha\beta^{i'}\gamma'$ for all $i,i'\geq F+k$ and $\gamma,\gamma'\in\Gamma^*$,
but also after executing a run of length $k$ from $q\alpha\beta^i\gamma$ and an analogous (i.e., performing the same transitions) run from $q\alpha\beta^{i'}\gamma'$,
the resulting configurations remain bisimilar.
We utilize this observation by introducing stack contents of the form $\alpha\beta^\omega\gamma$.
The $\omega$ exponent can be formalized in three ways.
First, we may assume that $\omega$ denotes some ``very large'' finite number (large enough for purposes of the proof).
A second formalization is that the $\beta$ part of the stack is repeated infinitely many times.
In our proofs we have chosen yet another formalization, where we understand $\beta^\omega$ as a formal expression such that after popping $\beta$ from $\beta^\omega$ we again have $\beta^\omega$.
Nevertheless, no matter which formalization one chooses,
it is only important that the finite (although possibly very long) prefix of the stack that may be analyzed by the PDS consists of one $\alpha$ and repetitions of $\beta$.

\paragraph{Bounding the number of repetitions.}

Next, given some stack contents $\alpha,\beta$, and $\gamma$
we are interested in a possibly small $n\in\N$
such that
$q\alpha\beta^n\gamma\approx q\alpha\beta^\omega\gamma$.
We fix $e\in\Nat$ to be the smallest number such that $q\alpha\beta^e\gamma$ becomes equivalent to $q\alpha\beta^\omega\gamma$
(strictly speaking, we require that $r\beta^e\gamma\approx r\beta^\omega\gamma$ for all control states $r$ reachable after popping $\alpha$;
this implies that $q\alpha\beta^e\gamma\approx q\alpha\beta^\omega\gamma$, but it is a slightly stronger condition).
We remark that Jančar~\cite{Jancar20} has shown an Ackermannian bound for $e$.
For us, an important step is to prove that $e$ can be at most doubly exponential.
To this end, for every $i\in\interval{0,e}$, we consider the
control state $r_i$ for which $\eqlev(r_i\beta^i\gamma,r_i\beta^\omega\gamma)$ is minimal
(i.e., the control state from which the stack contents $\beta^i\gamma$ and $\beta^\omega\gamma$ can be distinguished in the least possible number of transitions),
and for this control state we take $M_i=\eqlev(r_i\beta^i\gamma,r_i\beta^\omega\gamma)$.
It is rather easy to prove that $M_i<M_{i+1}$ for all $i<e$:
the similarity between $\beta^i\gamma$ and $\beta^\omega\gamma$ grows when we increase $i$, until it reaches level $\omega$ for $i=e$.
In particular (because $\eqlev(\cdot,\cdot)$ depends only on the class of a configuration) this means that if $r_i=r_{i'}$ for $i\neq i'$,
then the classes $\class{r_i\beta^i\gamma}$ and $\class{r_i\beta^{i'}\gamma}$ are different.

Given some stack contents $\alpha,\beta$, and $\gamma$
and some $i\in\interval{0,e-1}$,
we now want to provide a succinct description of the class of $r_i\beta^i\gamma$.
To this end, we consider two runs.
The first of them, $\pi$, is the shortest possible run from $q\alpha\beta^e\gamma$ to $r_i\beta^i\gamma$.
The second of them, $\pi'$, is created as a run parallel to $\pi$ that starts in the ``much larger'' configuration $q\alpha\beta^\omega\gamma$
(it exists because $q\alpha\beta^\omega\gamma\approx q\alpha\beta^e\gamma$).
Recall that then $\pi$ and $\pi'$ are required to visit bisimilar configurations,
but they are \emph{not} required to execute the same transitions
(in fact, they cannot execute the same transitions: if
$\pi'$, starting from $q\alpha\beta^\omega$, were to pop
$\alpha\beta^{e-i}$ in the same way as $\pi$
did,
then $\pi'$ would end in $r_i\beta^\omega\gamma$, a configuration that is not bisimilar to
$r_i\beta^i\gamma$ by assumption).
Depending on the shape of $\pi'$, we have two cases.

First, it is possible that no suffix of
$\pi'$ pushes $\Dc$ stack symbols, for an appropriately chosen exponential constant $\Dc$;
so, roughly speaking, $\pi'$ concentrates on popping.
We then observe that no matter how much we pop from $q\alpha\beta^\omega\gamma$,
we can
only pop to a stack content of the form $\beta'\beta^\omega\gamma$, where $\beta'$ is a suffix of either $\alpha$ or $\beta$.
Even taking into account the fact that exponentially many symbols may be pushed at the very end, on top of $\beta'$,
in this case we have only doubly exponentially
many possibilities for the configuration $\pi'(|\pi'|)$,
hence also for its class, which is simultaneously the class of $r_i\beta^i\gamma$.
Recalling that for different values of $i$ (having the same state $r_i$) the classes of $r_i\beta^i\gamma$ are all different,
this means that only doubly exponentially many values of $i$ may be handled by this case.
We can thus concentrate on the opposite case,
which we discuss next.

This opposite case, being significantly more
complicated, is that $\pi'$ has a suffix $\pi'_3$ that pushes $\Dc$ stack symbols
(note that if more than $\Dc$ symbols are pushed, then we can consider a shorter suffix, pushing exactly $\Dc$ symbols).
The suffix $\pi'_3$ leads from $t_1X\theta$ to $t\mu\theta$ for some control states $t_1,t$, stack contents $\mu,\theta$,
and stack symbol $X$, where $|\mu|=\Dc+1$.
We now use the simple observation
that the class of $t\mu\theta$ is determined by the control state $t$,
by the small stack content $\mu$,
and by the classes of $u\theta$ for all control
states $u$ reachable from $t\mu\theta$ after popping $\mu$ (as formalized in \cref{equiv-because-equiv-below}).

Next, we would somehow like to describe the classes of $u\theta$.
To this end, we observe that since
$u\theta$ is reachable from $t\mu\theta$, it
is also reachable from $t_1X\theta$
(due to the existence of the run $\pi'_3$ from
$t_1X\theta$ to $t\mu\theta$).
Moreover, a simple pumping argument (namely, \cref{can pop fast}) allows us to shorten the run from
$t_1X\theta$ to $u\theta$
(observe that, in total, only one stack symbol
needs to be popped)
into a very short run; namely, shorter than some constant $\Bc$ (being exponential in $|\Pp|$).
Coming back to the analysis of $\pi$,
note that the parallel counterpart of
$\pi'_3(0)=t_1X\theta$ in $\pi'$ is, say, the configuration
$q_1\chi\beta^j\gamma$ in $\pi$: it is visited
while going from $q\alpha\beta^e\gamma$ to $r_i\beta^i\gamma$.
The $\chi$ part is obtained from $\beta$ (or $\alpha$) by popping
its prefix as quickly as possible by definition of
$\pi$. Thus, standard pumping arguments imply that
the length of $\chi$ is short.
Possibly decreasing $j$ and appending a few copies of $\beta$ into $\chi$, we may artificially redefine $|\chi|$
to be larger than $\Bc$, but not much larger.
If the value of the constant $\Dc$ is appropriate
so that $\pi_3'$ is long enough, then also the
subrun of $\pi$ from $q_1\chi\beta^j\gamma$ to
$r_i\beta^i\gamma$ is long enough to ensure that $j>i$.

We have thus two bisimilar configurations, $q_1\chi\beta^j\gamma$ and $t_1X\theta$, and we have a run of length $\Bc$ from $t_1X\theta$ to $u\theta$.
Then a run of the same length exists also from $q_1\chi\beta^j\gamma$ to a configuration bisimilar to $u\theta$;
this configuration is necessarily of the form $v_u\xi_u\beta^j\gamma$ (we cannot pop the whole $\chi$ with $|\chi|>\Bc$ in only $\Bc$ steps),
where the size of $\xi_u$ is again bounded exponentially.
We can find such $v_u$ and $\xi_u$ for all considered states $u$ (note that $j$ does not depend on $u$).

The classes of $v_u\xi_u\beta^j\gamma$ in turn
are determined by the control states $v_u$, by the small stack contents $\xi_u$, and by the classes $r\beta^j\gamma$
for all control states $r$ reachable from $v_u\xi_u\beta^j\gamma$ after popping $\xi_u$ (hence also reachable from $q\alpha\beta^e\gamma$ after popping $\alpha\beta^{e-j}$).

Summing up this part, we have shown that the class of $r_i\beta^i\gamma$ can
be determined by the control state/stack pairs $t\mu$ and
$v_u\xi_u$ (indexed by control states $u$ from an appropriate set $U$)
and by the classes of $r\beta^j\gamma$ for some $j>i$.
All the stack contents $\mu$ and $\xi_u$ have
at most exponential size.
The number of possibilities for $(t\mu,(v_u\xi_u)_{u\in U})$ is then
doubly exponential.

Is it possible that the same tuple $(t\mu,(v_u\xi_u)_{u\in U})$ is assigned to two distinct indices $i,i'<e$, such that moreover $r_i=r_{i'}$?
We prove that this is not possible, which immediately implies a
doubly exponential upper bound for $e$, being the number of considered indices $i\in\interval{0,e-1}$.
Suppose thus, to the contrary, we had two such indices
$i,i'$, where $i<i'$.
Then the class of $r_i\beta^i\gamma$ is determined by the classes of $r\beta^j\gamma$, for some $j>i$,
in the same way as the class of $r_i\beta^{i'}\gamma$ is determined by the classes of $r\beta^{j'}\gamma$, for some $j'>i'>i$ (due to equality of the descriptions assigned to $i$ and $i'$).
Recalling the inequalities $M_i+1\leq M_j$ and
$M_i+1\leq M_{i'}< M_{j'}$, and the definition of $M_j$ and $M_{j'}$,
we have $r\beta^j\gamma\approx_{M_i+1}r\beta^\omega\gamma\approx_{M_i+1}r\beta^{j'}\gamma$ for all control states $r$ under
consideration.
It follows easily by the equality of their descriptions,
that then also $r_i\beta^i\gamma\approx_{M_i+1}r_i\beta^{i'}\gamma$
(we depend in the same way on $\approx_{M_i+1}$-equivalent configurations, so we remain $\approx_{M_i+1}$-equivalent).
Due to $r_i\beta^{i'}\gamma\approx_{M_i+1}r_i\beta^\omega\gamma$ this implies $r_i\beta^i\gamma\approx_{M_i+1}r_i\beta^\omega\gamma$,
which contradicts with $\eqlev(r_i\beta^i\gamma,r_i\beta^\omega\gamma)=M_i$.
This finishes the proof of the doubly exponential
bound on $e$.

\paragraph{Bounding the number of classes.}

We now come back to our initial goal, namely bounding
the number of classes of $r\gamma$, where $r\gamma$ is reachable from the initial configuration via $q\alpha\beta^k\gamma$ for all $k\geq 1$,
where $\alpha$ and $\beta$ are fixed and short.
We employ here exactly the same characterization of the classes as in the previous part, where we were bounding $e$.
This time we only take $i=0$ (i.e., we consider configurations $r\beta^0\gamma=r\gamma$),
but we allow an arbitrary control state $r$ in place of $r_i$,
and do not assume that $\gamma$ is fixed.
As previously, we have two possibilities.

First, it may happen that $r\gamma$ is bisimilar to a configuration of the form $t\beta'\beta^\omega\gamma$
that is reachable from $q\alpha\beta^\omega\gamma$ by a run $\pi'$ that ``focused on popping'',
that is, belonging to the first case we have considered.
As $\beta'$ is small (at most exponential), and
since the class of $t\beta'\beta^\omega\gamma$ does not depend on $\gamma$
(as already explained, $(\approx)=(\approx_F)$, while $\gamma$ is too deep in the stack to be seen in the first $F$ transitions),
we only have a doubly exponential number of classes
of $r\gamma$ in this case.

The second case is that the class of $r\gamma$ is determined by a small (exponential) information $(t\mu,(v_u,\xi_u)_{u\in U})$ and by the classes of $r'\beta^j\gamma$ for some $j>i=0$.
The latter classes, in turn, are determined by the stack content $\beta^j$ (hence by $j$, because $\beta$ is fixed) and by the classes of $r''\gamma$.
Thus, if we add $j$ to the remembered information, we can say that the class of $r\gamma$ is determined by $(t\mu,(v_u,\xi_u)_{u\in U},j)$ and by the classes of $r''\gamma$.
The previously shown doubly exponential bound on $e$ implies that the number of possibilities for $(t\mu,(v_u,\xi_u)_{u\in U},j)$ is
doubly exponential.
A minor detail is that the tuple $(t\mu,(v_u,\xi_u)_{u\in U},j)$ describes the class of $r\gamma$ for a single control
state $r$;
we should rather, given a stack content $\gamma$, consider a tuple of such tuples, indexed by states $r$ from an appropriate set $R$
(control states reachable from $q\alpha\beta^e\gamma$ after popping $\alpha\beta^e$),
hence describing the classes of $r\gamma$ simultaneously for all $r\in R$.
We now have a cyclic situation:
the classes of $r\gamma$ are determined by the small
a small tuple (indexed by control states)
of the above-mentioned information and by the classes of
$r\gamma$ themselves.
One can show that in such a situation, the small information is enough to determine the classes of $r\gamma$.
We thus have a doubly exponential
bound on the number of these classes.

\paragraph{Adding $\epsilon$-transitions.}

We now briefly discuss on adapting the above proof idea to the general case, where $\epsilon$-transitions may be present.
Recall that every pushing transition reads some action symbol;
$\epsilon$-transitions are allowed only for popping
and need to be deterministic.
It turns out that if, while popping $\beta$,
at least one action symbol is necessarily read, then the proof sketch presented above still works, up to adjusting some details.
The main difficulty comes with the fact that it may be possible to pop arbitrarily many copies of $\beta$ without reading any action symbols.
In particular, starting from a configuration of the form $q\alpha\beta^k\gamma$,
the $\gamma$ part of the stack may be possibly reached after reading very few action symbols, no matter how large $k$ is.

In order to deal with this difficulty, we proceed as follows:
instead of splitting the whole stack content into a
single \emph{pumping triple} $\alpha\beta\gamma$,
we now rather try to find pumping triples $\alpha\beta\gamma$ with slightly stronger properties, but being only infixes of the whole stack.
If some infix of the stack content is represented as a pumping triple $\alpha\beta\gamma$, being located on top of a stack content $\eta$,
then we still have the previous requirement saying, roughly, that the $\beta$ infix may be pumped.
Beside that, we have a new requirement: if $r\beta^i\gamma\eta$ is reached by an $\epsilon$-run popping $\beta$ from some larger configuration $r'\beta^{i+1}\gamma\eta$,
then further $\epsilon$-transitions allow to pop the whole $\beta^i\gamma$, and reach a configuration of the form $s\eta$.

As already said, the previous proof works correctly when the considered runs from $q\alpha\beta^e\gamma\eta$ (and likewise from $q\alpha\beta^\omega\gamma\eta$)
read at least one action symbol while popping every copy of $\beta$.
There are also runs that at some moment pop a copy of $\beta$ using only $\epsilon$-transitions.
Then these runs (possibly after prolonging them) continue with popping $\epsilon$-transitions until the stack content $\eta$ is uncovered.
In the proof we need to add special cases for such runs.
It turns out, however, that these runs do not introduce many new possibilities.
An intuitive reason for this is that if a run, while being in some configuration $r\beta^i\gamma\eta$, starts performing $\epsilon$-transitions until $s\eta$ is reached,
then the number $i$ and the stack content $\gamma$ are ``forgotten''; neither the target configuration nor the action symbols read on the way (we do not read anything) depend on $i$ or on $\gamma$.

In \cref{sec:decompositions} we introduce decompositions of stacks.
They give a way to decompose the whole stack content into many pumping triples $\alpha\beta\gamma$, located possibly in a nested way.
In this decomposition, most of the stack content, except for exponentially many symbols, is put into the $\gamma$ parts of pumping triples.
The proof presented above, dealing with a single pumping triple, is made formal and complete in \cref{main lemma}, which is our main technical lemma.
Then, an appropriate induction gives us the final
upper-bound result, \cref{thm:main}.

\section{Some basics on pushdown systems}\label{sec:basics}

In this section we present some definitions and known facts about pushdown systems, useful in our proofs.
For purposes of the whole section let $\Pp=(Q,\Gamma,\Act,\Delta)$ be an \PDS.

Our first \lcnamecref{can pop fast} says that if there is a run between two similar configurations, then there is a short run between them.
Essentially, this boils down to the standard pumping lemma for pushdown automata.

\begin{lemma}[{\cite[Lemma 3.3]{LICS20}}]\label{can pop fast}
	There exists a constant\footnote{%
		Constants depending on $\Pp$ are denoted by capital letters in Sans Serif font;
		for constants singly exponential in the size of $\Pp$ we use initial letters of the alphabet ($\Bc,\Cc,\Dc$),
		while for doubly exponential constants we use letters near the end of the alphabet $(\Tc,\Uc,\Vc,\Zc$).
	} $\Bc\in 2^{|\Pp|^{O(1)}}$ such
	that whenever $p\alpha\to^*q\beta$ for two configurations $p\alpha, q\beta\in Q\Gamma^*$, then $\dist(p\alpha,q\beta)\leq(|\alpha|+|\beta|)\cdot\Bc$.
\qed\end{lemma}

For every stack content $\alpha\in\Gamma^*$, we define two sets, describing possible ways of pushing $\alpha$:
\begin{align*}
	\up(\alpha)&=\set{(pX,qY)\in(Q\Gamma)^2\mid pX\to^*qY\alpha}&&\mbox{for $\alpha\in\Gamma^*$, and}\\
	\up_0(\alpha)&=\set{(pX,q)\in(Q\Gamma)\times Q\mid pX\to^*q\alpha}&&\mbox{for $\alpha\in\Gamma^*\setminus\set{\epsilon}$.}
\end{align*}
It is easy to see that for all $\alpha,\beta,\gamma\in\Gamma^*$ with $\alpha\neq\epsilon$,
\begin{align*}
	\up(\beta\gamma)=\up(\gamma)\circ\up(\beta)&&\mbox{and}&&
	\up_0(\alpha\beta)&=\up(\beta)\circ\up_0(\alpha)\,,
\end{align*}
where ``$\circ$'' denotes the composition of relations, understood in the usual way.

As explained in \cref{sec:overview}, it is convenient to extend the set of configurations into expressions containing $\omega$, which we do as follows.
A \emph{generalized stack} is defined by induction: it is a
(possibly empty) sequence
$\alpha_1\alpha_2\dots\alpha_k$, where every $\alpha_i$ is either
\begin{compactitem}
\item	a stack symbol from $\Gamma$, or
\item	an expression of the form $(\beta_1\beta_2\dots\beta_\ell)^\omega$, where $\beta_1\beta_2\dots\beta_\ell$ is a nonempty (i.e., satisfying $\ell\geq 1$) generalized stack.
\end{compactitem}
The set of generalized stacks is denoted $\Gamma^\gs$.
A \emph{generalized configuration} is an element of $Q\Gamma^\gs$.
From now on, whenever we talk about a configuration, we usually mean a generalized one.
Configurations that are not generalized, are called \emph{standard}.
Note that \cref{can pop fast} and the definition of $\up(\cdot)$ were only given for standard configurations.

Let us emphasize that, formally, generalized stacks are expressions (terms);
we do not assign any concrete meaning to the $(\cdot)^\omega$ operation.
However, the intuition staying behind is that $(\beta_1\beta_2\dots\beta_\ell)^\omega$ describes a stack in which the sequence $\beta_1\beta_2\dots\beta_\ell$ is repeated very many times.

We allow ourselves to write $\alpha^e$ when $e\in\Nat\cup\set{\omega}$, mixing two different semantics of the same notation:
for $e\in\N$ we mean the result of repeating $e$ times the sequence $\alpha$, while for $e=\omega$ we mean the expression $\alpha^\omega$.

For $a\in\Act_\epsilon$ we define $\to_a$ to be the unique relation on $Q\Gamma^\gs$
that holds only in the following cases:
\begin{compactenum}
\item	if $(p,X,a,q,\beta)\in\Delta$, then $pX\gamma\to_a q\beta\gamma$ for all $\gamma\in\Gamma^\gs$,
\item	if $p\alpha\to_a q\beta$ and $(p\alpha,p)\not\in(\to_\epsilon^*)$, then $p\alpha^\omega\gamma\to_a q\beta\alpha^\omega\gamma$, and
\item	if $p\alpha\to_\epsilon^*p$, then $p\alpha^\omega\gamma\to_\epsilon p\gamma$.
\end{compactenum}
Note that in order to know whether $p\alpha^\omega\gamma\to_a q\beta\alpha^\omega\gamma$ or $p\alpha^\omega\gamma\to_\epsilon p\gamma$,
one has to check whether $p\alpha\to_a q\beta$ and whether $p\alpha\to_\epsilon^*p$.
It is important here that the nesting depth of the $(\cdot)^\omega$ operation in $\alpha$ is strictly smaller than in $\alpha^\omega\gamma$,
so the whole definition is well-formed.

Based on that, we then define the $\weak{w}$ relation over generalized configurations, as previously.
Note that, when restricted to standard configurations, the above definition coincides with the previous one.
Moreover, only standard configurations may be reached from $\initstate$.

Item 2 of the above definition treats $p\alpha^\omega\gamma$ almost as $p\alpha\alpha^\omega\gamma$:
we can go to $q\beta\alpha^\omega\gamma$ when there is a transition from $p\alpha$ to $q\beta$.
There is one exception, though: if $p\alpha\to_\epsilon^*p$, then we rather pop the whole $\alpha^\omega$, going directly to $p\gamma$ (cf.~Item 3).
The intuition is that if $\alpha$ may be popped using $\epsilon$-transitions,
then also arbitrarily many copies of $\alpha$, embodied in the $\alpha^\omega$ expression, may be popped using $\epsilon$-transitions.
Transitions from Item 2 are disallowed in this case, so that $\epsilon$-transitions remain deterministic.
In Item 3 it is required that the state $p$ before and after popping $\alpha$ is the same;
then it makes sense to say that after popping very many copies of $\alpha$ the state will be again $p$.
If we had $p\alpha\to_\epsilon^*q$ and $q\alpha\to_\epsilon^*p$, then it would be unclear in which state (in $p$ or in $q$) we should finish popping $\alpha^\omega$.
For this reason we do not add any special transition in such a case.

The following property is a direct consequence of determinism of $\epsilon$-transitions:
\begin{quote}
	if $c\to^*_\epsilon d$, then $c\approx d$ (and $c\approx_k d$ for all $k\in\Nat)$.
\end{quote}
Yet another useful property of generalized configurations is the following:
\begin{quote}\label{expand-equiv}
	for all $p\in Q$ and $\alpha,\gamma\in\Gamma^\gs$ we have $p\alpha^\omega\gamma\approx p\alpha\alpha^\omega\gamma$.
\end{quote}
Indeed, if $p\alpha\to^*_\epsilon p$ (hence also $p\alpha\alpha^\omega\gamma\to_\epsilon^* p\alpha^\omega\gamma$), equivalence follows from the previous property;
otherwise, exactly the same transitions can be performed from $p\alpha^\omega\gamma$ and from $p\alpha\alpha^\omega\gamma$, leading to exactly the same configurations.

\medskip

We say that a run $\rho'$ is a \emph{shift} of a run $\rho$ if $\rho'$ is obtained by appending or removing the same suffix to every configuration in $\rho$.

For $\alpha\in\Gamma^\gs$ and $P\subseteq Q$ we define two functions, saying how states can change while popping $\alpha$:
\begin{align*}
	\eats{\alpha}(P)=\set{r\in Q\mid \exists q\in P.\,q\alpha\to^*r}&&\mbox{and}&&
	\eats{\alpha}_\epsilon(P)=\set{r\in Q\mid \exists q\in P.\,q\alpha\to_\epsilon^*r}\,.
\end{align*}
For singleton sets we abbreviate $\eats{\alpha}(\set{q})$ as $\eats{\alpha}(q)$,
and likewise we abbreviate $\eats{\alpha}_\epsilon(\set{q})$ as $\eats{\alpha}_\epsilon(q)$.
Observe that
\begin{align*}
	\eats{\beta}(\eats{\alpha}(P))=\eats{\alpha\beta}(P),&&
	\eats{\beta}_\epsilon(\eats{\alpha}_\epsilon(P))=\eats{\alpha\beta}_\epsilon(P),&&\mbox{and}&&
	\eats{\alpha}_\epsilon(P)\subseteq\eats{\alpha}(P)
\end{align*}
for all $\alpha,\beta\in\Gamma^\gs$ and $P\subseteq Q$.
The following \lcnamecref{run2eats} is a direct consequence of the definition:

\begin{lemma}\label{run2eats}
	Let $p,q\in Q$ and $\alpha,\beta\in\Gamma^\gs$.
	If $p\alpha\to^*q\beta$, then $\eats{\alpha}(p)\supseteq\eats{\beta}(q)$.
\qed\end{lemma}

The next property follows from determinism of $\epsilon$-transitions:

\begin{lemma}\label{epsilon-det}
	Let $q\in Q$ and $\alpha\in\Gamma^\gs$.
	If $\eats{\alpha}_\epsilon(q)\neq\emptyset$, then $\eats{\alpha}(q)=\eats{\alpha}_\epsilon(q)$, and this set has size $1$.
\end{lemma}

\begin{proof}
	If $q\alpha\to_\epsilon^*r$, then we cannot have $q\alpha\to^*s$ for any other state $s\neq r$.
\end{proof}

Recall that $\dist$ counts the number of action symbols needed to reach some configuration.
Because $\epsilon$-transitions are not counted, there may be arbitrarily many configurations in a given distance.
We fix this by defining $\Near$:
for $p\alpha\in Q\Gamma^\gs$ and $k\in\Nat$
let $\Near(p\alpha,k)$ be the set of configurations that can be reached from $p\alpha$ by a run reading at most $k$ action symbols and not ending with an $\epsilon$-transition.

\begin{lemma}\label{determ short run}
	For all $p\alpha\in Q\Gamma^\gs$ and $k\in\Nat$ we have
	\begin{align*}
		|\set{\classC\mid\dist(\class{p\alpha},\classC)\leq k}|\leq|\Near(p\alpha,k)|\leq|\Pp|^k\,.
	\end{align*}
\end{lemma}

\begin{proofnoqed}
	If $\dist(\class{p\alpha},\classC)\leq k$, then there is a run $\rho$ from $p\alpha$ to a configuration $c\in\classC$,
	reading at most $k$ action symbols.
	Let $c'$ be the first configuration reached by $\rho$ after reading all these action symbols;
	after $c'$ we have only $\epsilon$-transitions.
	Then $c'\in\Near(p\alpha,k)$, and $c'\approx c$, so $c'\in\classC$.
	Thus in every class from the set $\set{\classC\mid\dist(\class{p\alpha},\classC)\leq k}$ there is an element of the set $\Near(p\alpha,k)$,
	so the former set cannot be larger;
	we obtain the first inequality.
	
	To obtain the second inequality, we observe that every configuration in $\Near(p\alpha,k)$ is fully characterized by the list of at most $k$ transitions
	performed while reading action symbols by a run not ending with an $\epsilon$-transition.
	Indeed, from every configuration $c$ we can go by $\epsilon$-transitions to at most one configuration $d$ such that action symbols can be read from $d$.
	When we are in $d$, we perform a transition from our list.
	Then again we go by $\epsilon$-transitions to the unique configuration from which action symbols can be read, and so on.
	At the very end, we have to stop immediately after performing the last transition from the list.
	The number of lists of at most $k$ transitions is
	\begin{align*}
		\sum\nolimits_{i=0}^k|\Delta|^i\leq(|\Delta|+1)^k\leq|\Pp|^k\,.
	\qedhere\end{align*}
\end{proofnoqed}

For a length-$1$ run $\rho$ let
\begin{align*}
	\sgrow(\rho)=\left\{\begin{array}{ll}
		|\beta|-1&\mbox{if }\rho=(pX\gamma\to_a q\beta\gamma),\\
		\sgrow(p\alpha\to_a q\beta)&\mbox{if }\rho=(p\alpha^\omega\gamma\to_a q\beta\alpha^\omega\gamma),\mbox{ and}\\
		-\omega&\mbox{if }\rho=(p\alpha^\omega\gamma\to_\epsilon p\gamma),
	\end{array}\right.
\end{align*}
and for a run $\rho=\rho_1\circ\dots\circ\rho_n$ of length $n$ let
\begin{align*}
	\sgrow(\rho)=\sum_{i=1}^n\sgrow(\rho_i)\,.
\end{align*}
Here we assume that $-\omega+n=n-\omega=-\omega-\omega=-\omega$ for every $n\in\N$.
Note that if $\rho$ is a run between standard configurations $p\alpha,q\beta\in Q\Gamma^*$, then $\sgrow(\rho)=|\beta|-|\alpha|$.

Finally, we state two simple but important \lcnamecrefs{equiv-because-equiv-below}, which are useful while proving that two configurations are
weakly bisimilar
(similar lemmata appeared already in prior work~\cite{LICS20,Jancar-arxiv,Jancar20}):

\begin{lemma}\label{equiv-because-equiv-below}
	For all $q\in Q$ and $\alpha,\eta,\eta'\in\Gamma^\gs$ we have
	\begin{align*}
		\eqlev(q\alpha\eta,q\alpha\eta')\geq\min\set{\eqlev(r\eta,r\eta')\mid r\in\eats{\alpha}(q)}\,.
	\end{align*}
	In other words, if $r\eta\approx r\eta'$ for all $r\in\eats{\alpha}(q)$ then $q\alpha\eta\approx q\alpha\eta'$,
	and if $r\eta\approx_k r\eta'$ for all $r\in\eats{\alpha}(q)$ then $q\alpha\eta\approx_k q\alpha\eta'$.
\end{lemma}

\begin{proof}
	It should be clear that the second part is equivalent to the first part.
	To see the first part, we consider two cases:
	If there is some $r\in\eats{\alpha}_\epsilon(q)\subseteq\eats{\alpha}(q)$, then $q\alpha\eta\to_\epsilon^*r\eta$ and $q\alpha\eta'\to_\epsilon^*r\eta'$;
	we have $q\alpha\eta\approx r\eta$ and $q\alpha\eta'\approx r\eta'$, so $\eqlev(q\alpha\eta,q\alpha\eta')=\eqlev(r\eta,r\eta')$; the inequality follows.
	Otherwise, if $\eats{\alpha}_\epsilon(q)=\emptyset$, we just use \cref{more-equiv-if-read}, stated below.
\end{proof}

\begin{lemma}\label{more-equiv-if-read}
	Let $q\in Q$ and $\alpha,\eta,\eta'\in\Gamma^\gs$.
	If $\eats{\alpha}_\epsilon(q)=\emptyset$, then
	\begin{align*}
		\eqlev(q\alpha\eta,q\alpha\eta')\geq 1+\min\set{\eqlev(r\eta,r\eta')\mid r\in\eats{\alpha}(q)}\,.
	\end{align*}
\end{lemma}

\begin{proof}
	\newcommand{\labelStar}{($\star$)}
	\newcommand{\PropertyStar}{Property \labelStar\xspace}
	We prove for every $k\in\Nat$ that if $\eats{\alpha}_\epsilon(q)=\emptyset$ and $r\eta\approx_k r\eta'$ for all $r\in\eats{\alpha}(q)$, then $q\alpha\eta\approx_{k+1}q\alpha\eta'$;
	the thesis of the lemma follows easily from this statement.
	The proof is by induction on $k$.
	Consider two configurations $q\alpha\eta$ and $q\alpha\eta'$ such that $\eats{\alpha}_\epsilon(q)=\emptyset$ and $r\eta\approx_k r\eta'$ for all $r\in\eats{\alpha}(q)$;
	we want to prove that $q\alpha\eta\approx_{k+1}q\alpha\eta'$.
	Before starting observe that
	\begin{quote}
		if $q\alpha\to^*s\beta$ then $s\beta\eta\approx_k s\beta\eta'$.
		\hfill\labelStar
	\end{quote}
	Indeed, if $k=0$, this holds trivially (all configurations are in the $\approx_0$ relation).
	If $k\geq 1$
	and there is some $r\in\eats{\beta}_\epsilon(s)\subseteq\eats{\alpha}(q)$ (inclusion by \cref{run2eats}),
	then $s\beta\eta\to_\epsilon^*r\eta$ and $s\beta\eta'\to_\epsilon^*r\eta'$ imply that
	$s\beta\eta\approx r\eta\approx_k r\eta'\approx s\beta\eta'$.
	If $k\geq 1$ and $\eats{\beta}_\epsilon(s)=\emptyset$,
	we observe that $r\eta\approx_{k-1} r\eta'$ for all $r\in\eats{\beta}(s)$ (due to $(\approx_k)\subseteq(\approx_{k-1})$ and $\eats{\beta}(s)\subseteq\eats{\alpha}(q)$),
	so we have $s\beta\eta\approx_k s\beta\eta'$ by the induction hypothesis.
	This finishes the proof of \PropertyStar.
	
	In order to obtain $q\alpha\eta\approx_{k+1}q\alpha\eta'$,
	it is enough to show that the two-element symmetric relation $\set{(q\alpha\eta,q\alpha\eta'),(q\alpha\eta',q\alpha\eta)}$ is covered by $\approx_k$;
	the relation will be then contained in $\approx_{k+1}$, the largest symmetric relation covered by $\approx_k$.
	Expanding the definition of being covered,
	we should thus prove that if $q\alpha\eta\weak{a}c$ with $a\in\Act_\epsilon$, then there exists $c'$ such that $c\approx_k c'$ and $q\alpha\eta'\weak{a}c'$
	(and that if $q\alpha\eta'\weak{a}c'$, then there exists $c$ such that $c'\approx_k c$ and $q\alpha\eta\weak{a}c$;
	this part is symmetric, so we do not need a separate proof).
	Suppose thus that $q\alpha\eta\weak{a}c$.
	If $a=\epsilon$, we can take $c'=q\alpha\eta'$; then $c\approx q\alpha\eta\approx_k q\alpha\eta'=c'$ by \PropertyStar and obviously $q\alpha\eta'\weak{\epsilon}c'$.
	If $a\in\Act$, then $q\alpha\eta\to_\epsilon^* c_1\to_a c_2\to_\epsilon^* c$ for some $c_1,c_2$.
	Because $\eats{\alpha}_\epsilon(q)=\emptyset$, no configuration on a run from $q\alpha\eta$ to $c_1$ may have stack $\eta$;
	we have $c_1=s_1\beta\eta$ for some $s\in Q$ and $\beta\in\Gamma^\gs$ with $\beta\neq\epsilon$.
	Then also $q\alpha\eta'\to_\epsilon^*s_1\beta\eta'$.
	A single transition cannot pop more than the whole $\beta$,
	so $c_2=s_2\gamma\eta$ for some $s_2\in Q$ and $\gamma\in\Gamma^\gs$;
	then also $s_1\beta\eta'\to_a s_2\gamma\eta'$.
	We take $c'=s_2\gamma\eta'$.
	We have $q\alpha\eta'\weak{a}c'$, and we obtain $c\approx c_2=s_2\gamma\eta\approx_k s_2\gamma\eta=c'$ by \PropertyStar.
\end{proof}

\section{Decompositions of stacks}\label{sec:decompositions}

In this section we define decompositions of stacks, prove that they exist, and show their basic properties.
In essence, decompositions specify how to split a stack content in a way suitable for pumping.
For the whole section we fix an \PDS $\Pp=(Q,\Gamma,\Act,\Delta)$ together with an initial configuration $\initstate$.

A \emph{decomposition} $\overline\delta$ is defined by induction, and can be in one of three forms:
\begin{compactitem}
\item	a single symbol (base case): $\overline\delta=X\in\Gamma$,
\item	a nonempty sequence of decompositions: $\overline\delta=\overline\delta_1\dots\overline\delta_k$, where $\overline\delta_1,\dots,\overline\delta_k$ are decompositions and $k\geq 1$, or
\item	a \emph{pumping triple}: $\overline\delta=\decomp{\overline\alpha,\overline\beta,\gamma}$,
	where $\overline\alpha$ and $\overline\beta$ are decompositions and $\gamma\in\Gamma^*$ is a stack.
\end{compactitem}
The {\em height} (resp.\@ {\em degree}) of a decomposition $\overline\delta$ is inductively defined
as follows:
\begin{compactitem}
	\item the height and degree of $\overline{\delta}=X\in\Gamma$ is $0$,
	\item the height (resp.\@ degree) of 
		$\overline{\delta}=\overline{\delta}_1\dots\overline{\delta}_k$
		is $1$ plus the maximum of heights of any of the $\overline{\delta}_i$
		(resp.\@ the maximum of $k$ and the degree of any
		of the $\overline{\delta}_i$), and
	\item the height (resp.\@ degree) of 
		$\overline{\delta}=\decomp{\overline\alpha,\overline\beta,\gamma}$
		is $1$ plus the maximum of the height of $\overline{\alpha}$
		and the height of $\overline{\beta}$ (resp.\@ the maximum of $2$,
		the degree of $\overline{\alpha}$, and the
		degree of $\overline{\beta}$).
\end{compactitem}
In other words, a decomposition can be seen as a tree for which, as expected, 
the height is the length of some longest path
from the root to some leaf and the degree is the maximal number of children 
of some node in the tree.

For a decomposition $\overline\delta$ and for $e\in\Nat\cup\set{\omega}$ we define a stack $\rea{\overline\delta}{e}\in\Gamma^*$, called the \emph{$e$-th realization} of $\overline\delta$:
\begin{compactitem}
\item	$\rea{X}{e}=X$;
\item	$\rea{(\overline\delta_1\dots\overline\delta_k)}{e}=\rea{\overline\delta}{e}_1\dots\rea{\overline\delta}{e}_k$;
\item	$\rea{\decomp{\overline\alpha,\overline\beta,\gamma}}{e}=\rea{\overline\alpha}{e}(\rea{\overline\beta}{e})^e\gamma$ (i.e., we repeat the middle part $e$ times).
\end{compactitem}
If $\delta=\rea{\overline\delta}{1}$, we say that $\overline\delta$ is a \emph{decomposition of $\delta$}.
Note that $\rea{\overline\delta}{1}$ is obtained by simply concatenating all stack symbols that literally appear in $\overline\delta$,
and that $\rea{\overline\delta}{e}$ for greater numbers $e$ is obtained by repeating $e$ times some fragments of $\rea{\overline\delta}{1}$, in a nested way.
For $e=\omega$ the intuition is that we repeat infinitely many times appropriate fragments of the stack, but formally we just write $\omega$ in the exponent, without repeating anything.

Having the above definition, we can say when a decomposition is \emph{well-formed}:
A decomposition consisting of a single stack symbol is always well-formed,
a decomposition $\overline\delta_1\dots\overline\delta_k$ is a well-formed when all $\overline\delta_i$ are well-formed,
and a decomposition $\overline\delta=\decomp{\overline\alpha,\overline\beta,\gamma}$ is well-formed when both $\overline\alpha$ and $\overline\beta$ are well-formed and additionally
\begin{compactenum}
\item	$\up(\gamma)=\up(\rea{\overline\beta}{1}\gamma)$,
\item	$\eats{\rea{\overline\alpha}{\omega}}=\eats{\rea{\overline\alpha}{\omega}\rea{\overline\beta}{\omega}}$, and
\item	$\eats{\rea{\overline\beta}{\omega}}_\epsilon(r)=\set{r}$ and $\eats{\gamma}_\epsilon(r)\neq\emptyset$ for all $r\in\eats{\rea{\overline\beta}{\omega}}_\epsilon(Q)$.
\end{compactenum}
Below, whenever talking about a decomposition, we assume that it is well-formed.

Let us explain the above conditions.
Item 1 implies that $\up(\gamma)=\up(\rea{(\overline\beta}{1})^e\gamma)$ for every $e\in\Nat$,
so it ensures that when we pump a stack content of a reachable configuration according to a decomposition, then the obtained configuration is again reachable.
Item 2 implies that $\eats{\rea{\overline\alpha}{\omega}}=\eats{\rea{\overline\alpha}{\omega}(\rea{\overline\beta}{\omega})^e}$ for every $e\in\Nat$:
no matter how many copies of $\rea{\overline\beta}{\omega}$ we have, the set of states that can be reached by popping the stack is the same.
Item 3 says that there are only two ways of popping $\rea{\overline\beta}{\omega}$:
either we read at least one action symbol, or
we end in a state from which there is an $\epsilon$-run popping arbitrarily many copies of $\rea{\overline\beta}{\omega}$ and then $\gamma$;
in particular, the only way of accessing the $\gamma$ part of $\rea{\overline\delta}{\omega}$ is by an $\epsilon$-run popping the whole $\gamma$.

There are some similarities between our pumping triples $\decomp{\overline\alpha,\overline\beta,\gamma}$ and notions from Valiant's paper~\cite{Valiant75}.
Namely, the $\overline\beta$ part of such a triple, which can be pumped, corresponds to Valiant's notion of a null-transparent segment,
and the $\gamma$ part corresponds to Valiant's notion of an invisible segment
(it is ``invisible'' after repeating $\overline\beta$ many times).
There is one important difference, though: in Valiant's proof it is enough to find just one invisible segment and one null-transparent segment;
here, we rather need to decompose the whole stack into such segments.

We now give two simple lemmata, being immediate consequences of well-formedness:

\begin{lemma}\label{easy-eats-decomp}
	For every well-formed pumping triple $\overline\delta=\decomp{\overline\alpha,\overline\beta,\gamma}$ and every set $P\subseteq Q$ we have
	\begin{compactenum}
	\item	$\eats{(\rea{\overline\beta}{\omega})^\omega}_\epsilon(P)=\eats{\rea{\overline\beta}{\omega}}_\epsilon(P)$,
	\item	$\eats{\rea{\overline\delta}{\omega}}_\epsilon(P)=\eats{\rea{\overline\alpha}{\omega}\rea{\overline\beta}{\omega}\gamma}_\epsilon(P)$, and
	\item	$\eats{\rea{\overline\delta}{\omega}}(P)=\eats{\rea{\overline\beta}{\omega}\gamma}_\epsilon(\eats{\rea{\overline\alpha}{\omega}}(P))$ for all $P\subseteq Q$.
	\end{compactenum}
\end{lemma}

\begin{proof}
	For the left-to-right inclusion in Item 1,
	suppose that $r\in\eats{(\rea{\overline\beta}{\omega})^\omega}_\epsilon(r_0)$ for some $r_0\in P$,
	which by definition means that $r_0(\rea{\overline\beta}{\omega})^\omega\to_\epsilon^*r$.
	We have to prove that $r\in\eats{\rea{\overline\beta}{\omega}}_\epsilon(r_0)$.
	The only way of popping $(\rea{\overline\beta}{\omega})^\omega$ is to pop $\rea{\overline\beta}{\omega}$ some number of times,
	and then to pop the whole $(\rea{\overline\beta}{\omega})^\omega$ in a single step, without changing the state;
	we thus have states $r_1,\dots,r_k$ for some $k\in\Nat$ such that $r_i\in\eats{\rea{\overline\beta}{\omega}}_\epsilon(r_{i-1})$ for $i\in\interval{1,k}$,
	and $r_k=r\in\eats{\rea{\overline\beta}{\omega}}_\epsilon(r_k)$.
	If $k=0$, the latter gives us the thesis.
	If $k\geq 1$, for $i\in\interval{1,k-1}$ we have $r_i\in\eats{\rea{\overline\beta}{\omega}}_\epsilon(r_{i-1})\subseteq\eats{\rea{\overline\beta}{\omega}}_\epsilon(Q)$,
	so well-formedness of $\overline\delta$ implies that $\eats{\rea{\overline\beta}{\omega}}_\epsilon(r_i)=\set{r_i}$;
	because $r_{i+1}\in\eats{\rea{\overline\beta}{\omega}}_\epsilon(r_i)$, necessarily $r_{i+1}=r_i$.
	We thus have $r_1=r_2=\dots=r_k=r$, hence $r_1\in\eats{\rea{\overline\beta}{\omega}}_\epsilon(r_0)$ gives us the thesis.
	
	Conversely, if $r\in\eats{\rea{\overline\beta}{\omega}}_\epsilon(r_0)$ for some $r_0\in P$,
	then by well-formedness of $\overline\delta$ we also have $r\in\eats{\rea{\overline\beta}{\omega}}_\epsilon(r)$.
	We thus have $r_0(\rea{\overline\beta}{\omega})^\omega\to_\epsilon^* r(\rea{\overline\beta}{\omega})^\omega\to_\epsilon^* r$, that is,
	$r\in\eats{(\rea{\overline\beta}{\omega})^\omega}_\epsilon(r_0)$.
	This finishes the proof of Item 1.

	We have
	$\eats{\rea{\overline\delta}{\omega}}_\epsilon(P)
	=\eats{\gamma}_\epsilon(\eats{(\rea{\overline\beta}{\omega})^\omega}_\epsilon(\eats{\rea{\overline\alpha}{\omega}}_\epsilon(P)))$
	and
	$\eats{\rea{\overline\alpha}{\omega}\rea{\overline\beta}{\omega}\gamma}_\epsilon(P)
	=\eats{\gamma}_\epsilon(\eats{\rea{\overline\beta}{\omega}}_\epsilon(\eats{\rea{\overline\alpha}{\omega}}_\epsilon(P)))$,
	so Item 2 is an immediate consequence of Item 1.
	
	In order to prove that $\eats{\rea{\overline\delta}{\omega}}(P)\subseteq\eats{\rea{\overline\beta}{\omega}\gamma}_\epsilon(\eats{\rea{\overline\alpha}{\omega}}(P))$
	consider a run that pops $\rea{\overline\delta}{\omega}$.
	Such a run pops $\rea{\overline\alpha}{\omega}$ and $k$ copies of $\rea{\overline\beta}{\omega}$, for some $k\in\Nat$,
	then it pops the whole $(\rea{\overline\beta}{\omega})^\omega$ while staying in some state $r$, and then pops $\gamma$ going to some state $s$.
	Then $r\in\eats{\rea{\overline\beta}{\omega}}_\epsilon(r)\subseteq\eats{\rea{\overline\beta}{\omega}}_\epsilon(Q)$ and $s\in\eats{\gamma}(r)$.
	By well-formedness of $\overline\delta$ we then have $r\in\eats{\rea{\overline\alpha}{\omega}(\rea{\overline\beta}{\omega})^k}(P)=\eats{\rea{\overline\alpha}{\omega}}(P)$
	and $\eats{\gamma}_\epsilon(r)\neq\emptyset$, which implies $s\in\eats{\gamma}_\epsilon(r)$ by \cref{epsilon-det}.
	In consequence $s\in \eats{\rea{\overline\beta}{\omega}\gamma}_\epsilon(\eats{\rea{\overline\alpha}{\omega}}(P))$, as needed.
	
	Conversely, consider some $s\in\eats{\rea{\overline\beta}{\omega}\gamma}_\epsilon(\eats{\rea{\overline\alpha}{\omega}}(P))$.
	Then there is a state $r\in\eats{\rea{\overline\beta}{\omega}}_\epsilon(\eats{\rea{\overline\alpha}{\omega}}(P))\subseteq\eats{\rea{\overline\beta}{\omega}}_\epsilon(Q)$
	for which $s\in\eats{\gamma}_\epsilon(r)$.
	By well-formedness of $\overline\delta$ also $r\in\eats{\rea{\overline\beta}{\omega}}_\epsilon(r)$.
	We can thus pop $\rea{\overline\alpha}{\omega}$ and one copy of $\rea{\overline\beta}{\omega}$ while going to state $r$,
	then pop the whole $(\rea{\overline\beta}{\omega})^\omega$ remaining in state $r$, and then pop $\gamma$ going to state $s$;
	we have $s\in\eats{\rea{\overline\delta}{\omega}}(P)$, as needed.	
\end{proof}

\begin{lemma}\label{epsilon-preserved}
	For every well-formed decomposition $\overline\delta$, every set $P\subseteq Q$, and every $e\in\Nat$ we have
	\begin{compactenum}
	\item	$\eats{\rea{\overline\delta}{\omega}}_\epsilon(P)=\eats{\rea{\overline\delta}{1}}_\epsilon(P)$,
	\item	$\eats{\rea{\overline\delta}{\omega}}(P)\subseteq\eats{\rea{\overline\delta}{e}}(P)$, and
	\item	$\up(\rea{\overline\delta}{e})=\up(\rea{\overline\delta}{1})$ and $\up_0(\rea{\overline\delta}{e})=\up_0(\rea{\overline\delta}{1})$.
	\end{compactenum}
\end{lemma}

\begin{proof}
	Induction on the structure of $\overline\delta$.
	The base case of $\overline\delta=X\in\Gamma$ is trivial, because $\rea{\overline\delta}{\omega}=\rea{\overline\delta}{e}=\rea{\overline\delta}{1}=X$.
	In the case of $\overline\delta=\overline\delta_1\dots\overline\beta_k$ the thesis follows directly from the induction hypothesis,
	because we can write
	$\eats{\rea{\overline\delta}{\omega}}(P)=\eats{\rea{\overline\delta_k}{\omega}}(\eats{\rea{\overline\delta_{k-1}}{\omega}}(\dots(\eats{\rea{\overline\delta_1}{\omega}}(P))\dots))$,
	and likewise for the other functions and relations.
	For $\overline\delta=\decomp{\overline\alpha,\overline\beta,\gamma}$ we have, using \cref{easy-eats-decomp}, the induction hypothesis, and well-formedness of $\overline\delta$:
	\begin{align*}
		\eats{\rea{\overline\delta}{\omega}}_\epsilon(P)
		&=\eats{\gamma}_\epsilon(\eats{\rea{\overline\beta}{\omega}}_\epsilon(\eats{\rea{\overline\alpha}{\omega}}_\epsilon(P)))
		=\eats{\gamma}_\epsilon(\eats{\rea{\overline\beta}{1}}_\epsilon(\eats{\rea{\overline\alpha}{1}}_\epsilon(P)))
		=\eats{\rea{\overline\delta}{1}}_\epsilon(P)\,,\displaybreak[0]\\
		\eats{\rea{\overline\delta}{\omega}}(P)
		&=\eats{\gamma}_\epsilon(\eats{\rea{\overline\beta}{\omega}}_\epsilon(\eats{\rea{\overline\alpha}{\omega}}(P)))
		\subseteq\eats{\gamma}(\eats{\rea{\overline\beta}{\omega}}(\eats{\rea{\overline\alpha}{\omega}}(P)))\\
		&=\eats{\gamma}(\underbrace{\eats{\rea{\overline\beta}{\omega}}(\dots(\eats{\rea{\overline\beta}{\omega}}}_e(\eats{\rea{\overline\alpha}{\omega}}(P)))\dots))\\
		&\subseteq\eats{\gamma}(\underbrace{\eats{\rea{\overline\beta}{e}}(\dots(\eats{\rea{\overline\beta}{e}}}_e(\eats{\rea{\overline\alpha}{e}}(P)))\dots))
		=\eats{\rea{\overline\delta}{e}}(P)\,,\displaybreak[0]\\
		\up(\rea{\overline\delta}{e})
		&=\up(\gamma)\circ\underbrace{\up(\rea{\overline\beta}{e})\circ\dots\circ\up(\rea{\overline\beta}{e})}_e{}\circ\up(\rea{\overline\alpha}{e})\\
		&=\up(\gamma)\circ\underbrace{\up(\rea{\overline\beta}{1})\circ\dots\circ\up(\rea{\overline\beta}{1})}_e{}\circ\up(\rea{\overline\alpha}{1})
		=\up(\gamma)\circ\up(\rea{\overline\beta}{1})\circ\up(\rea{\overline\alpha}{1})=\up(\rea{\overline\delta}{1})\,,
	\end{align*}
	and likewise for $\up_0(\cdot)$.
\end{proof}

The next \lcnamecref{decomp-exists} is very important.
It allows us to find a decomposition of bounded degree and height for an arbitrarily large stack content $\delta$.
This means that arbitrarily large segments of $\delta$ have to fit into the $\gamma$ parts of pumping triples in the decomposition.

\begin{lemma}\label{decomp-exists}
	Every nonempty stack content $\delta\in\Gamma^*$ has a well-formed decomposition $\overline\delta$ of height at most $2|Q|+2$
	and degree at most $\Cc$, for some $\Cc\in 2^{|\Pp|^{O(1)}}$.
\end{lemma}

\begin{proof}
	Take $\Cc=2^{|Q|^2(|\Gamma|^2+1)+1}+1$, and denote $h(\alpha)=|Q\setminus\eats{\alpha}_\epsilon(Q)|$ for $\alpha\in\Gamma^*$.
	We strengthen the thesis of the \lcnamecref{decomp-exists}, showing that the height of the resulting decomposition is at most $2h(\delta)+2$ (which is not greater than $2|Q|+2$).
	The proof is by induction on $h(\delta)$.
	
	We represent the stack content $\delta$ as $\delta=\delta_1 X_1 \delta_2 X_2 \dots \delta_n X_n \delta_{n+1}$ (where $\delta_i\in\Gamma^*$ and $X_i\in\Gamma$)
	in such a way that either $h(\delta_i)<h(\delta)$ or $\delta_i=\epsilon$ for all $i\in\interval{1,n+1}$,
	and $h(\delta_i X_i)\geq h(\delta)$ for all $i\in\interval{1,n}$.
	To obtain such a representation, as $\delta_1$ we take the longest prefix of $\delta$ satisfying $h(\delta_1)<h(\delta)$, or just $\delta_1=\epsilon$ if there is no such prefix;
	as $X_1$ we take the next letter for which by maximality of $\delta_1$ we have $h(\delta_1 X_1)\geq h(\delta)$;
	then as $\delta_2$ we take the longest prefix of the remaining part of $\delta$ satisfying..., and so on.
	
	By the induction hypothesis, for every nonempty fragment $\delta_i$ we have a well-formed decomposition $\overline\delta_i$ of height at most $2h(\delta_i)+2\leq 2h(\delta)$
	and degree at most $\Cc$.
	Note that when $h(\delta)=0$, then all $\delta_i$ are empty (we cannot have $h(\delta_i)<0$), hence the induction hypothesis is not needed in this base case.
	
	If $n\leq 2^{|Q|^2(|\Gamma|^2+1)}$, as $\overline\delta$ we just take the sequence consisting of the decompositions $\overline\delta_i$ for nonempty $\delta_i$,
	and of single-node decompositions for symbols $X_i$;
	this sequence has length at most $2n+1\leq\Cc$.
	Well-formedness of $\overline\delta$ follows immediately from well-formedness of the decompositions $\overline\delta_i$.
	
	Suppose that $n>2^{|Q|^2(|\Gamma|^2+1)}$.
	For $k\leq\ell\leq n$ let $\overline\delta_{k,\ell}$ be the well-formed decomposition of $\delta_k X_k\dots\delta_\ell X_\ell$,
	obtained as the sequence consisting of the decompositions $\overline\delta_i$ for nonempty $\delta_i$,
	and of single-node decompositions for symbols $X_i$, where $i\in\interval{k,\ell}$.
	To every $i\in\interval{1,n}$ we assign the tuple
	\begin{align*}
		S_i=\big((\eats{\rea{\overline\delta_{1,i}}{\omega}}(q))_{q\in Q},\up(\delta_{i+1}X_{i+1}\dots\delta_n X_n\delta_{n+1})\big)\,.
	\end{align*}
	The first $|Q|$ components are subsets of $Q$, and the last component is a subset of $(Q\times\Gamma)^2$, so there are at most $2^{|Q|^2}\cdot 2^{|Q|^2\cdot|\Gamma|^2}$
	possible tuples $S_i$.
	Thus, there are some $k,\ell$ with $k<\ell\leq 2^{|Q|^2(|\Gamma|^2+1)}+1\leq n$ for which $S_k=S_\ell$.
	Let
	\begin{align*}
		\alpha=\delta_1 X_1 \dots \delta_k X_k\,,&&
		\beta=\delta_{k+1} X_{k+1} \dots \delta_\ell X_\ell\,,&&\mbox{and}&&
		\gamma=\delta_{\ell+1} X_{\ell+1} \dots \delta_n X_n\delta_{n+1}\,.
	\end{align*}
	Then $\overline\alpha=\overline\delta_{1,k}$ and $\overline\beta=\overline\delta_{k+1,\ell}$ are well-formed decompositions for $\alpha$ and $\beta$, respectively.
	These decompositions are sequences of length at most $2k$ and $2(\ell-k)$, respectively, which is smaller than $\Cc$.
	Moreover, they both have height at most $2h(\delta)+1$.
	Then $\overline\delta=\decomp{\overline\alpha,\overline\beta,\gamma}$ is a decomposition of $\delta$, has degree at most $\Cc$, and height a most $2h(\delta)+2$,
	so it remains to see that $\overline\delta$ is well-formed.
	
	Observe that
	\begin{align*}
		S_k=\big((\eats{\rea{\overline\alpha}{\omega}}(q))_{q\in Q},\up(\rea{\overline\beta}{1}\gamma)\big) &&\mbox{and}&&
		S_\ell=\big((\eats{\rea{\overline\alpha}{\omega}\rea{\overline\beta}{\omega}}(q))_{q\in Q},\up(\gamma)\big)\,.
	\end{align*}
	Equality of the last components of $S_k$ and $S_\ell$ gives us the first condition of well-formedness, namely $\up(\gamma)=\up(\rea{\overline\beta}{1}\gamma)$.
	Equalities $\eats{\rea{\overline\alpha}{\omega}}(q)=\eats{\rea{\overline\alpha}{\omega}\rea{\overline\beta}{\omega}}(q)$ for all $q\in Q$ give us
	$\eats{\rea{\overline\alpha}{\omega}}(P)
	=\bigcup_{q\in P}\eats{\rea{\overline\alpha}{\omega}}(q)
	=\bigcup_{q\in P}\eats{\rea{\overline\alpha}{\omega}\rea{\overline\beta}{\omega}}(q)
	=\eats{\rea{\overline\alpha}{\omega}\rea{\overline\beta}{\omega}}(P)$ for every set $P\subseteq Q$, which is the second condition of well-formedness.
	
	Recall that $h(\delta_\ell X_\ell)\geq h(\delta)$;
	equivalently $|\eats{\delta_\ell X_\ell}_\epsilon(Q)|=|Q|-h(\delta_\ell X_\ell)\leq |Q|-h(\delta)=|\eats{\delta}_\epsilon(Q)|$.
	Because $\eats{\alpha\beta}_\epsilon(Q)=\eats{\beta}_\epsilon(\eats{\alpha}_\epsilon(Q))\subseteq\eats{\beta}_\epsilon(Q)=
	\eats{\delta_\ell X_\ell}_\epsilon(\eats{\delta_k X_k\dots\delta_{\ell-1} X_{\ell-1}}_\epsilon(Q))\subseteq\eats{\delta_\ell X_\ell}_\epsilon(Q)$,
	this implies $|\eats{\alpha\beta}_\epsilon(Q)|\leq|\eats{\beta}_\epsilon(Q)|\leq|\eats{\delta}_\epsilon(Q)|$.
	Consider now the relation $\Rr=\set{(r,s)\in\eats{\alpha\beta}_\epsilon(Q)\times\eats{\delta}_\epsilon(Q)\mid s\in\eats{\gamma}_\epsilon(r)}$.
	By \cref{epsilon-det} this is a partial function (for every $r$ the set $\eats{\gamma}_\epsilon(r)$ contains at most one $s$).
	On the other hand, for every state $s\in\eats{\delta}_\epsilon(Q)=\eats{\gamma}_\epsilon(\eats{\alpha\beta}_\epsilon(Q))$ there exists $r\in\eats{\alpha\beta}_\epsilon(Q)$
	such that $s\in\eats{\gamma}_\epsilon(r)$.
	Together with the cardinality argument $|\eats{\alpha\beta}_\epsilon(Q)|\leq|\eats{\delta}_\epsilon(Q)|$
	this implies that $\Rr$ is a bijection.
	Thus, the three sets, $\eats{\alpha\beta}_\epsilon(Q)$, $\eats{\beta}_\epsilon(Q)$, and $\eats{\delta}_\epsilon(Q)$, are of the same size.
	Recalling that $\eats{\alpha\beta}_\epsilon(Q)$ is a subset of $\eats{\beta}_\epsilon(Q)$, this implies that actually $\eats{\alpha\beta}_\epsilon(Q)=\eats{\beta}_\epsilon(Q)$.
	Taking into account \cref{epsilon-preserved} (which can be used because $\overline{\alpha}$ and $\overline\beta$ are both well-formed), we moreover have
	$\eats{\rea{\overline\alpha}{\omega}\rea{\overline\beta}{\omega}}_\epsilon(Q)=\eats{\alpha\beta}_\epsilon(Q)=\eats{\beta}_\epsilon(Q)=\eats{\rea{\overline\beta}{\omega}}_\epsilon(Q)$.
	
	In order to show the last condition of well-formedness, consider now a state $r\in\eats{\rea{\overline\beta}{\omega}}_\epsilon(Q)$.
	We need to show that $\eats{\rea{\overline\beta}{\omega}}_\epsilon(r)=\set{r}$ and $\eats{\gamma}_\epsilon(r)\neq\emptyset$.
	The latter follows immediately from the fact that $\Rr$ is a bijection: $r\in\eats{\alpha\beta}_\epsilon(Q)$, so $r$ has a corresponding element in $\eats{\gamma}_\epsilon(r)$.
	In order to show the former, note that, due to $r\in\eats{\rea{\overline\alpha}{\omega}\rea{\overline\beta}{\omega}}_\epsilon(Q)$,
	there have to exist states $q\in Q$ and $r'\in\eats{\rea{\overline\alpha}{\omega}}_\epsilon(q)$ such that $r\in\eats{\rea{\overline\beta}{\omega}}_\epsilon(r')$.
	By \cref{epsilon-det} we then have $\eats{\rea{\overline\alpha}{\omega}}(q)=\eats{\rea{\overline\alpha}{\omega}}_\epsilon(q)=\set{r'}$,
	and we have already shown that $\eats{\rea{\overline\alpha}{\omega}}(q)=\eats{\rea{\overline\alpha}{\omega}\rea{\overline\beta}{\omega}}(q)$
	(cf.~the second condition of well-formedness).
	Thus $r\in\eats{\rea{\overline\alpha}{\omega}\rea{\overline\beta}{\omega}}_\epsilon(q)\subseteq\eats{\rea{\overline\alpha}{\omega}\rea{\overline\beta}{\omega}}(q)=\set{r'}$ implies $r=r'$,
	so we have $r\in\eats{\rea{\overline\beta}{\omega}}_\epsilon(r)$, as required (by \cref{epsilon-det} the set $\eats{\rea{\overline\beta}{\omega}}_\epsilon(r)$ is a singleton).
\end{proof}

An interesting property of decompositions is described by \cref{eqlev-grows-while-pumping};
it says that if we increase the number of copies of $\rea{\overline\beta}{\omega}$ in a configuration,
then the configuration becomes more similar to a configuration with $\omega$ copies of $\rea{\overline\beta}{\omega}$.
We assume here that $1+\omega=\omega$.

\begin{lemma}\label{eqlev-grows-while-pumping}
	Let $q\in Q$, let $\decomp{\overline\alpha,\overline\beta,\gamma}$ be a well-formed pumping triple, let $\eta\in\Gamma^\gs$,
	and let $M_i=\min\set{\eqlev(r(\rea{\overline\beta}{\omega})^i\gamma\eta,r(\rea{\overline\beta}{\omega})^\omega\gamma\eta)\mid r\in\eats{\rea{\overline\alpha}{\omega}}(q)}$
	for all $i\in\Nat$.
	Then $M_{i+1}\geq 1+M_i$ for all $i\in\Nat$.
\end{lemma}

\begin{proof}
	Fix some $i\in\Nat$.
	Let $r_{\min}\in\eats{\rea{\overline\alpha}{\omega}}(q)$ be a state for which
	$M_{i+1}=\eqlev(r_{\min}(\rea{\overline\beta}{\omega})^{i+1}\gamma\eta,\allowbreak r_{\min}(\rea{\overline\beta}{\omega})^\omega\gamma\eta)$.
	We distinguish two cases.

	Suppose first that $\eats{\rea{\overline\beta}{\omega}}_\epsilon(r_{\min})=\emptyset$.
	As observed on \cpageref{expand-equiv}, we have
	$r_{\min}\rea{\overline\beta}{\omega}(\rea{\overline\beta}{\omega})^\omega\gamma\eta\approx r_{\min}(\rea{\overline\beta}{\omega})^\omega\gamma\eta$.
	Thus by \cref{more-equiv-if-read} we have
	\begin{align*}
		M_{i+1}&=\eqlev(r_{\min}\rea{\overline\beta}{\omega}(\rea{\overline\beta}{\omega})^i\gamma\eta,r_{\min}\rea{\overline\beta}{\omega}(\rea{\overline\beta}{\omega})^\omega\gamma\eta)\\
		&\geq 1+\min\set{\eqlev(r(\rea{\overline\beta}{\omega})^i\gamma\eta,r(\rea{\overline\beta}{\omega})^\omega\gamma\eta)\mid r\in\eats{\rea{\overline\beta}{\omega}}(r_{\min})}\,.
	\end{align*}
	By well-formedness of the pumping triple we have
	$\eats{\rea{\overline\beta}{\omega}}(r_{\min})\subseteq\eats{\rea{\overline\beta}{\omega}}(\eats{\rea{\overline\alpha}{\omega}}(q))
	=\eats{\rea{\overline\alpha}{\omega}\rea{\overline\beta}{\omega}}(q)
	=\eats{\rea{\overline\alpha}{\omega}}(q)$,
	so for all $r\in\eats{\rea{\overline\beta}{\omega}}(r_{\min})$ we have $\eqlev(r(\rea{\overline\beta}{\omega})^i\gamma\eta,r(\rea{\overline\beta}{\omega})^\omega\gamma\eta)\geq M_i$,
	which implies $M_{i+1}\geq 1+M_i$.

	Conversely, suppose that there is some $r\in\eats{\rea{\overline\beta}{\omega}}_\epsilon(r_{\min})\subseteq\eats{\rea{\overline\beta}{\omega}}_\epsilon(Q)$.
	By well-formedness of the pumping triple
	we have $\eats{\rea{\overline\beta}{\omega}}_\epsilon(r)=\set{r}$ and $s\in\eats{\gamma}_\epsilon(r)$ for some $s$.
	This implies that $r_{\min}(\rea{\overline\beta}{\omega})^{i+1}\gamma\eta\to^*_\epsilon r(\rea{\overline\beta}{\omega})^i\gamma\eta\to_\epsilon^*s\eta$
	and $r_{\min}(\rea{\overline\beta}{\omega})^\omega\gamma\eta\to_\epsilon^* r(\rea{\overline\beta}{\omega})^\omega\gamma\eta\to_\epsilon^*s\eta$, that is,
	$M_{i+1}=\omega\geq 1+M_i$.
\end{proof}

The next lemma says that if we can pop a stack content of the form $\rea{\overline\delta}{\omega}$, then we can pop it quickly:

\begin{lemma}\label{can pop fast decomp}
	If $s\in\eats{\rea{\overline\delta}{\omega}}(q)$ for a well-formed decomposition $\overline\delta$ of degree at most $\Cc$ and height $\ell$,
	then $\dist(q\rea{\overline\delta}{\omega},s)\leq\Bc\Cc^\ell$.
\end{lemma}

\begin{proofnoqed}
	Induction on the structure of $\overline\delta$.
	Suppose that $s\in\eats{\rea{\overline\delta}{\omega}}(q)$.
	If $\overline\delta=X=\rea{\overline\delta}{\omega}$, then $s\in\eats{X}(q)$ (i.e., $qX\to^*s$) implies $\dist(q\rea{\overline\delta}{\omega},s)\leq\Bc$ by \cref{can pop fast}.
	
	Suppose that $\overline\delta=\overline\delta_1\dots\overline\delta_k$.
	Then $\rea{\overline\delta}{\omega}=\rea{\overline\delta_1}{\omega}\dots\rea{\overline\delta_k}{\omega}$,
	so there have to exist states $r_0,r_1,\dots,r_k$ such that $r_0=q$, $r_k=r$, and $r_i\in\eats{\rea{\overline\delta_i}{\omega}}(r_{i-1})$ for all $i\in\interval{1,k}$.
	By the induction hypothesis
	$\dist(r_{i-1}\rea{\overline\delta_i}{\omega}\rea{\overline\delta_{i+1}}{\omega}\dots\rea{\overline\delta_k}{\omega},\allowbreak
	r_i\rea{\overline\delta_{i+1}}{\omega}\dots\rea{\overline\delta_k}{\omega})
	\leq\dist(r_{i-1}\rea{\overline\delta_i}{\omega},r_i)\leq\Bc\Cc^{\ell-1}$.
	Because $k\leq\Cc$, we obtain $\dist(r_0\rea{\overline\delta}{\omega},r_k)\leq\Cc\cdot\Bc\Cc^{\ell-1}$.

	Finally, suppose that $\overline\delta=\decomp{\overline\alpha,\overline\beta,\gamma}$.
	By \cref{easy-eats-decomp}, $s\in\eats{\rea{\overline\delta}{\omega}}(q)$ implies that there is a state $r\in\eats{\rea{\overline\alpha}{\omega}}(q)$ for which
	$s\in\eats{\rea{\overline\beta}{\omega}\gamma}_\epsilon(r)=\eats{(\rea{\overline\beta}{\omega})^\omega\gamma}_\epsilon(r)$.
	The latter means that $\dist(r(\rea{\overline\beta}{\omega})^\omega\gamma,s)=0$.
	We conclude using the induction hypothesis:
	\begin{align*}
		\dist(q\rea{\overline\delta}{\omega},s)
		\leq\dist(q\rea{\overline\delta}{\omega},r(\rea{\overline\beta}{\omega})^\omega\gamma)+\dist(r(\rea{\overline\beta}{\omega})^\omega\gamma,s)
		\leq\dist(q\rea{\overline\alpha}{\omega},r)+0\leq\Bc\Cc^{\ell-1}\leq\Bc\Cc^\ell\,.
	\qedhere\end{align*}
\end{proofnoqed}

We now define by induction when two decompositions $\overline\delta,\hat\delta$ have the same type, written $\overline\delta\tpeq\hat\delta$:
\begin{compactitem}
\item	if $\overline\delta=\hat\delta$ then $\overline\delta\tpeq\hat\delta$,
\item	if $\overline\delta=\overline\delta_1\dots\overline\delta_k$, and $\hat\delta=\hat\delta_1\dots\hat\delta_k$,
	and $\overline\delta_i\tpeq\hat\delta_i$ for all $i\in\set{1,\dots,k}$, then $\overline\delta\tpeq\hat\delta$, and
\item	if $\overline\delta=\decomp{\overline\alpha,\overline\beta,\gamma}$, and $\hat\delta=\decomp{\hat\alpha,\hat\beta,\gamma'}$, and
	$\overline\alpha\tpeq\hat\alpha$, and $\overline\beta\tpeq\hat\beta$, and $\eats{\gamma}_\epsilon=\eats{\gamma'}_\epsilon$, and
	$\up(\gamma)=\up(\gamma')$,
	then $\overline\delta\tpeq\hat\delta$.
\end{compactitem}
The next two lemmata say that, on the one hand, the number of different types is small, and, on the other hand, that while considering stacks of the form $\rea{\overline\delta}{\omega}$,
only the type of $\overline\delta$ is relevant:

\begin{lemma}\label{num types}
	The number of different types (i.e., equivalence classes of the $\tpeq$ relation) for decompositions of height at most $\ell$ and degree at most $\Cc$ is at most
	$2^{|\Pp|^4\cdot(\Cc+1)^{\ell+1}}$.
\end{lemma}

\begin{proofnoqed}
	Induction on $\ell$.
	For $\ell=0$ we have only $|\Gamma|$ decompositions, consisting of single stack symbols, and $|\Gamma|\leq 2^{|\Pp|^4\cdot(\Cc+1)}$.
	For $\ell\geq 1$ we have
	\begin{compactitem}
	\item	$|\Gamma|$ decompositions consisting of single stack symbols;
	\item	decompositions being sequences of at most $\Cc$ decompositions of height at most $\ell-1$;
		for each of them we have to know its type, where by 
		the induction hypothesis we have at most $2^{|\Pp|^4\cdot(\Cc+1)^\ell}$ possibilities;
	\item	pumping triples $\decomp{\overline\alpha,\overline\beta,\gamma}$, where we have to know the types of $\overline\alpha$ and $\overline\beta$
		(at most $2^{|\Pp|^4\cdot(\Cc+1)^\ell}$ possibilities for each, by 
		the induction hypothesis), and the objects $\eats{\gamma}_\epsilon$ and $\up(\gamma)$.
	\end{compactitem}
	In the last case, we notice that $\eats{\gamma}_\epsilon(P)=\bigcup_{p\in P}\eats{\gamma}_\epsilon(p)$ for every $P\subseteq Q$,
	meaning that the function $\eats{\gamma}_\epsilon$ is determined by the set $\set{(p,q)\mid q\in\eats{\gamma}_\epsilon(p)}\subseteq Q^2$.
	Likewise, $\up(\gamma)$ is a subset of $(Q\Gamma)^2$.
	We thus have at most $2^{|Q|^2+|Q|^2|\Gamma|^2}\leq 2^{|\Pp|^4}$ 
	possibilities for combinations of $\eats{\gamma}_\epsilon$
	and $\up(\gamma)$.
	Using inequalities $\Cc\geq 3$ and $\Cc+3\allowbreak\leq 2^\Cc\allowbreak\leq2^{|\Pp|^4\cdot(\C+1)^\ell}$,
	we obtain that the number of possible types is at most
	\begin{align*}
		|\Gamma|+\sum\nolimits_{i=0}^{\Cc}2^{|\Pp|^4\cdot(\Cc+1)^\ell\cdot i}+2^{|\Pp|^4\cdot(\Cc+1)^\ell\cdot 2+|\Pp|}
		\leq(\Cc+3)\cdot 2^{|\Pp|^4\cdot(\Cc+1)^\ell\cdot\Cc}
		\leq 2^{|\Pp|^4\cdot(\Cc+1)^{\ell+1}}\,.
	\qedhere\end{align*}	
\end{proofnoqed}

\begin{lemma}\label{only type important}
	Let $\overline\delta,\hat\delta$ be decompositions.
	If $\overline\delta$ is well-formed and $\overline\delta\tpeq\hat\delta$, then
	\begin{compactenum}
	\item	$\eats{\rea{\overline\delta}{\omega}}=\eats{\rea{\hat\delta}{\omega}}$ and $\eats{\rea{\overline\delta}{\omega}}_\epsilon=\eats{\rea{\hat\delta}{\omega}}_\epsilon$,
	\item	$\up(\rea{\overline\delta}{1})=\up(\rea{\hat\delta}{1})$ and $\up_0(\rea{\overline\delta}{1})=\up_0(\rea{\hat\delta}{1})$,
	\item	$\hat\delta$ is well-formed as well, and
	\item	$q\rea{\overline\delta}{\omega}\eta\approx q\rea{\hat\delta}{\omega}\eta$ for all $q\in Q$ and $\eta\in\Gamma^\gs$.
	\end{compactenum}
\end{lemma}

\begin{proof}
	We proceed by induction on the structure of $\overline\delta$.
	If $\overline\delta=X$, then also $\hat\delta=X$, so the thesis is trivial.
	
	Suppose that $\overline\delta=\overline\delta_1\dots\overline\delta_k$.
	Then $\hat\delta=\hat\delta_1\dots\hat\delta_k$ with $\overline\delta_i\tpeq\hat\delta_i$ for all $i\in\interval{1,k}$.
	Item 3 follows directly from the induction hypothesis.
	Items 1 and 2 as well:
	$\eats{\rea{\overline\delta}{\omega}}(P)
	=\eats{\rea{\overline\delta_k}{\omega}}(\eats{\rea{\overline\delta_{k-1}}{\omega}}(\dots(\eats{\rea{\overline\delta_1}{\omega}}(P))\dots))
	=\eats{\rea{\hat\delta_k}{\omega}}(\eats{\rea{\hat\delta_{k-1}}{\omega}}(\dots(\eats{\rea{\hat\delta_1}{\omega}}(P))\dots))
	=\eats{\rea{\hat\delta}{\omega}}(P)$ for every $P\subseteq Q$, and likewise for the other equalities.
	By the induction hypothesis we also have
	$r\rea{\overline\delta_i}{\omega}\rea{\overline\delta_{i+1}}{\omega}\dots\rea{\overline\delta_k}{\omega}\eta
	\approx r\rea{\hat\delta_i}{\omega}\rea{\overline\delta_{i+1}}{\omega}\dots\rea{\overline\delta_k}{\omega}\eta$ for all $r\in Q$ and $i\in\interval{1,k}$,
	which implies $q\rea{\hat\delta_1}{\omega}\dots\rea{\hat\delta_{i-1}}{\omega}\rea{\overline\delta_i}{\omega}\rea{\overline\delta_{i+1}}{\omega}\dots\rea{\overline\delta_k}{\omega}\eta
	\approx q\rea{\hat\delta_1}{\omega}\dots\rea{\hat\delta_{i-1}}{\omega}\rea{\hat\delta_i}{\omega}\rea{\overline\delta_{i+1}}{\omega}\dots\rea{\overline\delta_k}{\omega}\eta$
	by \cref{equiv-because-equiv-below}.
	Having this for all $i\in\interval{1,k}$ we obtain $q\rea{\overline\delta}{\omega}\eta\approx q\rea{\hat\delta}{\omega}\eta$, as needed for Item 4.

	Finally, suppose that $\overline\delta=\decomp{\overline\alpha,\overline\beta,\gamma}$.
	Then $\hat\delta=\decomp{\hat\alpha,\hat\beta,\gamma'}$ with $\overline\alpha\tpeq\hat\alpha$, $\overline\beta\tpeq\hat\beta$, $\eats{\gamma}_\epsilon=\eats{\gamma'}_\epsilon$,
	and $\up(\gamma)=\up(\gamma')$.
	We obtain Item 1 using \cref{easy-eats-decomp} and the induction hypothesis:
	\begin{align*}
		\eats{\rea{\overline\delta}{\omega}}(P)
		&=\eats{\gamma}_\epsilon(\eats{\rea{\overline\beta}{\omega}}_\epsilon(\eats{\rea{\overline\alpha}{\omega}}(P)))
		=\eats{\gamma'}_\epsilon(\eats{\rea{\hat\beta}{\omega}}_\epsilon(\eats{\rea{\hat\alpha}{\omega}}(P)))
		=\eats{\rea{\hat\delta}{\omega}}(P)\,,&&\mbox{and}\\
		\eats{\rea{\overline\delta}{\omega}}_\epsilon(P)
		&=\eats{\gamma}_\epsilon(\eats{\rea{\overline\beta}{\omega}}_\epsilon(\eats{\rea{\overline\alpha}{\omega}}_\epsilon(P)))
		=\eats{\gamma'}_\epsilon(\eats{\rea{\hat\beta}{\omega}}_\epsilon(\eats{\rea{\hat\alpha}{\omega}}_\epsilon(P)))
		=\eats{\rea{\hat\delta}{\omega}}_\epsilon(P)&&\mbox{for all $P\subseteq Q$.}
	\end{align*}	
	Item 2 follows directly from the induction hypothesis and the equality $\up(\gamma)=\up(\gamma')$.
	We observe that the definition of well-formedness for $\overline\delta$ uses only the objects $\up(\gamma)$, $\up(\rea{\overline\beta}{1})$,
	$\eats{\rea{\overline\alpha}{\omega}}$, $\eats{\rea{\overline\beta}{\omega}}$, $\eats{\rea{\overline\beta}{\omega}}_\epsilon$, and $\eats{\gamma}_\epsilon$,
	and for each of them we already have equality with analogous object for the other decomposition; Item 3 follows.

	In order to obtain Item 4, we prove that
	\begin{align*}
		r(\rea{\overline\beta}{\omega})^\omega\gamma\eta\approx_k r(\rea{\hat\beta}{\omega})^\omega\gamma'\eta && \mbox{for all $r\in Q$ and $k\in\Nat$.}
	\end{align*}
	We proceed by induction on $k$.
	For $k=0$ this is trivial, as all configurations are related by $\approx_0$.
	Suppose now that the thesis holds for some $k$, and we want to prove it for $k+1$.
	As observed on \cpageref{expand-equiv}, we have $r(\rea{\overline\beta}{\omega})^\omega\gamma\eta\approx r\rea{\overline\beta}{\omega}(\rea{\overline\beta}{\omega})^\omega\gamma\eta$
	and $r(\rea{\hat\beta}{\omega})^\omega\gamma'\eta\approx r\rea{\hat\beta}{\omega}(\rea{\hat\beta}{\omega})^\omega\gamma'\eta$,
	and by the induction hypothesis of the external induction we have
	$r\rea{\hat\beta}{\omega}(\rea{\hat\beta}{\omega})^\omega\gamma'\eta\approx r\rea{\overline\beta}{\omega}(\rea{\hat\beta}{\omega})^\omega\gamma'\eta$.
	We thus have to prove that
	\begin{align*}
		r\rea{\overline\beta}{\omega}(\rea{\overline\beta}{\omega})^\omega\gamma\eta\approx_{k+1} r\rea{\overline\beta}{\omega}(\rea{\hat\beta}{\omega})^\omega\gamma'\eta\,.
	\end{align*}
	If $\eats{\rea{\overline\beta}{\omega}}_\epsilon(r)=\emptyset$, this follows by \cref{more-equiv-if-read} from the induction hypothesis saying that
	$r'(\rea{\overline\beta}{\omega})^\omega\gamma\eta\approx_k r'(\rea{\hat\beta}{\omega})^\omega\gamma'\eta$ for all $r'\in\eats{\rea{\overline\beta}{\omega}}(r)\subseteq Q$.
	Suppose thus that $\eats{\rea{\overline\beta}{\omega}}_\epsilon(r)\neq\emptyset$,
	and take some state $r'\in\eats{\rea{\overline\beta}{\omega}}_\epsilon(r)\subseteq\eats{\rea{\overline\beta}{\omega}}_\epsilon(Q)$.
	By well-formedness of $\overline\delta$ we have $r'\in\eats{\rea{\overline\beta}{\omega}}_\epsilon(r')=\eats{\rea{\hat\beta}{\omega}}_\epsilon(r')$
	and $s\in\eats{\gamma}_\epsilon(r')=\eats{\gamma'}_\epsilon(r')$ for some $s$.
	In consequence
	$r\rea{\overline\beta}{\omega}(\rea{\overline\beta}{\omega})^\omega\gamma\eta
	\to_\epsilon^*r'(\rea{\overline\beta}{\omega})^\omega\gamma\eta
	\to_\epsilon^*r'\gamma\eta
	\to_\epsilon^*s\eta$ and
	$r\rea{\overline\beta}{\omega}(\rea{\hat\beta}{\omega})^\omega\gamma'\eta
	\to_\epsilon^*r'(\rea{\hat\beta}{\omega})^\omega\gamma'\eta
	\to_\epsilon^*r'\gamma'\eta
	\to_\epsilon^*s\eta$,
	so we even have
	$r\rea{\overline\beta}{\omega}(\rea{\overline\beta}{\omega})^\omega\gamma\eta\approx s\eta\approx\rea{\overline\beta}{\omega}(\rea{\hat\beta}{\omega})^\omega\gamma'\eta$.
	
	Because $\bigcap_{k\in\Nat}(\approx_k)=(\approx)$, from the above it follows that
	$r(\rea{\overline\beta}{\omega})^\omega\gamma\eta\approx r(\rea{\hat\beta}{\omega})^\omega\gamma'\eta$ for all $r\in Q$,
	so $q\rea{\hat\alpha}{\omega}(\rea{\overline\beta}{\omega})^\omega\gamma\eta\approx q\rea{\hat\alpha}{\omega}(\rea{\hat\beta}{\omega})^\omega\gamma'\eta$
	by \cref{equiv-because-equiv-below}.
	We conclude using the induction hypothesis, saying that
	$q\rea{\overline\delta}{\omega}\eta
	=q\rea{\overline\alpha}{\omega}(\rea{\overline\beta}{\omega})^\omega\gamma\eta
	\approx q\rea{\hat\alpha}{\omega}(\rea{\overline\beta}{\omega})^\omega\gamma\eta$.
\end{proof}

\section{Runs going down}\label{S Runs down}

In this section we concentrate on runs ``going down'', that is, runs oriented on popping the stack.
Such runs are called $\Dc$-almost-popping runs, and classes reachable by them are collected in $\Below(\cdot)$ sets.
We also prove basic properties of these notions.
As in previous sections, we assume some fixed \PDS $\Pp=(Q,\Gamma,\Act,\Delta)$.

Let $\Dc=(\Bc+3)\cdot\Bc\Cc^{2|Q|+1}+2$ (it turns out that this value of $\Dc$ is suitable for our proofs presented in the next section).
Note that $\Dc\in 2^{|\Pp|^{O(1)}}$.
We say that a run $\pi$ is \emph{$\Dc$-almost-popping} if there is no run $\pi'_\Dc$ that is parallel to a suffix of $\pi$ and satisfies $\sgrow(\pi'_\Dc)\geq\Dc$.
For a set $\Omega$ of classes, we define $\Below(\Omega)$ to be the set of classes $\classC$ for which
there is a $\Dc$-almost-popping run $\pi$ from a configuration in a class in $\Omega$ to a configuration in $\classC$.
When $\Omega$ is a set of configurations (instead of classes), we abbreviate $\Below(\set{\class{c}\mid c\in\Omega})$ into $\Below(\Omega)$,
and for a single configuration $c$ we abbreviate $\Below(\set{c})$ into $\Below(c)$.

Observe that if $\pi$ is $\Dc$-almost-popping, then every run parallel to $\pi$ is $\Dc$-almost-popping as well.
It follows that if from a configuration $c$ there exists a $\Dc$-almost-popping run to a configuration in a class $\classC$,
then such a run exists from every configuration $c'\in\class{c}$.
Observe also that every suffix of a $\Dc$-almost-popping run is $\Dc$-almost-popping as well.
The next four \lcnamecrefs{few-conf-without-pushing} bound the size of the $\Below(\cdot)$ sets in specific situations:

\begin{lemma}\label{few-conf-without-pushing}
	Let $\overline\delta$ be a well-formed decomposition of degree at most $\Cc$ and height $\ell$, and let $\eta\in\Gamma^\gs$.
	Then there are at most $\Cc^\ell\cdot|\Pp|^{\Dc+1}$ classes $\classC$ for which there is a $\Dc$-almost-popping run $\pi$ that
	starts in $p\rea{\overline\delta}{\omega}\eta$ for some $p\in Q$, leads to a configuration in $\classC$, and can be shifted to a run from $p\rea{\overline\delta}{\omega}$.
	In particular, for every $P\subseteq Q$,
	\begin{align*}
		|\Below(\set{p\rea{\overline\delta}{\omega}\eta\mid p\in P})|\leq|\Below(\set{r\eta\mid r\in\eats{\rea{\overline\delta}{\omega}}(P)})|+\Cc^\ell\cdot|\Pp|^{\Dc+1}\,.
	\end{align*}
\end{lemma}

\begin{proof}
	The second part of the lemma easily follows from the first part.
	Indeed, consider a class $\classC\in\Below(\set{p\rea{\overline\delta}{\omega}\eta\mid p\in P})$.
	By definition, there exists a $\Dc$-almost-popping run from $p\rea{\overline\delta}{\omega}\eta$ for some $p\in P\subseteq Q$ to a configuration in $\classC$.
	If $\pi$ can be shifted to a run from $p\rea{\overline\delta}{\omega}$,
	then $\classC$ is among the at most $\Cc^\ell\cdot|\Pp|^{\Dc+1}$ classes counted by the first part of the lemma.
	If not, then $\pi$ reaches a configuration with stack $\eta$, at some moment.
	For the first such configuration $r\eta$ we have $r\in\eats{\rea{\overline\delta}{\omega}}(P)$,
	so the second part of $\pi$, from $r\eta$ to a configuration in $\classC$, witnesses that $\classC\in\Below(\set{r\eta\mid r\in\eats{\rea{\overline\delta}{\omega}}(P)})$.
	
	We prove the first part by induction on the structure of $\overline\delta$.
	Suppose first that $\overline\delta=X\in\Gamma$, and consider a $\Dc$-almost-popping run $\pi$ that starts in $pX\eta$ for some $p\in Q$, and can be shifted to a run from $pX$.
	Necessarily $\sgrow(\pi)<\Dc$ (because $\pi$ itself is also parallel to a suffix of $\pi$).
	Thus $\pi$ leads to a configuration $q\mu\eta$ with $|\mu|=|X|+\sgrow(\pi)\leq\Dc$.
	The number of such configurations $q\mu\eta$, hence also the number of classes of these configurations, is at most
	$|Q|\cdot\sum_{i=0}^{\Dc}|\Gamma|^i\leq|Q|\cdot(|\Gamma|+1)^\Dc\leq|\Pp|^{\Dc+1}$ (recall that $\eta$ is fixed).
	
	Next, suppose that $\overline\delta=\overline\delta_1\dots\overline\delta_k$.
	For every $\Dc$-almost-popping run that starts in $p\rea{\overline\delta}{\omega}\eta$ and can be shifted to a run from $p\rea{\overline\delta}{\omega}$
	we can find the smallest $i\in\interval{1,k}$ such that $\pi$ can be shifted to a run from $p\rea{\overline\delta_1}{\omega}\dots\rea{\overline\delta_i}{\omega}$.
	By minimality of $i$, the run $\pi$ crosses a configuration of the form $p'\rea{\overline\delta_i}{\omega}\rea{\overline\delta_{i+1}}{\omega}\dots\rea{\overline\delta_k}{\omega}\eta$;
	then the suffix $\pi'$ of $\pi$ starting there
	is a $\Dc$-almost-popping run that starts in $p'\rea{\overline\delta_i}{\omega}\rea{\overline\delta_{i+1}}{\omega}\dots\rea{\overline\delta_k}{\omega}\eta$
	and can be shifted to a run from $p'\rea{\overline\delta_i}{\omega}$.
	By the induction hypothesis, the number of classes reached by such runs $\pi'$, for every fixed $i$, is at most $\Cc^{\ell-1}\cdot|\Pp|^{\Dc+1}$.
	There are $\Cc$ possible values for $i$, so in total we have at most $\Cc^\ell\cdot|\Pp|^{\Dc+1}$ classes.

	Finally, suppose that $\overline\delta=\decomp{\overline\alpha,\overline\beta,\gamma}$.
	For every $\Dc$-almost-popping run $\pi$ that starts in $p\rea{\overline\delta}{\omega}\eta$ and can be shifted to a run from $p\rea{\overline\delta}{\omega}$
	we have one of the following possibilities:
	\begin{compactitem}
	\item
		Maybe $\pi$ can be shifted to a run from $p\rea{\overline\alpha}{\omega}$.
		By the induction hypothesis, the number of classes reached by such runs $\pi$ is at most $\Cc^{\ell-1}\cdot|\Pp|^{\Dc+1}$.
	\item
		Maybe $\pi$ cannot be shifted to a run from $p\rea{\overline\alpha}{\omega}$
		(hence $\pi$ visits a configuration with stack $(\rea{\overline\beta}{\omega})^\omega\gamma\eta$),
		but $\pi$ never pops $(\rea{\overline\beta}{\omega})^\omega$ going to a configuration with stack $\gamma\eta$.
		Consider the last moment when $\pi$ visits a configuration of the form $r(\rea{\overline\beta}{\omega})^\omega\gamma\eta$;
		let $\pi'$ be the suffix of $\pi$ starting at this moment.
		Let also $\pi''$ be the run starting in $r\rea{\overline\beta}{\omega}(\rea{\overline\beta}{\omega})^\omega\gamma\eta$
		and performing exactly the same transitions as $\pi'$;
		after performing the first transition the two runs coincide.
		We have $r\rea{\overline\beta}{\omega}(\rea{\overline\beta}{\omega})^\omega\gamma\eta\approx r(\rea{\overline\beta}{\omega})^\omega\gamma\eta$,
		so $\pi'$ ans $\pi''$ are parallel;
		in particular $\pi''$ is $\Dc$-almost-popping.
		Moreover, $\pi''$ can be shifted to a run from $r\rea{\overline\beta}{\omega}$,
		because otherwise $\pi''$ and $\pi'$ would visit a configuration with stack $(\rea{\overline\beta}{\omega})^\omega\gamma\eta$ again,
		contradicting the maximality of the starting point.
		Thus, by the induction hypothesis, the number of classes reached by such runs $\pi''$ is at most $\Cc^{\ell-1}\cdot|\Pp|^{\Dc+1}$.
	\item
		Finally, maybe at some moment $\pi$ reaches $r(\rea{\overline\beta}{\omega})^\omega\gamma\eta$,
		then goes to $r\gamma\eta$, and then continues somehow.
		Because $(\rea{\overline\beta}{\omega})^\omega$ could be popped, we have $r\in\eats{\rea{\overline\beta}{\omega}}_\epsilon(r)$,
		hence by well-formedness of $\overline\delta$ there exists some $s\in\eats{\gamma}_\epsilon(r)$;
		in other words, $r\gamma\to_\epsilon^*s$ for some state $s$.	
		Recall that $\pi$ can be shifted to a run from $p\rea{\overline\delta}{\omega}$, hence its suffix after visiting $r\gamma\eta$ can be shifted to a run $\rho$ from $r\gamma$.
		Thus, by determinism of $\epsilon$-transitions, $\rho$ consists entirely of $\epsilon$-transitions.
		In consequence, the prefix $\pi'$ of $\pi$ ending in $r(\rea{\overline\beta}{\omega})^\omega\gamma\eta$
		leads to the same class as $\pi$;
		moreover, $\pi'$ is parallel to $\pi$, so it is $\Dc$-almost-popping as well.
		But $\pi'$ satisfies assumptions of the previous case;
		no new classes can be reached.
	\end{compactitem}
	Thus, the total number of classes under consideration is at most $2\Cc^{\ell-1}\cdot|\Pp|^{\Dc+1}\leq\Cc^\ell\cdot|\Pp|^{\Dc+1}$.
\end{proof}

\begin{lemma}\label{few-conf-without-pushing-seq}
	Let $\overline\delta_1,\dots,\overline\delta_k$ be well-formed decompositions of degree at most $\Cc$ and height $\ell$, and let $\eta\in\Gamma^\gs$.
	Then there are at most $k\cdot\Cc^\ell\cdot|\Pp|^{\Dc+1}$ classes $\classC$ for which there is a $\Dc$-almost-popping run $\pi$ that
	starts in $p\rea{\overline\delta_1}{\omega}\dots\rea{\overline\delta_k}{\omega}\eta$ for some $p\in Q$,
	leads to a configuration in $\classC$, and can be shifted to a run from $p\rea{\overline\delta_1}{\omega}\dots\rea{\overline\delta_k}{\omega}$.
\end{lemma}

\begin{proof}
	As in the proof of \cref{few-conf-without-pushing} (the case of $\overline\delta=\overline\delta_1\dots\overline\delta_k$)
	we observe that for every class $\classC$ under consideration there
	is a $\Dc$-almost-popping run that starts in $p'\rea{\overline\delta_i}{\omega}\rea{\overline\delta_{i+1}}{\omega}\dots\rea{\overline\delta_k}{\omega}\eta$,
	ends in a configuration in $\classC$,
	and can be shifted to a run from $p'\rea{\overline\delta_i}{\omega}$, for some $i\in\interval{1,k}$.
	By \cref{few-conf-without-pushing}, the number of such classes for every fixed $i$ is at most $\Cc^\ell\cdot|\Pp|^{\Dc+1}$.
	There are $k$ possible values for $i$, so in total we have at most $k\cdot\Cc^\ell\cdot|\Pp|^{\Dc+1}$ classes.
\end{proof}

\begin{lemma}\label{Lbelowbelow}
	Let $\pi$ be a $\Dc$-almost-popping run.
	Then there are at most $|\Pp|^{2\Bc\Dc}-1$ classes $\classC$ for which there exists a $\Dc$-almost-popping run $\pi'$ from $\pi(|\pi|)$ to a configuration in $\classC$,
	but the composition $\pi\circ\pi'$ is not $\Dc$-almost-popping.
\end{lemma}

\begin{proof}
	We prove that every class $\classC$ under consideration satisfies $\dist(\class{\pi(|\pi|)},\classC)\leq\Bc\cdot(2\Dc-1)$;
	by \cref{determ short run} this bounds the number of possible classes $\classC$ by $|\Pp|^{\Bc\cdot(2\Dc-1)}\leq|\Pp|^{2\Bc\Dc}-1$.
	
	Consider thus a $\Dc$-almost-popping run $\pi'$ that starts in $\pi(|\pi|)$, ends in a class $\classC$, and is such that $\pi\circ\pi'$ is not $\Dc$-almost-popping.
	By definition we then have a run $\rho_\Dc$ that is parallel to a suffix of $\pi\circ\pi'$ and satisfies $\sgrow(\rho_\Dc)\geq\Dc$.
	Because $\pi'$ is $\Dc$-almost-popping, $\rho_\Dc$ cannot be parallel to a suffix of $\pi'$.
	We can thus represent $\rho_\Dc$ as $\rho_1\circ\rho_\Dc'$, where $\rho_1$ is parallel to a suffix of $\pi$, and $\rho_\Dc'$ is parallel to $\pi'$.
	Let also $\rho_3$ be a suffix of $\rho_\Dc'$ for which $\sgrow(\rho_3)$ is maximal;
	let us write $\rho_\Dc'=\rho_2\circ\rho_3$.
	Denote $k_i=\sgrow(\rho_i)$ for $i\in\set{1,2,3}$.
	Because $\pi$ and $\pi'$ are $\Dc$-almost-popping, necessarily $k_1\leq D-1$ and $k_3\leq D-1$ (because $\rho_1$ and $\rho_3$ are parallel to suffixes of $\pi$ and $\pi'$, respectively).
	Recall also that $\sgrow(\rho_\Dc)=k_1+k_2+k_3\geq\Dc$, so $k_1+k_2\geq 1$,
	implying that on the top of the stack of $\rho_3(0)$ we have a real stack symbol, pushed there by $\rho_1\circ\rho_2$ (i.e., not an expression of the form $\gamma^\omega$).
	Maximality of $\sgrow(\rho_3)$ means that we can then write $\rho_2(0)=p\alpha\eta$, $\rho_3(0)=qX\eta$, and $\rho_3(|\rho_3|)=r\beta\eta$ (for $\alpha,\beta\in\Gamma^*$),
	where $|X|-|\alpha|=k_2$ and $|\beta|-|X|=k_3$,
	and $\rho_2\circ\rho_3$ can be shifted to a run from $p\alpha$ to $r\beta$ (intuitively: the ``stack height'' during $\rho_2\circ\rho_3$ is minimal at $\rho_3(0)$).
	Due to $k_3\leq\Dc-1$ we have $|\beta|=k_3+1\leq\Dc$, and due to $k_1\leq D-1$ and $k_1+k_2\geq 1$ we have $|\alpha|=1-k_2\leq k_1\leq\Dc-1$.
	Thus
	\begin{align*}
		\dist(\class{\pi(|\pi|)},\classC)
		\leq\dist(\rho_2(0),\rho_3(|\rho_3|))
		\leq\dist(p\alpha,r\beta)
		\leq(|\alpha|+|\beta|)\cdot\Bc
		\leq\Bc\cdot(2\Dc-1)\,,
	\end{align*}
	as we wanted;
	the third inequality is by \cref{can pop fast}.
\end{proof}

\begin{lemma}\label{Cbelowbelow}
	For every set $\Omega$ of classes,
	\begin{align*}
		|\Below(\Below(\Omega))|\leq|\Below(\Omega)|\cdot|\Pp|^{2\Bc\Dc}\,.
	\end{align*}
\end{lemma}

\begin{proof}
	For every class $\classC\in\Below(\Omega)$ we fix some $\Dc$-almost-popping run $\pi_\classC$ from a configuration in a class in $\Omega$ to a configuration in $\classC$.
	For every class $\classD\in\Below(\Below(\Omega))\setminus\Below(\Omega)$ we fix a class $\classC(\classD)\in\Below(\Omega)$ such that $\classD\in\Below(\classC(\classD))$,
	and we fix a $\Dc$-almost-popping run $\pi'_\classD$ from $\pi_{\classC(\classD)}(|\pi_{\classC(\classD)}|)$ to a configuration in $\classD$
	(recall that a $\Dc$-almost-popping run to a configuration in $\classD$ exists not only from some configuration in $\classC(\classD)$,
	but actually from every configuration in $\classC(\classD)$, in particular from $\pi_{\classC(\classD)}(|\pi_{\classC(\classD)}|)$).
	Because $\classD\not\in\Below(\Omega)$, the composition $\pi_{\classC(\classD)}\circ\pi'_\classD$ is not $\Dc$-almost-popping.
	Thus, for every class $\classC_0$, the number of classes $\classD\in\Below(\Below(\Omega))\setminus\Below(\Omega)$ with $\classC(\classD)=\classC_0$ is at most $|\Pp|^{2\Bc\Dc}-1$,
	by \cref{Lbelowbelow}.
	We then have at most $|\Below(\Omega)|$ choices for the class $\classC_0$.
	Additionally, $\Below(\Below(\Omega))$ may contain at most $|\Below(\Omega)|$ classes from $\Below(\Omega)$.
	We thus have $|\Below(\Below(\Omega))|\leq|\Below(\Omega)|\cdot(|\Pp|^{2\Bc\Dc}-1)+|\Below(\Omega)|$, as required.
\end{proof}

\section{The core of the upper bound proof}\label{sec:core}

Let us fix an \PDS $\Pp=(Q,\Gamma,\Act,\Delta)$ together with an initial configuration $\initstate\in Q\Gamma$.
Assuming that $(\Ll(\Pp),\initstate)$ is weakly bisimulation finite, we want to bound the number of its classes reachable from $\class{\initstate}$,
showing that it is in $2^{2^{|\Pp|^{O(1)}}}$.

Our first lemma says that if the \PDS is weakly bisimulation finite, then generalized configurations with $\omega$ repetitions of some stack fragment
are equivalent to configurations with a large enough number $e\in\Nat$ of repetitions:

\begin{lemma}\label{F is omega}
	Let $q\in Q$, let $\decomp{\overline\alpha,\overline\beta,\gamma}$ be a well-formed pumping triple, and let $\eta\in\Gamma^*$.
	If $\initstate\to^*q\rea{\overline\alpha}{1}\rea{\overline\beta}{1}\gamma\eta$,
	and if $(\Ll(\Pp),\initstate)$ is weakly bisimulation finite,
	then for some $e\in\Nat$ we have $r(\rea{\overline\beta}{\omega})^e\gamma\eta\approx r(\rea{\overline\beta}{\omega})^\omega\gamma\eta$ for all $r\in\eats{\rea{\overline\alpha}{\omega}}(q)$.
\end{lemma}

\begin{proof}
	We assume that $(\Ll(\Pp),\initstate)$ is weakly bisimulation finite, and we denote the number of its classes by $F$.
	We prove the lemma for $e=F^{|Q|}$.
	Let us formulate a version of the lemma suitable for induction:
	for every $q\in Q$, every well-formed decomposition $\overline\delta$, and every $\eta\in\Gamma^*$,
	if $\initstate\to^*q\rea{\overline\delta}{e}\eta$,
	then $q\rea{\overline\delta}{e}\eta\approx q\rea{\overline\delta}{\omega}\eta$.
	We prove this statement by induction on the structure of $\overline\delta$,
	and in the case of $\overline\delta=\decomp{\overline\alpha,\overline\beta,\gamma}$ we also prove that
	$r(\rea{\overline\beta}{\omega})^e\gamma\eta\approx r(\rea{\overline\beta}{\omega})^\omega\gamma\eta$ for all $r\in\eats{\rea{\overline\alpha}{\omega}}(q)$.
	This gives us the actual statement of the lemma; note that the assumption of the lemma, $\initstate\to^*q\rea{\overline\delta}{1}\eta$
	(which can be reformulated as $(\initstate,q)\in\up_0(\rea{\overline\delta}{1})$),
	is equivalent to the assumption of our inductive statement,
	$\initstate\to^*q\rea{\overline\delta}{e}\eta$ (which can be reformulated as $(\initstate,q)\in\up_0(\rea{\overline\delta}{e})$),
	by \cref{epsilon-preserved}.

	If $\overline\delta=X\in\Gamma$, the thesis is trivial because $\rea{\overline\delta}{\omega}=\rea{\overline\delta}{e}=X$.
	
	Suppose that $\overline\delta=\overline\delta_1\dots\overline\delta_k$, and consider some $i\in\interval{1,k}$.
	For every $r\in\eats{\rea{\overline\delta_1}{\omega}\dots\rea{\overline\delta_{i-1}}{\omega}}(q)$
	we have $r\in\eats{\rea{\overline\delta_1}{e}\dots\rea{\overline\delta_{i-1}}{e}}(q)$ by \cref{epsilon-preserved},
	so $\initstate\to^*q\rea{\overline\delta}{e}\eta\to^*r\rea{\overline\delta_i}{e}\rea{\overline\delta_{i+1}}{e}\dots\rea{\overline\delta_k}{e}\eta$;
	by the induction hypothesis we thus have
	$r\rea{\overline\delta_i}{e}\rea{\overline\delta_{i+1}}{e}\dots\rea{\overline\delta_k}{e}\eta
	\approx
	r\rea{\overline\delta_i}{\omega}\rea{\overline\delta_{i+1}}{e}\dots\rea{\overline\delta_k}{e}\eta$.
	By \cref{equiv-because-equiv-below} this implies
	$q\rea{\overline\delta_1}{\omega}\dots\rea{\overline\delta_{i-1}}{\omega}\rea{\overline\delta_i}{e}\rea{\overline\delta_{i+1}}{e}\dots\rea{\overline\delta_k}{e}\eta
	\approx
	q\rea{\overline\delta_1}{\omega}\dots\rea{\overline\delta_{i-1}}{\omega}\rea{\overline\delta_i}{\omega}\rea{\overline\delta_{i+1}}{e}\dots\rea{\overline\delta_k}{e}\eta$.
	Having this for all $i\in\interval{1,k}$, we obtain $q\rea{\overline\delta}{e}\eta\approx q\rea{\overline\delta}{\omega}\eta$, as needed.
	
	Finally, suppose that $\overline\delta=\decomp{\overline\alpha,\overline\beta,\gamma}$.
	Let $R=\eats{\rea{\overline\alpha}{\omega}}(q)$.
	By well-formedness of $\overline\delta$ we have $R=\eats{\rea{\overline\alpha}{\omega}(\rea{\overline\beta}{\omega})^i}(q)$ for all $i\geq 0$.
	By applying the previous case to the decomposition $\hat\delta=\overline\alpha\underbrace{\overline\beta\dots\overline\beta}_e$
	we obtain that 
	\begin{align*}
		q\rea{\overline\delta}{e}\eta
		=q\rea{\overline\alpha}{e}(\rea{\overline\beta}{e})^e\gamma\eta
		\approx q\rea{\overline\alpha}{\omega}(\rea{\overline\beta}{\omega})^e\gamma\eta\,.
	\end{align*}
	The former configuration is reachable from $\initstate$, so the class of $q\rea{\overline\alpha}{\omega}(\rea{\overline\beta}{\omega})^e\gamma\eta$ is reachable from $\class{\initstate}$.
	Moreover, $q\rea{\overline\alpha}{\omega}(\rea{\overline\beta}{\omega})^e\gamma\eta\to^* r(\rea{\overline\beta}{\omega})^i\gamma\eta$ for all $i\in\interval{0,e}$ and $r\in R$,
	so the classes of $r(\rea{\overline\beta}{\omega})^i\gamma\eta$ are reachable as well;
	they are among the $F$ classes of our pointed $\epsilon$-LTS.
	For every $i\in\interval{0,e}$ let
	\begin{align*}
		\tupleS_i=(\class{r(\rea{\overline\beta}{\omega})^i\gamma\eta})_{r\in R}
		&&\mbox{and}&&
		M_i=\min\set{\eqlev(r(\rea{\overline\beta}{\omega})^i\gamma\eta,r(\rea{\overline\beta}{\omega})^\omega\gamma\eta)\mid r\in R}\,.
	\end{align*}
	There exist at most $F^{|Q|}=e$ distinct tuples of $|R|$ classes,
	so necessarily $\tupleS_i=\tupleS_j$ for some $i,j$ with $0\leq i<j\leq e$, and thus also $M_i=M_j$.
	On the other hand $M_j\geq(j-i)+M_i$ by \cref{eqlev-grows-while-pumping} applied to indices $i,i+1,\dots,j-1$.
	Thus necessarily $M_i=\omega$, meaning that $r(\rea{\overline\beta}{\omega})^i\gamma\eta\approx r(\rea{\overline\beta}{\omega})^\omega\gamma\eta$ for all $r\in R$.
	By \cref{equiv-because-equiv-below} we also have
	$r(\rea{\overline\beta}{\omega})^e\gamma\eta
	\approx r(\rea{\overline\beta}{\omega})^{e-i}(\rea{\overline\beta}{\omega})^\omega\gamma\eta
	\approx r(\rea{\overline\beta}{\omega})^\omega\gamma\eta$ for all $r\in R$,
	and $q\rea{\overline\alpha}{\omega}(\rea{\overline\beta}{\omega})^e\gamma\eta
	\approx q\rea{\overline\alpha}{\omega}(\rea{\overline\beta}{\omega})^\omega\gamma\eta$.
	Altogether we obtain
	\begin{align*}
		q\rea{\overline\delta}{e}\eta
		\approx	q\rea{\overline\alpha}{\omega}(\rea{\overline\beta}{\omega})^e\gamma\eta
		\approx q\rea{\overline\alpha}{\omega}(\rea{\overline\beta}{\omega})^\omega\gamma\eta
		=q\rea{\overline\delta}{\omega}\eta,
	\end{align*}
	as needed.
\end{proof}

We now have \cref{main lemma}, our main technical lemma.
It says that the number of possible classes of $r\gamma\eta$ is small,
assuming that the classes $\classC_s$ of $s\eta$ are fixed,
and that $\gamma$ belongs to a pumping triple $\decomp{\cdot,\cdot,\gamma}$ of a fixed type.
This is very powerful, because infinitely many stack contents $\gamma$ may be handled this way (intuitively, $\gamma$ is almost arbitrary).
The assumption that $\gamma$ belongs to a pumping triple of a fixed type is very mild;
we easily deal with it later (using the fact that every stack content has a decomposition whose type comes from a small set).

\begin{lemma}\label{main lemma}
	Let $q\in Q$, let $\overline\delta_0=\decomp{\overline\alpha,\overline\beta,\gamma_0}$ be a well-formed decomposition of degree at most $\Cc$ and height at most $2|Q|+2$,
	let $(\classC_s)_{s\in\eats{\rea{\overline\delta_0}{\omega}}(q)}$ be a tuple of classes, and let $K=|\Below(\set{\classC_s\mid\allowbreak s\in\eats{\rea{\overline\delta_0}{\omega}}(q)})|$.
	Let also $\Omega$ be the set of pairs $(\gamma,\eta)\in\Gamma^*\times\Gamma^*$ for which
	\begin{compactitem}
	\item	$\decomp{\overline\alpha,\overline\beta,\gamma}\tpeq\overline\delta_0$,
	\item	$\initstate\to^*q\rea{\overline\alpha}{1}\rea{\overline\beta}{1}\gamma\eta$, and
	\item	$s\eta\in\classC_s$ for all $s\in\eats{\rea{\overline\delta_0}{\omega}}(q)$.
	\end{compactitem}
	If $(\Ll(\Pp),\initstate)$ is weakly bisimulation finite, then the set of tuples of classes
	\begin{align*}
		\Theta=\set{(\class{r\gamma\eta})_{r\in\eats{\rea{\overline\alpha}{\omega}}(q)}\mid(\gamma,\eta)\in\Omega}
	\end{align*}
	has at most $(K+\Tc)^{3|Q|}$ elements, for some $\Tc\in 2^{2^{|\Pp|^{O(1)}}}$.
	Moreover, for every tuple $(\classC'_r)_{r\in\eats{\rea{\overline\alpha}{\omega}}(q)}\in\Theta$
	we have $|\Below(\set{\classC'_r\mid r\in\eats{\rea{\overline\alpha}{\omega}}(q)})|\leq(K+1)\cdot\Uc$ for some $\Uc\in 2^{2^{|\Pp|^{O(1)}}}$.
\end{lemma}

\begin{proof}
	If $\Omega$ is empty, the thesis holds trivially.
	Assuming that $\Omega$ is nonempty,
	let us fix, for the duration of the whole proof, a stack content $\eta_0\in\Gamma^*$ 
	that occurs on the second coordinate of some pair from $\Omega$.
	Then $s\eta_0\in\classC_s$ for all $s\in\eats{\rea{\overline\delta_0}{\omega}}(q)$, 
	and thus the condition $s\eta\in\classC_s$ can be reformulated as $s\eta\approx s\eta_0$.
	We can also write $K=|\Below(\set{s\eta_0\mid\allowbreak s\in\eats{\rea{\overline\delta_0}{\omega}}(q)})|$.

	Denote $\widetilde\alpha=\rea{\overline\alpha}{\omega}$ and $\widetilde\beta=\rea{\overline\beta}{\omega}$.
	Note that $\widetilde\beta$ is obtained from $\overline\beta$ by adding $\omega$ exponents in all pumping triples inside $\overline\beta$,
	but the whole $\overline\beta$ remains repeated only once (in particular, if $\overline\beta$ is just a sequence of symbols, then $\widetilde\beta=\overline\beta$).
	Let also $R=\eats{\widetilde\alpha}(q)$.
	By well-formedness of $\overline\delta_0$ we have
	$R=\eats{\widetilde\alpha\widetilde\beta^i}(q)$ for all $i\geq 0$, as well as $\eats{\widetilde\beta}(R)=R$.
	
	For $(\gamma,\eta)\in \Omega$ let $e_{\gamma,\eta}$ be the smallest number such that
	$(\class{r\widetilde\beta^{e_{\gamma,\eta}}\gamma\eta})_{r\in R}=(\class{r\widetilde\beta^\omega\gamma\eta})_{r\in R}$.
	By \cref{F is omega} we have $e_{\gamma,\eta}\in\Nat$.
	We define also the constant $\Dc'=(\Bc+2)\cdot\Bc\Cc^{2|Q|+1}+1$ and remark that $\Dc'<\Dc$,
	recalling that $\Dc=(\Bc+3)\cdot\Bc\Cc^{2|Q|+1}+2$.
	
	Next, for all $r\in R$, all $(\gamma,\eta)\in \Omega$, and all $i\in\interval{0,e_{\gamma,\eta}-1}$ we fix some
	$j_{r,\gamma,\eta,i}\in\interval{i+1,e_{\gamma,\eta}}$, $\theta_{r,\gamma,\eta,i}\in\Gamma^\gs$, and 
	some finite information $\tuple_{r,\gamma,\eta,i}$
	(whose goal is to describe the class of $r\widetilde\beta^i\gamma\eta$) in such a way that
	either
	\begin{compactenum}[(A)]
	\item	$\tuple_{r,\gamma,\eta,i}=\class{r\widetilde\beta^i\gamma\eta}\in\Below(q\rea{\overline\delta_0}{\omega}\eta_0)$, or
	\item	$\tuple_{r,\gamma,\eta,i}=(r',t,\mu,(\tau_u)_{u\in\eats{\mu}(t)})$ with $r'\in R$, $t\in Q$, $\mu\in\Gamma^{\Dc'+1}$,
		where $r\widetilde\beta^i\gamma\eta\approx t\mu\theta_{r,\gamma,\eta,i}$, and for all $u\in\eats{\mu}(t)$ we have either
		(where the goal of $\tau_u$ is to describe the class of $u\theta_{r,\gamma,\eta,i}$)
		\begin{compactenum}[(B1)]
		\item	$\tau_u=(v_u,\xi_u)\in Q\times\Gamma^\gs$, where
			$u\theta_{r,\gamma,\eta,i}\approx v_u\xi_u\widetilde\beta^{j_{r,\gamma,\eta,i}}\gamma\eta$
			and
			$v_u\xi_u\in\Near(r'\widetilde\beta^{\Bc+2},\allowbreak\Bc\Cc^{2|Q|+1}+\Bc)$, or
		\item	$\tau_u=\class{u\theta_{r,\gamma,\eta,i}}$ and $\dist(\classC_{s_u},\class{u\theta_{r,\gamma,\eta,i}})\leq\Bc$ for some $s_u\in\eats{\rea{\overline\delta_0}{\omega}}(q)$.
		\end{compactenum}
	\end{compactenum}
	If there are multiple choices for some $r,\gamma,\eta,i$, we just take any of them.
	In the following claim we prove that there is always at least one choice.

	\begin{claim}
		For all $r\in R$, all $(\gamma,\eta)\in \Omega$, and all $i\in\interval{0,e_{\gamma,\eta}-1}$
		there exist $j_{r,\gamma,\eta,i}$, $\theta_{r,\gamma,\eta,i}$, and $\tuple_{r,\gamma,\eta,i}$ satisfying the above conditions.
	\end{claim}
	
	\begin{proofnoqed}
		Let $r\in R$, $(\gamma,\eta)\in \Omega$, and $i\in\interval{0,e_{\gamma,\eta}-1}$.
		To simplify the notation, let us write $e$ for $e_{\gamma,\eta}$, for the chosen $\gamma,\eta$.
		
		If $\class{r\widetilde\beta^i\gamma\eta}\in\Below(q\rea{\overline\delta_0}{\omega}\eta_0)$,
		we can take $\tuple_{r,\gamma,\eta,i}=\class{r\widetilde\beta^i\gamma\eta}$, obtaining Item (A).
		Suppose thus that $\class{r\widetilde\beta^i\gamma\eta}\not\in\Below(q\rea{\overline\delta_0}{\omega}\eta_0)$,
		in which case we prove Item (B).
		Note that $q\rea{\overline\delta_0}{\omega}\eta_0\approx q\rea{\overline\delta_0}{\omega}\eta\approx
		q\widetilde\alpha\widetilde\beta^\omega\gamma\eta \approx q\widetilde\alpha\widetilde\beta^e\gamma\eta$,
		where the first equivalence holds by \cref{equiv-because-equiv-below} because $s\eta\approx s\eta_0$ for all $s\in\eats{\rea{\overline\delta_0}{\omega}}(q)$,
		the second equivalence holds by \cref{only type important} because $\decomp{\overline\alpha,\overline\beta,\gamma}\tpeq\overline\delta_0$,
		and the third equivalence holds by \cref{equiv-because-equiv-below} because $r'\widetilde\beta^\omega\gamma\eta\approx r'\widetilde\beta^e\gamma\eta$
		for all $r'\in\eats{\widetilde\alpha}(q)$ (cf.~the definition of $e$).
		We can thus write $\class{r\widetilde\beta^i\gamma\eta}\not\in\Below(q\widetilde\alpha\widetilde\beta^e\gamma\eta)$.
		
		Recall that $r\in R=\eats{\widetilde\alpha\widetilde\beta^{e-i}}(q)$, which implies that
		there exists a sequence of states $r_e,r_{e-1},\dots,\allowbreak r_i\in R$
		such that $r_e\in\eats{\widetilde\alpha}(q)$, and $r_j\in\eats{\widetilde\beta}(r_{j+1})$ for all $j\in\interval{i,e-1}$, and $r_i=r$.
		It follows that there is a run from $q\widetilde\alpha$ to $r_e$ and runs from $r_{j+1}\widetilde\beta$ to $r_j$ for all $j\in\interval{i,e-1}$,
		each of them reading at most $\Bc\Cc^{2|Q|+1}$ action symbols (cf.\@ \cref{can pop fast decomp}).
		We compose these runs shifted appropriately, so that we obtain a run $\pi$ from $q\widetilde\alpha\widetilde\beta^e\gamma\eta$ to $r\widetilde\beta^i\gamma\eta$
		(going through $r_e\widetilde\beta^e\gamma\eta,r_{e-1}\widetilde\beta^{e-1}\gamma\eta,\dots,r_{i+1}\widetilde\beta^{i+1}\gamma\eta$).
		
		Because $\class{\pi(|\pi|)}\not\in\Below(\pi(0))$,
		the run $\pi$ is not $\Dc$-almost-popping, which by definition means that
		there exists a run $\pi_\Dc'$ that is parallel to a suffix of $\pi$ and satisfies $\sgrow(\pi'_\Dc)\geq\Dc$.
		Let $\pi'_3$ be the shortest suffix of $\pi'_\Dc$ satisfying $\sgrow(\pi'_3)=\Dc'$;
		note that the stack growth 
		of $\pi'_\Dc$ changes by at most one in each step, and $\Dc'<\Dc$,
		so such a suffix exists.
		For some $t_1,t\in Q$, $X\in\Gamma$, $\mu\in\Gamma^{\Dc'+1}$, and $\theta_{r,\gamma,\eta,i}\in\Gamma^\gs$ we can write
		\begin{align*}
			t_1X\theta_{r,\gamma,\eta,i}\xrightarrow{\pi_3'}t\mu\theta_{r,\gamma,\eta,i}
		\end{align*}
		(in particular, on the top of the stack of $\pi_3'(0)$ there is a standard symbol,
		pushed earlier by $\pi'_\Dc$).
		Note that $\pi_3'$ can be shifted to a run from $t_1X$ to $t\mu$,
		and that $\pi_3'$ is parallel to a suffix $\pi_3$ of $\pi$.
		
		Because every non-popping transition reads an action symbol,
		$\pi_\Dc'$ reads at least $\Dc-\Dc'=\Bc\Cc^{2|Q|+1}+1$ action symbols before $\pi_3'$ starts.
		Simultaneously $r_e\widetilde\beta^e\gamma\eta$ is reached by $\pi$ before $\Bc\Cc^{2|Q|+1}+1$ action symbols are read, so before $\pi_3$ starts.
		Let $h$ be the smallest number in $\interval{i+1,e}$ such that $r_h\widetilde\beta^h\gamma\eta$ is reached by $\pi$ before $\pi_3$ starts.
		We take $j_{r,\gamma,\eta,i}=h-\Bc-2$ and $r'=r_h$. 
		Let us prove that $j_{r,\gamma,\eta,i}>i$.
		Clearly $j_{r,\gamma,\eta,i}+\Bc+2>i$, by definition.
		The configuration $\pi_3(0)$ can be written as $q_1\chi\beta^{j_{r,\gamma,\eta,i}+B+1}\gamma\eta$ for some $q_1\in Q$ and $\chi\in\Gamma^\gs$.
		Let $\pi_1$ and $\pi_2$ be the fragments of $\pi$ before $r_h\widetilde\beta^h\gamma\eta$, and between $r_h\widetilde\beta^h\gamma\eta$ and $\pi_3(0)$, respectively;
		we have
		\begin{align*}
			q\widetilde\alpha\widetilde\beta^e\gamma\eta\xrightarrow{\pi_1}
			r'\widetilde\beta^{j_{r,\gamma,\eta,i}+B+2}\gamma\eta\xrightarrow{\pi_2}
			q_1\chi\widetilde\beta^{j_{r,\gamma,\eta,i}+B+1}\gamma\eta\xrightarrow{\pi_3}
			r\widetilde\beta^i\gamma\eta\,.
		\end{align*}
		Recalling how $\pi$ was constructed, we see that $\pi_3$ reads at most $(j_{r,\gamma,\eta,i}+B+2-i)\cdot\Bc\Cc^{2|Q|+1}$ action symbols
		(it is a suffix of a run popping $j_{r,\gamma,\eta,i}+B+2-i$ copies of $\widetilde\beta$);
		likewise $\pi_2$ reads at most $\Bc\Cc^{2|Q|+1}$ action symbols.
		On the other hand, every non-popping transition reads some action symbol,
		so $\pi'_3$ (and $\pi_3$ as well) reads at least $|\mu|-1=\Dc'>(\Bc+2)\cdot\Bc\Cc^{2|Q|+1}$ of them.
		This implies that $(j_{r,\gamma,\eta,i}+\Bc+2-i)\cdot\Bc\Cc^{2|Q|+1}>(\Bc+2)\cdot\Bc\Cc^{2|Q|+1}$, that is, $j_{r,\gamma,\eta,i}>i$.
		By construction we also have $r'\in R$, $j_{r,\gamma,\eta,i}\leq e$, and $r\widetilde\beta^i\gamma\eta\approx t\mu\theta_{r,\gamma,\eta,i}$.
		
		It remains to define $\tau_u$ for $u\in\eats{\mu}(t)$, satisfying Items (B1) or (B2).
		Take some $u\in\eats{\mu}(t)$.
		Because $t_1X\to^*t\mu$, we also have $t_1X\to^*u$, so $\dist(t_1X,u)\leq\Bc$ by \cref{can pop fast}.
		On the other hand $\dist(t_1X,u)>0$, because $\epsilon$-transitions are deterministic and cannot push, while there is a run from $t_1X$ to $u$ going through $t\mu$
		(a configuration with a larger stack).
		Since $q_1\chi\widetilde\beta^{j_{r,\gamma,\eta,i}+\Bc+1}\gamma\eta\approx t_1X\theta_{r,\gamma,\eta,i}$,
		there is a run $\rho$ from $q_1\chi\widetilde\beta^{j_{r,\gamma,\eta,i}+\Bc+1}\gamma\eta$ to a configuration $c$ in $\class{u\theta_{r,\gamma,\eta,i}}$,
		reading at most $\Bc$ action symbols.
		Let us choose $c$ so that no $\epsilon$-transitions are performed at the end of this run (this is possible because $\epsilon$-transitions do not change the class).
		Because $\rho$ reads some action symbols, also the composition $\pi_2\circ\rho$ does not end with an $\epsilon$-transition.
		We have two cases:
		\begin{compactitem}
		\item	Suppose first that $\pi_2\circ\rho$ may be shifted to a run from $r'\widetilde\beta^{\Bc+2}$.
			This means that its final configuration $c$ may be written as $v_u\xi_u\widetilde\beta^{j_{r,\gamma,\eta,i}}\gamma\eta$,
			and the considered shift of the run $\pi_2\circ\rho$ witnesses that $v_u\xi_u\in\Near(r'\widetilde\beta^{\Bc+2},\Bc\Cc^{2|Q|+1}+\Bc)$;
			taking $\tau_u=(v_u,\xi_u)$ we obtain Item (B1).
		\item	Suppose conversely: $\pi_2\circ\rho$ cannot be shifted to a run from $r'\widetilde\beta^{\Bc+2}$.
			In other words, $\pi_2\circ\rho$ (before its end) reaches a configuration with stack $\widetilde\beta^{j_{r,\gamma,\eta,i}}\gamma\eta$.
			Recall that the stack in the $\pi_2$ part always contains $\widetilde\beta^{j_{r,\gamma,\eta,i}+\Bc+1}\gamma\eta$ as a suffix.
			Thus, $\rho$ pops the $\Bc+1$ copies of $\widetilde\beta$ while reading at most $\Bc$ action symbols;
			some copy of $\widetilde\beta$ is popped using only $\epsilon$-transitions.
			After popping this copy, we reach a configuration of the form $r''\widetilde\beta^{j_{r,\gamma,\eta,i}+k}\gamma\eta$,
			where $k\geq 0$ and $r''\in\eats{\widetilde\beta}_\epsilon(Q)$.
			By well-formedness of $\overline\delta_0$	
			this implies that $\eats{\widetilde\beta}_\epsilon(r'')=\set{r''}$ and $\eats{\gamma}_\epsilon(r'')=\eats{\gamma_0}_\epsilon(r'')\neq\emptyset$.
			Recall that $\epsilon$-transitions are deterministic and that $\rho$ does not end with an $\epsilon$-transition.
			It follows that $\rho$, after visiting $r''\widetilde\beta^{j_{r,\gamma,\eta,i}+k}\gamma\eta$, necessarily pops all the remaining copies of $\widetilde\beta$, then pops $\gamma$,
			and then continues from $s_u\eta$, for some $s_u\in Q$.
			Because $\rho$ reads at most $\Bc$ action symbols (and ends in $c\in\class{u\theta_{r,\gamma,\eta,i}}$),
			we obtain $\dist(\class{s_u\eta},\class{u\theta_{r,\gamma,\eta,i}})\leq\Bc$.
			A prefix of $\pi_1\circ\pi_2\circ\rho$ leading to $r''\widetilde\beta^{j_{r,\gamma,\eta,i}+k}\gamma\eta$
			can be shifted to a run from $q\widetilde\alpha\widetilde\beta^{e-j_{r,\gamma,\eta,i}-k}$ to $r''$,
			witnessing that $r''\in\eats{\widetilde\alpha\widetilde\beta^{e-j_{r,\gamma,\eta,i}-k}}(q)=\eats{\widetilde\alpha}(q)$.
			Together with $r''\in\eats{\widetilde\beta}_\epsilon(r'')$ and $s_u\in\eats{\gamma}_\epsilon(r'')=\eats{\gamma_0}_\epsilon(r'')$ this gives us
			$s_u\in\eats{\rea{\overline\delta_0}{\omega}}(q)$ (cf.~\cref{easy-eats-decomp}).
			Thus, by assumptions of the \lcnamecref{main lemma}, we have $s_u\eta\in\classC_{s_u}$,
			which gives us $\dist(\classC_{s_u},\class{u\theta_{r,\gamma,\eta,i}})\leq\Bc$, as needed for Item (B2), where we take $\tau_u=\class{u\theta_{r,\gamma,\eta,i}}$.
		\qedhere\end{compactitem}
	\end{proofnoqed}
	
	Observe that there is a number $\Tc\in 2^{2^{|\Pp|^{O(1)}}}$ such that
	\begin{align*}
		|\set{\tuple_{r,\gamma,\eta,i}\mid r\in R,(\gamma,\eta)\in\Omega,i\in\interval{0,e_{\gamma,\eta}-1}}|\leq K+\Tc\,.
	\end{align*}
	Indeed, we have $|\Below(q\rea{\overline\delta_0}{\omega}\eta_0)|$ possible values of $\tuple_{r,\gamma,\eta,i}$ conforming with Item (A) of the definition,
	which is at most $K+\Cc^{2|Q|+2}\cdot|\Pp|^{\Dc+1}$ by \cref{few-conf-without-pushing}.
	For Item (B) we have $|Q|^2\cdot|\Gamma|^{\Dc'+1}$ possibilities for $r',t,\mu$, and then for every of at most $|Q|$ states $u$ we have
	\begin{compactitem}
	\item	at most $|\Near(r'\widetilde\beta^{\Bc+2},\Bc\Cc^{2|Q|+1}+\Bc)|$ possibilities for $v_u\xi_u$ in Item (B1),
		which is at most $|\Pp|^{\Bc\Cc^{2|Q|+1}+\Bc}$ by \cref{determ short run}, and
	\item	at most $|Q|$ possibilities for a state $s_u$,
		and then at most $|\set{\classC\mid\dist(\classC_{s_u},\classC)\leq\Bc}|$ possibilities for $\class{u\theta_{r,\gamma,\eta,i}}$ in Item (B2),
		which is at most $|\Pp|^{\Bc}$ by \cref{determ short run}.
	\end{compactitem}

	\begin{claim}\label{can simulate}
		If $r\in R$, $(\gamma,\eta),(\gamma',\eta')\in\Omega$, $i\in\interval{0,e_{\gamma,\eta}-1}$, $i'\in\interval{0,e_{\gamma',\eta'}-1}$ are such that
		$\tuple_{r,\gamma,\eta,i}=\tuple_{r,\gamma',\eta',i'}$, then
		\begin{align*}
			\eqlev(r\widetilde\beta^i\gamma\eta,r\widetilde\beta^{i'}\gamma'\eta')\geq
			\min\set{\eqlev(r''\widetilde\beta^{j_{r,\gamma,\eta,i}}\gamma\eta,r''\widetilde\beta^{j_{r,\gamma',\eta',i'}}\gamma'\eta')\mid r''\in R}\,.
		\end{align*}
	\end{claim}
	
	\begin{proof}
		If $\tuple_{r,\gamma,\eta,i}$ (which equals $\tuple_{r,\gamma',\eta',i'}$) is defined according to Item (A), we have
		$\class{r\widetilde\beta^i\gamma\eta}=\tuple_{r,\gamma,\eta,i}=\class{r\widetilde\beta^{i'}\gamma'\eta'}$, that is,
		$\eqlev(r\widetilde\beta^i\gamma\eta,r\widetilde\beta^{i'}\gamma'\eta')=\omega$.
		Suppose that $\tuple_{r,\gamma,\eta,i}$ is defined according to Item (B):
		we have $\tuple_{r,\gamma,\eta,i}=(r',t,\mu,(\tau_u)_{u\in\eats{\mu}(t)})$,
		where $r\widetilde\beta^i\gamma\eta\approx t\mu\theta_{r,\gamma,\eta,i}$ and $r\widetilde\beta^{i'}\gamma'\eta'\approx t\mu\theta_{r,\gamma',\eta',i'}$.
		By \cref{equiv-because-equiv-below} we thus have
		\begin{align*}
			\eqlev(r\widetilde\beta^i\gamma\eta,r\widetilde\beta^{i'}\gamma'\eta')\geq
			\min\set{\eqlev(u\theta_{r,\gamma,\eta,i},u\theta_{r,\gamma',\eta',i'})\mid u\in\eats{\mu}(t)}\,.
		\end{align*}
		For each $u\in\eats{\mu}(t)$ we then consider the component $\tau_u$:
		\begin{compactitem}
		\item	If it is defined according to Item (B2), we simply have $\class{u\theta_{r,\gamma,\eta,i}}=\tau_u=\class{u\theta_{r,\gamma',\eta',i'}}$, that is,
			$\eqlev(u\theta_{r,\gamma,\eta,i},u\theta_{r,\gamma',\eta',i'})=\omega$.
		\item	Suppose that $\tau_u$ is defined according to Item (B1): we have $\tau_u=(v_u,\xi_u)$,
			where $u\theta_{r,\gamma,\eta,i}\approx v_u\xi_u\widetilde\beta^{j_{r,\gamma,\eta,i}}\gamma\eta$ and
			$u\theta_{r,\gamma',\eta',i'}\approx v_u\xi_u\widetilde\beta^{j_{r,\gamma',\eta',i'}}\gamma'\eta'$;
			moreover $v_u\xi_u\in\Near(r'\widetilde\beta^{\Bc+2},\Bc\Cc^{2|Q|+1}+\Bc)$,
			meaning that $r'\widetilde\beta^{\Bc+2}\to^*v_u\xi_u$, %\sg{$\epsilon$ to remove in $r'\widetilde\beta^{\Bc+2}\to_\epsilon^*v_u\xi_u$}
			which implies that $\eats{\xi_u}(v_u)\subseteq\eats{\widetilde\beta^{\Bc+2}}(r')\subseteq\eats{\widetilde\beta^{\Bc+2}}(R)=R$.
			By \cref{equiv-because-equiv-below} we have
			\begin{align*}
				\eqlev(u\theta_{r,\gamma,\eta,i},u\theta_{r,\gamma',\eta',i'})
				&=\eqlev(v_u\xi_u\widetilde\beta^{j_{r,\gamma,\eta,i}}\gamma\eta,v_u\xi_u\widetilde\beta^{j_{r,\gamma',\eta',i'}}\gamma'\eta')\\
				&\geq
				\min\set{\eqlev(r''\widetilde\beta^{j_{r,\gamma,\eta,i}}\gamma\eta,r''\widetilde\beta^{j_{r,\gamma',\eta',i'}}\gamma'\eta')\mid r''\in\eats{\xi_u}(v_u)}\\
				&\geq\min\set{\eqlev(r''\widetilde\beta^{j_{r,\gamma,\eta,i}}\gamma\eta,r''\widetilde\beta^{j_{r,\gamma',\eta',i'}}\gamma'\eta')\mid r''\in R}\,.
			\end{align*}
		\end{compactitem}
		By combining the inequalities obtained above we get the thesis of the claim.
	\end{proof}

	\begin{claim}
		For all $(\gamma,\eta)\in\Omega$ we have $e_{\gamma,\eta}\leq|Q|\cdot(K+\Tc)$.
	\end{claim}
	
	\begin{proof}
		Fix some $(\gamma,\eta)\in\Omega$.
		For every $i\in\Nat$ let $M_i=\min\set{\eqlev(r\widetilde\beta^i\gamma\eta,r\widetilde\beta^\omega\gamma\eta)\mid r\in R}$,
		and let $r_i\in R$ be a state such that $\eqlev(r_i\widetilde\beta^i\gamma\eta,r_i\widetilde\beta^\omega\gamma\eta)=M_i$.
		To shorten the notation denote $e=e_{\gamma,\eta}$,
		$\tuple_i=\tuple_{r_i,\gamma,\eta,i}$, and
		$j_i=j_{r_i,\gamma,\eta,i}$
		for $i\in\interval{0,e-1}$.
		
		By \cref{eqlev-grows-while-pumping} we have $M_{i+1}\geq 1+M_i$ for all $i\in\Nat$.
		Note that $M_i=\omega$ implies $(\class{r\widetilde\beta^i\gamma\eta})_{r\in R}=(\class{r\widetilde\beta^\omega\gamma\eta})_{r\in R}$, which is impossible for $i<e$.
		Thus for all $i<e$ we have $M_{i+1}>M_i$ with $M_i\in\Nat$.

		It remains to prove that $(r_i,\tuple_i)\neq(r_{i'},\tuple_{i'})$ whenever $0\leq i<i'<e$:
		because there are at most $|Q|$ possibilities for $r_i$ and at most $K+\Tc$ possibilities for $\tuple_i$, this implies that $e\leq|Q|\cdot(K+\Tc)$.
		Suppose to the contrary that $(r_i,\tuple_i)=(r_{i'},\tuple_{i'})$ for some $i,i'$ with $0\leq i<i'<e$.
		Recall that $M_0<M_1<\dots<M_e$, that is, $M_i<M_{i'}$.
		Because $\eqlev(r_i\widetilde\beta^i\gamma\eta,r_i\widetilde\beta^\omega\gamma\eta)=M_i$ and
		because $\eqlev(r_i\widetilde\beta^\omega\gamma\eta,r_{i'}\widetilde\beta^{i'}\gamma\eta)=M_{i'}>M_i$ (recall that $r_i=r_{i'}$),
		we have $\eqlev(r_i\widetilde\beta^i\gamma\eta,r_{i'}\widetilde\beta^{i'}\gamma\eta)=M_i$.
		Using \cref{can simulate} we obtain
		\begin{align*}
			M_i=\eqlev(r_i\widetilde\beta^i\gamma\eta,r_{i'}\widetilde\beta^{i'}\gamma\eta)
			\geq\min\set{\eqlev(r''\widetilde\beta^{j_i}\gamma\eta,r''\widetilde\beta^{j_{i'}}\gamma\eta)\mid r''\in R}\,.
		\end{align*}
		For all $r''\in R$, because $\eqlev(r''\widetilde\beta^{j_i}\gamma\eta,r''\widetilde\beta^\omega\gamma\eta)\geq M_{j_i}>M_i$ and
		$\eqlev(r''\widetilde\beta^\omega\gamma\eta,r''\widetilde\beta^{j_{i'}}\gamma\eta)\geq M_{j_{i'}}>M_{i'}>M_i$, we also have
		$\eqlev(r''\widetilde\beta^{j_i}\gamma\eta,r''\widetilde\beta^{j_{i'}}\gamma\eta)>M_i$.
		We thus obtain $M_i>M_i$, a contradiction.
	\end{proof}
	
	To every $(\gamma,\eta)\in\Omega$ we now assign the tuple $((\tuple_{r,\gamma,\eta,0},j_{r,\gamma,\eta,0}))_{r\in R}$.
	The number of possibilities for $\tuple_{r,\gamma,\eta,0}$ is at most $K+\Tc$,
	and the number of possibilities for $j_{r,\gamma,\eta,0}\in\interval{1,e_{\gamma,\eta}}$ is at most $|Q|\cdot(K+\Tc)$.
	We can safely assume that $|Q|\leq\Tc\leq K+\Tc$, so the number of possible tuples is at most $(K+\Tc)^{3|Q|}$.
	
	To prove the first part of the \lcnamecref{main lemma}, it remains to see that if the same tuple is assigned to two pairs $(\gamma,\eta),(\gamma',\eta')\in\Omega$
	then $r\gamma\eta\approx r\gamma'\eta'$ for all $r\in R$.
	Consider thus two such pairs $(\gamma,\eta),(\gamma',\eta')\in\Omega$.
	Let $M=\min\set{\eqlev(r\gamma\eta,r\gamma'\eta')\mid r\in R}$,
	and let $r_{\min}\in R$ be a state such that $\eqlev(r_{\min}\gamma\eta,r_{\min}\gamma'\eta')=M$.
	By \cref{can simulate} (used with $i=i'=0$) we have
	\begin{align*}
		M=\eqlev(r_{\min}\gamma\eta,r_{\min}\gamma'\eta')\geq
		\min\set{\eqlev(r''\widetilde\beta^j\gamma\eta,r''\widetilde\beta^j\gamma'\eta')\mid r''\in R}\,,
	\end{align*}
	where $j=j_{r,\gamma,\eta,0}=j_{r,\gamma',\eta',0}$.
	Next, for every $r''\in R$ we prove that $\eqlev(r''\widetilde\beta^j\gamma\eta,r''\widetilde\beta^j\gamma'\eta')\geq 1+M$, considering two cases:
	\begin{compactitem}
	\item	Suppose first that $\eats{\widetilde\beta^j}_\epsilon(r'')=\emptyset$.
		Then, by \cref{more-equiv-if-read} we have
		\begin{align*}
			\eqlev(r''\widetilde\beta^j\gamma\eta,r''\widetilde\beta^j\gamma'\eta')
			\geq1+\min\set{\eqlev(r\gamma\eta,r\gamma'\eta')\mid r\in\eats{\widetilde\beta^j}(r'')}\,.
		\end{align*}
		We can conclude recalling that $\eats{\widetilde\beta^j}(r'')\subseteq\eats{\widetilde\beta^j}(R)=R$ and that $\eqlev(r\gamma\eta,r\gamma'\eta')\geq M$ for all $r\in R$.
	\item	Conversely, suppose that there is some
		$r\in\eats{\widetilde\beta^j}_\epsilon(r'')=\eats{\widetilde\beta}_\epsilon(\eats{\widetilde\beta^{j-1}}_\epsilon(r''))\subseteq\eats{\widetilde\beta}_\epsilon(Q)$
		(recall that $j\geq 1$).
		By well-formedness of $\overline\delta_0$ we have $r\in\eats{\widetilde\beta}_\epsilon(r)$
		and $s\in\eats{\gamma_0}_\epsilon(r)=\eats{\gamma}_\epsilon(r)=\eats{\gamma'}_\epsilon(r)$ for some $s$.
		This implies that $r''\widetilde\beta^j\gamma\eta\to_\epsilon^* r\gamma\eta\to_\epsilon^*s\eta$
		and $r''\widetilde\beta^j\gamma'\eta'\to_\epsilon^* r\gamma'\eta'\to_\epsilon^*s\eta'$.
		Recalling that $r\in R=\eats{\widetilde\alpha}(q)$ we also have $s\in\eats{\widetilde\beta\gamma}_\epsilon(\eats{\widetilde\alpha}(q))=\eats{\rea{\overline\delta_0}{\omega}}(q)$
		(cf.~\cref{easy-eats-decomp}).
		By assumptions of the \lcnamecref{main lemma} we then have $s\eta\approx s\eta'$,
		implying that $\eqlev(r''\widetilde\beta^j\gamma\eta,r''\widetilde\beta^j\gamma'\eta')=\omega\geq 1+M$.
	\end{compactitem}
	It follows that $M\geq 1+M$, which is only possible for $M=\omega$;
	thus indeed $r\gamma\eta\approx r\gamma'\eta'$ for all $r\in R$.
	This finishes the proof of the first part of the thesis.
	
	Recall that every tuple in $\Theta$ is of the form $(\class{r\gamma\eta})_{r\in R}$ for some $(\gamma,\eta)\in\Omega$.
	Thus, in order to obtain the second part of the thesis, we should bound the size of the set $\Below(\set{r\gamma\eta\mid r\in R})$ for every pair $(\gamma,\eta)\in\Omega$.
	Fix some $(\gamma,\eta)\in\Omega$.
	We classify classes $\classC\in\Below(\set{r\gamma\eta\mid r\in R})$ as follows:
	\begin{compactenum}
	\item	Suppose that $\classC\in\Below(r\gamma\eta)$ for a state $r\in R$ such that $\class{r\gamma\eta}\in\Below(q\rea{\overline\delta_0}{\omega}\eta_0)$.
		Then $\classC\in\Below(\Below(q\rea{\overline\delta_0}{\omega}\eta_0))$, and the size of this set is at most
		$(K+\Cc^{2|Q|+2}\cdot|\Pp|^{\Dc+1})\cdot|\Pp|^{2\Bc\Dc}$ by \cref{few-conf-without-pushing,Cbelowbelow}.
	\item	Suppose that $\classC\in\Below(r\gamma\eta)$ for a state $r\in R\cap\eats{\widetilde\beta}_\epsilon(Q)$.
		By well-formedness of $\overline\delta_0$ we have
		$r\in\eats{\widetilde\beta}_\epsilon(r)$ and $s\in\eats{\gamma_0}_\epsilon(r)=\eats{\gamma}_\epsilon(r)$ for some $s$,
		which implies $r\gamma\eta\to^*_\epsilon s\eta$ and $s\in\eats{\rea{\overline\delta_0}{\omega}}(q)$.
		Then $r\gamma\eta\approx s\eta\approx s\eta_0$, so $\classC$ belongs to $\Below(\set{s\eta_0\mid s\in\eats{\rea{\overline\delta_0}{\omega}}(q)})$,
		being a set of size $K$.
	\item	Suppose that $\classC\in\Below(\classC')$ for a class $\classC'$ such that $\dist(\classC_s,\classC')\leq\Bc$ for some $s\in\eats{\rea{\overline\delta_0}{\omega}}(q)$.
		We then have a run from $s\eta_0\in\classC_s$ to a configuration in $\classC'$, reading at most $\Bc$ action symbols.
		While reading only $\Bc$ action symbols one cannot push $\Dc>\Bc$ stack symbols (only transitions reading a symbol can push), so this run is necessarily $\Dc$-almost-popping;
		we have $\classC'\in\Below(\set{\classC_s\mid s\in\eats{\rea{\overline\delta_0}{\omega}}(q)})$.
		This means that $\classC\in\Below(\Below(\set{\classC_s\mid s\in\eats{\rea{\overline\delta_0}{\omega}}(q)}))$;
		by \cref{Cbelowbelow} there are at most $K\cdot|\Pp|^{2\Bc\Dc}$ such classes $\classC$.
	\item	Let us assume that none of the three above cases occurs.
		Among $\Dc$-almost-popping runs that start from a configuration in $\{r\gamma\eta\mid r\in R\}$ 
		and lead to a configuration in $\classC$,
		we choose a run $\pi_\classC$ reading the minimal number of action symbols.
		Let us fix the state $r$ for which $\pi_\classC(0)=r\gamma\eta$;
		we have at most $|Q|$ possibilities for that.
		Let us abbreviate $j=j_{r,\gamma,\eta,0}$ and $\theta=\theta_{r,\gamma,\eta,0}$.
		Because the first case above does not apply, necessarily $\tuple_{r,\gamma,\eta,0}$ is of the form $(r',t,\mu,(\tau_u)_{u\in\eats{\mu}(t)})$.
		We have $r\gamma\eta\approx t\mu\theta$, so there is also a $\Dc$-almost-popping run $\pi'_\Cc$ that starts in $t\mu\theta$, leads a configuration in $\classC$,
		and reads the same number of action symbols as $\pi_\classC$.
		The stack $\mu\in\Gamma^{\Dc'+1}$ can be seen as a sequence of $\Dc'+1$ well-formed decompositions of height $0$ (consisting of single stack symbols).
		Thus, by \cref{few-conf-without-pushing-seq}, there are at most $(\Dc'+1)\cdot|\Pp|^{\Dc+1}$ classes for which the run $\pi'_\classC$ can be shifted to a run from $t\mu$.
		For the remaining classes $\classC$ the run $\pi'_\classC$ crosses a configuration with stack $\theta$, that is, a configuration $u\theta$ for some $u\in\eats{\mu}(t)$.
		Let us now also fix this state $u$; we have at most $|Q|$ possibilities for that.
		
		Of course every suffix of a $\Dc$-almost-popping run is again $\Dc$-almost-popping;
		in particular this is the case for the suffix of $\pi'_\classC$ starting in $u\theta$, so $\classC\in\Below(u\theta)$.
		Because the previous case does not apply, we cannot have $\dist(\classC_s,\class{u\theta})\leq\Bc$ for any $s\in\eats{\rea{\overline\delta_0}{\omega}}(q)$.
		Thus $\tau_u$ is necessarily of the form $(v_u,\xi_u)$, where $u\theta\approx v_u\xi_u\widetilde\beta^j\gamma\eta$.
		Let $\pi''_\classC$ be a $\Dc$-almost-popping run from $v_u\xi_u\widetilde\beta^j\gamma\eta$ to a configuration in $\classC$, parallel to $\pi'_\classC$.
		Because $v_u\xi_u\in\Near(r'\widetilde\beta^{\Bc+2},\Bc\Cc^{2|Q|+1}+\Bc)$, there is a run $\rho$ from $r'\widetilde\beta^{j+\Bc+2}\gamma\eta$
		to $v_u\xi_u\widetilde\beta^j\gamma\eta$ reading at most $\Bc\Cc^{2|Q|+1}+\Bc<\Dc$ action symbols.
		A run reading so few action symbols cannot push $\Dc$ stack symbols, so $\rho$ is $\Dc$-almost-popping.
		By \cref{Lbelowbelow} there are at most $|\Pp|^{2\Bc\Dc}-1$ classes $\classC$ for which the composition $\rho\circ\pi''_\classC$ is not $\Dc$-almost-popping.
		For remaining classes $\classC$ the run $\rho\circ\pi''_\classC$ is $\Dc$-almost-popping.
		
		Note that $j+\Bc+2\leq e_{\gamma,\eta}+\Bc+2\leq|Q|\cdot(K+\Tc)+\Bc+2$.
		By \cref{few-conf-without-pushing-seq} (applied to the sequence of $j+\Bc+2$ copies of $\overline\beta$)
		there are at most $(|Q|\cdot(K+\Tc)+\Bc+2)\cdot\Cc^{2|Q|+1}\cdot|\Pp|^{\Dc+1}$ classes $\classC$ for which the run $\rho\circ\pi''_\classC$
		can be shifted to a run from $r'\widetilde\beta^{j+\Bc+2}$.
		For remaining classes $\classC$ the run $\pi''_\classC$ crosses a configuration with stack $\gamma\eta$, that is,
		a configuration $r''\gamma\eta$ for some $r''\in\eats{\widetilde\beta^{j+\Bc+2}}(r')\subseteq R$.
		The suffix of $\pi''_\classC$ starting in $r''\gamma\eta$ is again $\Dc$-almost-popping, that is, $\classC\in\Below(r''\gamma\eta)$.
		Because the second case above does not apply, we have $r''\not\in\eats{\widetilde\beta}_\epsilon(Q)$,
		so the part of $\pi''_\classC$ that pops the last copy of $\widetilde\beta$ necessarily reads some action symbol;
		the suffix of $\pi''_\classC$ starting in $r''\gamma\eta$ reads less action symbols than $\pi_\classC$.
		However this contradicts the minimality of $\pi_\classC$, so there are no such classes $\classC$.
	\end{compactenum}
	
	By summing up the numbers obtained above we get a number $\Uc\in 2^{2^{|\Pp|^{O(1)}}}$ such that we have $|\Below(\set{r\gamma\eta\mid r\in R})|\leq(K+1)\cdot\Uc$.
\end{proof}

Let us reformulate \cref{main lemma} so that it uses a set of states $P$ in place of a single state $q$:

\begin{lemma}\label{main lemma set}
	Let $P\subseteq Q$, let $\overline\delta_0=\decomp{\overline\alpha,\overline\beta,\gamma_0}$ be a well-formed decomposition of degree at most $\Cc$ and height at most $2|Q|+2$,
	let $(\classC_s)_{s\in\eats{\rea{\overline\delta_0}{\omega}}(P)}$ be a tuple of classes, and let $K=|\Below(\set{\classC_s\mid\allowbreak s\in\eats{\rea{\overline\delta_0}{\omega}}(P)})|$.
	Let also $\Omega$ be the set of pairs $(\gamma,\eta)\in\Gamma^*\times\Gamma^*$ for which
	\begin{compactitem}
	\item	$\decomp{\overline\alpha,\overline\beta,\gamma}\tpeq\overline\delta_0$,
	\item	$\initstate\to^*q\rea{\overline\alpha}{1}\rea{\overline\beta}{1}\gamma\eta$ for all $q\in P$, and
	\item	$s\eta\in\classC_s$ for all $s\in\eats{\rea{\overline\delta_0}{\omega}}(P)$.
	\end{compactitem}
	If $(\Ll(\Pp),\initstate)$ is weakly bisimulation finite, then the set of tuples of classes
	\begin{align*}
		\Theta=\set{(\class{r\gamma\eta})_{r\in\eats{\rea{\overline\alpha}{\omega}}(P)}\mid(\gamma,\eta)\in\Omega}
	\end{align*}
	has at most $(K+\Tc)^{3|Q|^2}$ elements.
	Moreover, for every tuple $(\classC'_r)_{r\in\eats{\rea{\overline\alpha}{\omega}}(P)}\in\Theta$
	we have $|\Below(\set{\classC'_r\mid r\in\eats{\rea{\overline\alpha}{\omega}}(P)})|\leq(K+1)\cdot|Q|\cdot\Uc$.
\end{lemma}

\begin{proof}
	For every $q\in P$ let $\Omega_q$ be the set of pairs 
	$(\gamma,\eta)\in\Gamma^*\times\Gamma^*$ for which
	$\decomp{\overline\alpha,\overline\beta,\gamma}\tpeq\overline\delta_0$, $\initstate\to^*q\rea{\overline\alpha}{1}\rea{\overline\beta}{1}\gamma\eta$, and
	$s\eta\in\classC_s$ for all $s\in\eats{\rea{\overline\delta_0}{\omega}}(q)$;
	let also $\Theta_q=\set{(\class{r\gamma\eta})_{r\in\eats{\rea{\overline\alpha}{\omega}}(q)}\mid(\gamma,\eta)\in\Omega_q}$.
	\cref{main lemma} implies that $|\Theta_q|\leq(K+\Tc)^{3|Q|}$ and that
	$|\Below(\set{\classC'_r\mid r\in\eats{\rea{\overline\alpha}{\omega}}(q)})|\leq(K+1)\cdot\Uc$ for every tuple $(\classC'_r)_{r\in\eats{\rea{\overline\alpha}{\omega}}(q)}\in\Theta_q$.

	Observe that $\Omega=\bigcap_{q\in P}\Omega_q$.
	In consequence, for every tuple $(\classC'_r)_{r\in\eats{\rea{\overline\alpha}{\omega}}(P)}\in\Theta$ and for every $q\in P$,
	the sub-tuple $(\classC'_r)_{r\in\eats{\rea{\overline\alpha}{\omega}}(q)}$ belongs to $\Theta_q$.
	Because $\eats{\rea{\overline\alpha}{\omega}}(P)=\bigcup_{q\in P}\eats{\rea{\overline\alpha}{\omega}}(q)$ (i.e., every full tuple is completely determined by the sub-tuples for all $q\in P$),
	we thus have $|\Theta|\leq\Pi_{q\in P}|\Theta_q|\leq\big((K+\Tc)^{3|Q|})^{|Q|}$.
	Moreover, for every tuple $(\classC'_r)_{r\in\eats{\rea{\overline\alpha}{\omega}}(P)}\in\Theta$ we have
	$\Below(\set{\classC'_r\mid r\in\eats{\rea{\overline\alpha}{\omega}}(P)})=\bigcup_{q\in P}\Below(\set{\classC'_r\mid r\in\eats{\rea{\overline\alpha}{\omega}}(q)})$,
	which implies $|\Below(\set{\classC'_r\mid r\in\eats{\rea{\overline\alpha}{\omega}}(P)})|\leq|Q|\cdot (K+1)\cdot\Uc$, as required.
\end{proof}

While in \cref{main lemma set} we describe the situation after appending $\gamma$ to a stack $\eta$ (where $\gamma$ is the last component of some pumping triple),
in \cref{main lemma decomp} we describe the situation after appending the whole $\rea{\overline\delta}{1}$ for a decomposition $\overline\delta$.
The decomposition $\overline\delta$ may contain multiple pumping triples (also in a nested way);
\cref{main lemma decomp} uses \cref{main lemma set} for each of these triples.

\begin{lemma}\label{main lemma decomp}
	Let $P\subseteq Q$, let $\overline\delta_0$ be a well-formed decomposition of degree at most $\Cc$ and height $\ell\leq 2|Q|+2$,
	let $(\classC_s)_{s\in\eats{\rea{\overline\delta_0}{\omega}}(P)}$ be a tuple of classes, and let $K=|\Below(\set{\classC_s\mid s\in\eats{\rea{\overline\delta_0}{\omega}}(P)})|$.
	Let also $\Omega$ be the set of pairs $(\overline\delta,\eta)$ (of a decomposition and a standard stack content) for which
	\begin{compactitem}
	\item	$\overline\delta\tpeq\overline\delta_0$,
	\item	$\initstate\to^*q\rea{\overline\delta}{1}\eta$ for all $q\in P$, and
	\item	$s\eta\in\classC_s$ for all $s\in\eats{\rea{\overline\delta_0}{\omega}}(P)$.
	\end{compactitem}
	If $(\Ll(\Pp),\initstate)$ is weakly bisimulation finite, then the set of tuples of classes
	\begin{align*}
		\Theta=\set{(\class{q\rea{\overline\delta}{1}\eta})_{q\in P}\mid(\overline\delta,\eta)\in\Omega}
	\end{align*}
	has at most $((K+1)\cdot(|Q|\cdot\Uc+1)^{\Cc^\ell}+\Tc)^{3|Q|^2\cdot\Cc^\ell}$ elements.
	Moreover, for every tuple $(\classC'_q)_{q\in P}\in\Theta$
	we have $|\Below(\set{\classC'_q\mid q\in P})|\leq(K+1)\cdot(|Q|\cdot\Uc+1)^{\Cc^\ell}-1$.
\end{lemma}

\begin{proofnoqed}
	We proceed by induction on the structure of $\overline\delta_0$.
	Suppose first that $\overline\delta_0=X\in\Gamma$
	(i.e., $\ell=0$).
	For all $(\overline\delta,\eta),(\hat\delta,\eta')\in\Omega$ we have $\rea{\overline\delta}{1}=X=\rea{\hat\delta}{1}$, and $qX\eta\approx qX\eta'$ for all $q\in P$
	by \cref{equiv-because-equiv-below} (because $s\eta\approx s\eta'$ for all $s\in\eats{X}(q)\subseteq\eats{X}(P)$ by assumption).
	It follows that $|\Theta|\leq 1$.
	Moreover the only tuple in $\Theta$ (if exists) is of the form $(\class{qX\eta})_{q\in P}$ for $(X,\eta)\in\Omega$.
	By \cref{few-conf-without-pushing} we have $|\Below(\set{qX\eta\mid q\in P})|\leq K+|\Pp|^{\Dc+1}$.
	Checking the definition of $\Uc$ in the proof of \cref{main lemma} we see that $|\Pp|^{\Dc+1}\leq\Uc$ 
	(already the component of $\Uc$ needed to handle Case 1 is greater than $|\Pp|^{\Dc+1}$),
	which gives us the desired inequality
	$|\Below(\set{qX\eta\mid q\in P})|\leq (K+1)\cdot(|Q|\cdot\Uc+1)-1$.
		
	Next, suppose that $\overline\delta_0=\overline\delta_{0,1}\dots\overline\delta_{0,k}$.
	Let $P_0=P$ and $P_i=\eats{\rea{\overline\delta_{0,i}}{\omega}}(P_{i-1})$ for $i\in\interval{1,k}$.
	Note that $P_k=\eats{\rea{\overline\delta_0}{\omega}}(P)$.
	For each $i\in\interval{0,k}$ and each tuple of classes $\tupleS=(\classD_s)_{s\in P_i}$
	let $K_{i,\tupleS}=|\Below(\set{\classD_s\mid s\in P_i})|$;
	when $i\in\interval{1,k}$ let $\Omega_{i,\tupleS}$ be the set of pairs $(\overline\delta_i,\eta_i)$ for which
	$\overline\delta_i\tpeq\overline\delta_{0,i}$,
	$\initstate\to^*q\rea{\overline\delta_i}{1}\eta_i$ for all $q\in P_{i-1}$, and
	$s\eta_i\in\classD_s$ for all $s\in P_i$;
	let also
	$\Theta_{i,\tupleS}=\set{(\class{q\rea{\overline\delta_i}{1}\eta_i})_{q\in P_{i-1}}\mid(\overline\delta_i,\eta_i)\in\Omega_{i,\tupleS}}$.
	The induction hypothesis (where as $P$ we take $P_{i-1}$, as $\overline\delta_0$ we take $\overline\delta_{0,i}$, and as $(\classC_s)_{s\in P_i}$ we take $\tupleS$)
	says that
	\begin{inequalities}
		|\Theta_{i,\tupleS}|&\leq ((K_{i,\tupleS}+1)\cdot(|Q|\cdot\Uc+1)^{\Cc^{\ell-1}}+\Tc)^{3|Q|^2\cdot\Cc^{\ell-1}}&&\mbox{and}\label{ineq-theta}\\
		K_{i-1,\tupleS'}&\leq(K_{i,\tupleS}+1)\cdot(|Q|\cdot\Uc+1)^{\Cc^{\ell-1}}-1
			&&\mbox{for all }\tupleS'\in\Theta_{i,\tupleS}.\label{ineq-K}
	\end{inequalities}
	
	Consider now some tuple $(\classC_q')_{q\in P}\in\Theta$.
	It is of the form $(\class{q\rea{\overline\delta}{1}\eta})_{q\in P}$ for some $(\overline\delta,\eta)\in\Omega$.
	Because $\overline\delta\tpeq\overline\delta_0$, we have $\overline\delta=\overline\delta_1\dots\overline\delta_k$, where $\overline\delta_i\tpeq\overline\delta_{0,i}$
	for all $i\in\interval{1,k}$.
	For $i\in\interval{0,k}$ let $\eta_i=\rea{\overline\delta_{i+1}}{1}\dots\rea{\overline\delta_k}{1}\eta$,
	and $\tupleS_i=(\class{s\eta_i})_{s\in P_i}$.
	By definition of $\Omega$ we have
	$\initstate\to^* q\rea{\overline\delta}{1}\eta$ for all $q\in P$.
	By \cref{only type important,epsilon-preserved} we have $P_{i-1}=\eats{\rea{\overline\delta_{0,1}}{\omega}\dots\rea{\overline\delta_{0,i-1}}{\omega}}(P)
	=\eats{\rea{\overline\delta_1}{\omega}\dots\rea{\overline\delta_{i-1}}{\omega}}(P)\subseteq\eats{\rea{\overline\delta_1}{1}\dots\rea{\overline\delta_{i-1}}{1}}(P)$,
	which implies $\initstate\to^* q\rea{\overline\delta_i}{1}\eta_i$ for all $q\in P_{i-1}$ and $i\in\interval{1,k}$
	(we can first push the whole stack content $q\rea{\overline\delta}{1}\eta$, and then pop its topmost part).
	It follows that $(\overline\delta_i,\eta_i)\in\Omega_{i,\tupleS_i}$,
	hence also $\tupleS_{i-1}=(\class{q\rea{\overline\delta_i}{1}\eta_i})_{q\in P_{i-1}}\in\Theta_{i,\tupleS_i}$ for all $i\in\interval{1,k}$.
	\cref{ineq-K} used for the tuple $\tupleS_{i-1}\in\Theta_{i,\tupleS_i}$ gives us
	\begin{align*}
		K_{i-1,\tupleS_{i-1}}\leq(K_{i,\tupleS_i}+1)\cdot(|Q|\cdot\Uc+1)^{\Cc^{\ell-1}}-1 && \mbox{for all }i\in\interval{1,k}.
	\end{align*}
	Notice that $\tupleS_k=(\class{s\eta})_{s\in P_k}=(\classC_s)_{s\in P_k}$ and $\tupleS_0=(\class{q\rea{\overline\delta}{1}\eta})_{q\in P}=(\classC_q')_{q\in P}$;
	in particular $K_{k,\tupleS_k}=K$ and $\Below(\set{\classC'_q\mid q\in P})=K_{0,\tupleS_0}$.
	Thus, the above inequalities imply that
	\begin{align*}
		K_{i,\tupleS_i}\leq(K+1)\cdot(|Q|\cdot\Uc+1)^{(k-i)\cdot\Cc^{\ell-1}}-1 && \mbox{for all }i\in\interval{0,k}.
	\end{align*}
	In particular, because $k\leq\Cc$, we have
	\begin{inequalities}
		\Below(\set{\classC'_q\mid q\in P})=K_{0,\tupleS_0}&\leq(K+1)\cdot(|Q|\cdot\Uc+1)^{\Cc\cdot\Cc^{\ell-1}}-1&&\mbox{and}\label{ineq:3}\\
		K_{i,\tupleS_i}&\leq(K+1)\cdot(|Q|\cdot\Uc+1)^{(\Cc-1)\cdot\Cc^{\ell-1}}-1 && \mbox{for all }i\in\interval{1,k}.\label{ineq:4}
	\end{inequalities}
	\cref{ineq:3} already gives us the second part of the thesis of the lemma.
	\cref{ineq:4} can be substituted in \cref{ineq-theta}, giving us
	\begin{inequalities}
		|\Theta_{i,\tupleS_i}|\leq ((K+1)\cdot(|Q|\cdot\Uc+1)^{\Cc\cdot\Cc^{\ell-1}}+\Tc)^{3|Q|^2\cdot\Cc^{\ell-1}}\,.\label{ineq:5}
	\end{inequalities}
	From the above it follows that for every tuple $\tupleS_0\in\Theta$ we have found a sequence of tuples $\tupleS_1,\dots,\tupleS_k$ such that
	$\tupleS_{i-1}\in\Theta_{i,\tupleS_i}$ for all $i\in\interval{1,k}$, where $\tupleS_k=(\classC_s)_{s\in P_k}$ is fixed (does not depend on $\tupleS_0$).
	Notice that there are at most $|\Theta_{k,\tupleS_k}|$ choices for $\tupleS_{k-1}$, then at most $|\Theta_{k-1,\tupleS_{k-1}}|$ choices for $\tupleS_{k-2}$, and so on.
	Thus, by \cref{ineq:5}, we obtain the desired inequality
	\begin{align*}
		|\Theta|\leq \big(((K+1)\cdot(|Q|\cdot\Uc+1)^{\Cc^\ell}+\Tc)^{3|Q|^2\cdot\Cc^{\ell-1}}\big)^k\leq ((K+1)\cdot(|Q|\cdot\Uc+1)^{\Cc^\ell}+\Tc)^{3|Q|^2\cdot\Cc^\ell}\,.
	\end{align*}

	Finally, suppose that $\overline\delta_0=\decomp{\overline\alpha_0,\overline\beta_0,\gamma_0}$.
	Let $R=\eats{\rea{\overline\alpha_0}{\omega}}(P)$ and $S=\eats{\rea{\overline\delta_0}{\omega}}(P)$.
	Notice that $\eats{\rea{\overline\beta_0}{\omega}}(R)=R$, by well-formedness of $\overline\delta_0$.
	Let $\Omega_3$ be the set of pairs $(\gamma,\eta)$ for which
	$\decomp{\overline\alpha_0,\overline\beta_0,\gamma}\tpeq\overline\delta_0$,
	$\initstate\to^*q\rea{\overline\alpha_0}{1}\rea{\overline\beta_0}{1}\gamma\eta$ for all $q\in P$, and
	$s\eta\in\classC_s$ for all $s\in S$;
	let also $\Theta_3=\set{(\class{r\gamma\eta})_{r\in R}\mid(\gamma,\eta)\in\Omega_3}$.
	Going further, for every tuple $\tupleS=(\classD_r)_{r\in R}$
	let $K_{\tupleS}=|\Below(\set{\classD_r\mid r\in R})|$,
	let $\Omega_{2,\tupleS}$ be the set of pairs $(\overline\beta,\eta')$ for which
	$\overline\beta\tpeq\overline\beta_0$,
	$\initstate\to^*r\rea{\overline\beta}{1}\eta'$ for all $r\in R$, and
	$r\eta'\in\classD_r$ for all $r\in R$;
	let also $\Theta_{2,\tupleS}=\set{(\class{r\rea{\overline\beta}{1}\eta'})_{r\in R}\mid(\overline\beta,\eta')\in\Omega_{2,\tupleS}}$.
	Likewise, again for every tuple $\tupleS=(\classD_r)_{r\in R}$,
	let $\Omega_{1,\tupleS}$ be the set of pairs $(\overline\alpha,\eta'')$ for which
	$\overline\alpha\tpeq\overline\alpha_0$,
	$\initstate\to^*q\rea{\overline\alpha}{1}\eta''$ for all $q\in P$, and
	$r\eta''\in\classD_r$ for all $r\in R$;
	let also $\Theta_{1,\tupleS}=\set{(\class{q\rea{\overline\alpha}{1}\eta''})_{q\in Q}\mid(\overline\beta,\eta'')\in\Omega_{1,\tupleS}}$.
	Finally, for every tuple $\tupleS=(\classC'_q)_{q\in P}$ let $K'_{\tupleS}=|\Below(\set{\classC'_q\mid q\in P})|$.
	\cref{main lemma set} gives us inequalities
	\begin{inequalities}
		|\Theta_3|&\leq(K+\Tc)^{3|Q|^2}\leq((K+1)\cdot(|Q|\cdot\Uc+1)^{\Cc^\ell}+\Tc)^{3|Q|^2\cdot\Cc^{\ell-1}} &&\mbox{and}\label{ineq-theta-3}\\
		K_{\tupleS}&\leq(K+1)\cdot|Q|\cdot\Uc\leq(K+1)\cdot(|Q|\cdot\Uc+1)^{\Cc^{\ell-1}}-1 &&\mbox{for all }\tupleS\in\Theta_3.\label{ineq-K-3}
	\end{inequalities}
	For every tuple $\tupleS$ (indexed by elements of $R$), by 
	the induction hypothesis for $\overline\beta_0$ we obtain
	\begin{inequalities}
		|\Theta_{2,\tupleS}|&\leq ((K_{\tupleS}+1)\cdot(|Q|\cdot\Uc+1)^{\Cc^{\ell-1}}+\Tc)^{3|Q|^2\cdot\Cc^{\ell-1}}&&\mbox{and}\label{ineq-theta-2}\\
		K_{\tupleS'}&\leq(K_{\tupleS}+1)\cdot(|Q|\cdot\Uc+1)^{\Cc^{\ell-1}}-1
			&&\mbox{for all }\tupleS'\in\Theta_{2,\tupleS},\label{ineq-K-2}
	\end{inequalities}
	and by the induction hypothesis for $\overline\alpha_0$ we obtain
	\begin{inequalities}
		|\Theta_{1,\tupleS}|&\leq ((K_{\tupleS}+1)\cdot(|Q|\cdot\Uc+1)^{\Cc^{\ell-1}}+\Tc)^{3|Q|^2\cdot\Cc^{\ell-1}}&&\mbox{and}\label{ineq-theta-1}\\
		K'_{\tupleS'}&\leq(K_{\tupleS}+1)\cdot(|Q|\cdot\Uc+1)^{\Cc^{\ell-1}}-1
			&&\mbox{for all }\tupleS'\in\Theta_{1,\tupleS}.\label{ineq-K-1}
	\end{inequalities}

	Consider now some tuple $\tupleS_0=(\classC_q')_{q\in P}\in\Theta$.
	It is of the form $(\class{q\rea{\overline\delta}{1}\eta})_{q\in P}$ for some $(\overline\delta,\eta)\in\Omega$.
	Because $\overline\delta\tpeq\overline\delta_0$, we have $\overline\delta=\decomp{\overline\alpha,\overline\beta,\gamma}$,
	where $\overline\alpha\tpeq\overline\alpha_0$, $\overline\beta\tpeq\overline\beta_0$, $\eats{\gamma}_\epsilon=\eats{\gamma_0}_\epsilon$,
	and $\up(\gamma)=\up(\gamma_0)$.
	Let $\eta'=\gamma\eta$ and $\eta''=\rea{\overline\beta}{1}\gamma\eta$.
	Let also $\tupleS_1=(\class{r\eta''})_{r\in R}$ and $\tupleS_2=(\class{r\eta'})_{r\in R}$.
	By definition of $\Omega$ we have
	$\initstate\to^* q\rea{\overline\delta}{1}\eta=q\rea{\overline\alpha}{1}\eta''$ for all $q\in P$,
	so $(\overline\alpha,\eta'')\in\Omega_{1,\tupleS_1}$.
	We also have $\decomp{\overline\alpha_0,\overline\beta_0,\gamma}\tpeq\overline\delta_0$, so $(\gamma,\eta)\in\Omega_3$.
	By \cref{only type important,epsilon-preserved} we have $R=\eats{\rea{\overline\alpha_0}{\omega}}(P)=\eats{\rea{\overline\alpha}{\omega}}(P)\subseteq\eats{\rea{\overline\alpha}{1}}(P)$,
	which implies $\initstate\to^* r\rea{\overline\beta}{1}\eta'$ for all $r\in R$
	(we can first push the whole stack content $q\rea{\overline\delta}{1}\eta$, and then pop its topmost part), and thus
	$(\overline\beta,\eta')\in\Omega_{2,\tupleS_2}$.
	In consequence $\tupleS_0=(\class{q\rea{\overline\alpha}{1}\eta''})_{q\in P}\in\Theta_{1,\tupleS_1}$,
	$\tupleS_1=(\class{r\rea{\overline\beta}{1}\eta'})_{r\in R}\in\Theta_{2,\tupleS_2}$, and
	$\tupleS_2=(\class{r\gamma\eta})_{r\in R}\in\Theta_3$.
	\cref{ineq-K-3,ineq-K-2,ineq-K-1} give us
	\begin{align*}
		K'_{\tupleS_0}\leq(K+1)\cdot(|Q|\cdot\Uc+1)^{3\Cc^{\ell-1}}-1\leq(K+1)\cdot(|Q|\cdot\Uc+1)^{\Cc^\ell}-1\,,
	\end{align*}
	which is the second part of the thesis of the lemma.
	For $i\in\set{1,2}$ we rather need
	\begin{align*}
		K_{\tupleS_i}\leq(K+1)\cdot(|Q|\cdot\Uc+1)^{2\Cc^{\ell-1}}-1\leq(K+1)\cdot(|Q|\cdot\Uc+1)^{(\Cc-1)\Cc^{\ell-1}}-1\,,
	\end{align*}
	Substituting this to \cref{ineq-theta-2,ineq-theta-1} we obtain that
	\begin{inequalities}
		\max(|\Theta_{1,\tupleS_i}|, |\Theta_{2,\tupleS_2}|)\leq ((K+1)\cdot(|Q|\cdot\Uc+1)^{\Cc\cdot\Cc^{\ell-1}}+\Tc)^{3|Q|^2\cdot\Cc^{\ell-1}}\,.\label{ineq:100}
	\end{inequalities}
	
	From the above it follows that for every tuple $\tupleS_0\in\Theta$ we have found tuples $\tupleS_1,\tupleS_2$ such that
	$\tupleS_0\in\Theta_{1,\tupleS_1}$, $\tupleS_1\in\Theta_{2,\tupleS_2}$, and $\tupleS_2\in\Theta_3$.
	There are at most $|\Theta_3|$ choices for $\tupleS_2$, then at most $|\Theta_{2,\tupleS_2}|$ choices for $\tupleS_1$, and then at most $|\Theta_{1,\tupleS_1}|$ choices for $\tupleS_0$.
	Thus, by \cref{ineq-theta-3,ineq:100}, we obtain the desired inequality
	\begin{align*}
		|\Theta|\leq \big(((K+1)\cdot(|Q|\cdot\Uc+1)^{\Cc^\ell}+\Tc)^{3|Q|^2\cdot\Cc^{\ell-1}}\big)^3\leq ((K+1)\cdot(|Q|\cdot\Uc+1)^{\Cc^\ell}+\Tc)^{3|Q|^2\cdot\Cc^\ell}\,.
	\qedhere\end{align*}
\end{proofnoqed}

\begin{theorem}\label{thm:main}
	If $(\Ll(\Pp),\initstate)$ is weakly bisimulation finite, then it has at most $\Zc$ classes for some $\Zc\in 2^{2^{|\Pp|^{O(1)}}}$.
\end{theorem}

\begin{proof}
	Let $\classC_\bot=\class{s}$ for all $s\in Q$, which is the class containing configurations with no successors.
	Clearly $\Below(\classC_\bot)=1$.
	For every type of a well-formed decomposition (i.e., equivalence class of the $\tpeq$ relation) fix some decomposition $\overline\delta_0$ having this type.
	For such a decomposition $\overline\delta_0$ and for a state $q\in Q$ let $\Omega_{q,\overline\delta_0}$ be the set of pairs $(\overline\delta,\eta)$ for which
	$\overline\delta\tpeq\overline\delta_0$, $\initstate\to^*q\rea{\overline\delta}{1}\eta$, and $s\eta\in\classC_\bot$ for all $s\in\eats{\rea{\overline\delta_0}{\omega}}(q)$.
	Let also $\Theta_{q,\overline\delta_0}=\set{\class{q\rea{\overline\delta}{1}\eta}\mid(\overline\delta,\eta)\in\Omega_{q,\overline\delta_0}}$.
	Assuming that $\overline\delta_0$ has degree at most $\Cc$ and height at most $2|Q|+2$, by \cref{main lemma decomp} we have
	\begin{align*}
		|\Theta_{q,\overline\delta_0}|\leq\Vc=((1+1)\cdot(|Q|\cdot\Uc+1)^{\Cc^{2|Q|+2}}+\Tc)^{3|Q|^2\cdot\Cc^{2|Q|+2}}\,.
	\end{align*}

	Consider now an arbitrary configuration $q\delta$ reachable from $\initstate$.
	Configurations with empty stack belong to $\classC_\bot$.
	If $\delta$ is nonempty, by \cref{decomp-exists} we obtain a well-formed decomposition $\overline\delta$ of height at most $2|Q|+2$
	and degree at most $\Cc$.
	Let $\overline\delta_0$ be the fixed representative of the type of $\overline\delta$.
	We see that $(\overline\delta,\epsilon)\in\Omega_{q,\overline\delta_0}$, so $\class{q\delta}\in\Theta_{q,\overline\delta_0}$.
	By \cref{num types} there are at most $2^{|\Pp|^4\cdot(\Cc+1)^{2|Q|+3}}$ choices for the type of $\overline\delta$, hence for the decomposition $\overline\delta_0$.
	Once $q$ and $\overline\delta_0$ is fixed, we have only $\Vc$ choices for the class $\class{q\delta}\in\Theta_{q,\overline\delta_0}$.
	Thus, in total we have at most $\Zc=1+|Q|\cdot2^{|\Pp|^4\cdot(\Cc+1)^{2|Q|+3}}\cdot\Vc\in 2^{2^{|\Pp|^{O(1)}}}$ possible classes.
\end{proof}

\section{The Algorithm}\label{sec:algorithm}

In this section we prove the following \lcnamecref{thm:algorithm}:

\begin{theorem}\label{thm:algorithm}
	Given an \PDS $\Pp=(Q,\Gamma,\Act,\Delta)$, its initial configuration $\initstate\in Q\Gamma$, and a number $Z\in\Nat$,
	one can decide in time $O(Z^{|\Pp|^{O(1)}})$ whether the number of classes of $\Pp$ reachable from $\class{\initstate}$ is at most $Z$.
\end{theorem}

Having in mind results from the previous section (in particular, \cref{thm:main}), it follows that the weak bisimulation finiteness problem for \PDS can be solved in 2-\EXPTIME.

We remark that the algorithm provided by \cref{thm:algorithm} not only answers Yes or No,
but in the case of a positive answer, it actually computes a description of the weak bisimulation quotient of $\Pp$,
which is an $\epsilon$-LTS that is weakly bisimilar to $\Pp$ and has at most $Z$ configurations.

Let us fix the input to our problem: an \PDS $\Pp$, an initial configuration $\initstate$, and a bound $Z\in\Nat$.
The algorithm will compute relations $\approx_k$ over the set of reachable configurations, for consecutive $k=0,1,2,\dots$.
By definition we have $(\approx)\subseteq(\approx_{k+1})\subseteq(\approx_k)$ for all $k\in\Nat$,
and it is easy to see that if $(\approx_{k+1})=(\approx_k)$ for some $k$, then necessarily $(\approx)=(\approx_k)$.
It follows that either every $\approx_{k+1}$ has more classes than $\approx_k$ until the number of classes starts exceeding $Z$,
or at some moment, necessarily for $k< Z$, the number of classes stops growing at quantity at most $Z$ with $(\approx)=(\approx_{k+1})=(\approx_k)$.
It is thus enough to examine the first $Z$ relations $\approx_k$, checking whether any of them has more than $Z$ classes.

In order to avoid special treatment of configurations with empty stack, we assume that such configurations can never be reached.
This is without loss of generality: we can add a new initial configuration $q'_{\init}X_{\init}'$ together with a transition $q_{\init}'X_{\init}'\rewrite{\init}\initstate X_{\init}'$
(for a fresh action symbol $\init$).
After performing this transition, the system behaves as previously, but has additionally the $X_{\init}'$ symbol on the bottom on the stack;
such a change adds exactly one new weak bisimulation class, containing the new initial configuration.

Let us now see how we can represent the relations $\approx_k$ in a finite, succinct way.
First, note that $\approx_0$ always has exactly one class, so we need no representation for it.
Consider now some $k\geq 1$.
Note that a class of $\approx_k$ is uniquely described by a tuple of classes of $\approx_{k-1}$ reachable after reading a single action symbol.
More formally, for a configuration $c$ we define
\begin{align*}
	\Desc_k(c)=\set{(a,\classk{d}{k-1})\mid c\weak{a}d, a\in\Act}\,;
\end{align*}
then we have $\Desc_k(c)=\Desc_k(c')\Leftrightarrow c\approx_k c'$.
We can thus take $\Desc_k(c)$ as a finite representation of the class $\classk{c}{k}$.
By \cref{determ short run} $\Desc_k(c)$, as a set, has size at most $|\Pp|$,
and while storing it in memory, we can remember every class $\classk{d}{k-1}$ as a number of this class on the list of all classes of $\approx_{k-1}$.
The memory size needed for storing $\Desc_k(c)$ is thus negligible compared to the desired complexity of our algorithm.

In our algorithm, beside of a list of classes, we also compute a partial function $\Cons_k$ saying how a class changes when a stack grows.
It is a function that assigns a class of $\approx_k$ to some tuples $(q,X,(\classC_r)_{r\in\eats{X}(q)})$ with $q\in Q$, $X\in\Gamma$, and $\classC_r$ being classes of $\approx_k$,
defined by taking $\Cons_k(q,X,(\classk{r\alpha}{k})_{r\in\eats{X}(q)})=\classk{qX\alpha}{k}$ for all reachable configurations $qX\alpha$.
Note that $(\classk{r\alpha}{k})_{r\in\eats{X}(q)}=(\classk{r\alpha'}{k})_{r\in\eats{X}(q)}$ implies $qX\alpha\approx_k qX\alpha'$ by \cref{equiv-because-equiv-below},
meaning that $\classk{qX\alpha}{k}$ depends only on the classes $\classk{r\alpha}{k}$ (without necessarily
knowing the stack content $\alpha$).
While storing $\Cons_k$, we represent both the result and the classes $\classC_r$ as their numbers on the list of all reachable classes.

Moreover, for every reachable class $\classC$ of $\approx_k$ let $\textsc{Up}_k(\classC)$ be the class of $\approx_{k-1}$ containing $\classC$ (recall that $c\approx_k c'$ implies $c\approx_{k-1}c'$).

We have already said that our algorithm computes the relations $\approx_k$ for consecutive $k=0,1,2,\dots$.
Now we can be more precise: for every $k$, it computes a list of reachable classes of $\approx_k$, together with the functions $\Cons_k$ and $\textsc{Up}_k$
(except for $\textsc{Up}_0$, which makes no sense).
This can be done easily for $k=0$, because $\approx_0$ is a trivial relation consisting of a single class.

When we start the computation for some $k\geq 1$, we assume that the objects mentioned above are already computed for $k-1$.
Note first that whenever we have some new reachable class $\classC$ of $\approx_k$,
which means that we know $\Desc_k(c)$ for $c\in\classC$, then we can easily compute $\textsc{Up}_k(\classC)$.
Indeed, if $k=1$, then $\textsc{Up}_k(\classC)$ is always the only class of $\approx_0$, and if $k\geq 2$, 
recall that $\textsc{Up}_k(\classC)$ is a class described by $\Desc_{k-1}(c)$ for configurations $c\in\classC$, and
we have
\begin{align*}
	\Desc_{k-1}(c)=\set{(a,\textsc{Up}_{k-1}(\classD))\mid(a,\classD)\in\Desc_k(c)}\,.
\end{align*}

The function $\Cons_k$ is computed by considering configurations with stack height $n$, for consecutive $n=1,2,3,\dots$.
Formally, we define $\Cons_{k,n}$ to be $\Cons_k$ with domain restricted to tuples $(q,X,(\classk{r\alpha}{k})_{r\in\eats{X}(q)})$ for reachable configurations $qX\alpha$ satisfying $|X\alpha|\leq n$.
Note that the domain of $\Cons_{k,n}$ only becomes larger when $n$ increases.
For $n=0$ the domain of this function is empty (we cannot have $|X\alpha|\leq 0$).
For $n\geq 1$ we assume that $\Cons_{k,n-1}$ is already computed.
We first compute the domain of $\Cons_{k,n}$ using the following \lcnamecref{lem-cons}:

\begin{lemma}\label{lem-cons}
	The domain of $\Cons_{k,n}$ is the set of
	\begin{compactitem}
	\item	tuples $(q,X,()_{r\in\emptyset})$ with $\initstate\to^*qX$, and
	\item	tuples $(q,X,(\Cons_{k,n-1}(r,Y,(\classC_s)_{s\in\eats{Y}(r)}))_{r\in\eats{X}(q)})$ such that in the domain of $\Cons_{k,n-1}$ there is a tuple $(p',Y',(\classC_s)_{s\in\eats{Y'}(p')})$,
		and there is a transition $p'Y'\rewrite{a}pX'Y$, and $pX'\to^*qX$.
	\end{compactitem}
\end{lemma}

\begin{proof}
	By definition, the domain of $\Cons_{k,n}$ contains tuples $(q,X,(\classk{r\alpha}{k})_{r\in\eats{X}(q)})$ for reachable configurations $qX\alpha$ with $|X\alpha|\leq n$.
	Consider such a configuration.
	If $|X\alpha|=1$ (i.e., $\alpha=\epsilon$), we have $\initstate\to^*qX$; the tuple is added by the first item of the \lcnamecref{lem-cons}.
	Note that we have $\eats{X}(q)=\emptyset$ thanks to our assumption that no configurations with empty stack are reachable.
	
	Conversely, suppose that $2\leq|X\alpha|\leq n$,
	consider a run from $\initstate$ to $qX\alpha$, and consider the last configuration with stack of height $|\alpha|$ on this run.
	It is of the form $p'Y'\alpha'$, and $\alpha=Y\alpha'$ for some $Y$.
	Then we have a transition $p'Y'\rewrite{a}pX'Y$ leading to $pX'\alpha$.
	The remaining part of the run may be shifted to a run from $pX'$, so we have $pX'\to^*qX$.
	The configuration $p'Y'\alpha'$ is reachable, and $|Y'\alpha'|\leq n-1$, so the tuple $(p',Y',(\classC_s)_{s\in\eats{Y'}(p')})$ is in the domain of $\Cons_{k,n-1}$.
	Thus our original tuple fulfils the conditions of the second item of the \lcnamecref{lem-cons}.
	
	The proof of the opposite inclusion (i.e., that every tuple specified in the \lcnamecref{lem-cons} belongs to the domain of $\Cons_{k,n}$) is very similar, and thus it is left to the reader.
\end{proof}

Using standard means we can compute in time polynomial in $|\Pp|$:
\begin{compactitem}
\item	the set of pairs $pX,qY$ of configurations with stack of height $1$ such that $pX\to^*qY$, and
\item	the set $\eats{X}(p)$ for every configuration $pX$ with stack of height $1$.
\end{compactitem}
Thus, \cref{lem-cons} allows us to easily find the domain of $\Cons_{k,n}$ based on the previously computed function $\Cons_{k,n-1}$.

Next, we compute values of this function.
Take some tuple $(q,X,(\classC_r)_{r\in\eats{X}(q)})$ in its domain.
Imagine also a reachable configuration $qX\alpha$ such that $\classC_r=\classk{r\alpha}{k}$ for $r\in\eats{X}(q)$ and $|X\alpha|\leq n$
(the algorithm described below does not depend on $\alpha$, only its correctness proof does).
Notice that the height of the stack content $\alpha$ is at most $n-1$, so $\textsc{Up}_{k}(\classC_r)$ is already known.
We have two cases:
\begin{compactitem}
\item	If $(q,X)$ is in $\epsilon$-mode, then for some $r$ we have a popping transition $qX\rewrite{\epsilon}r$ (and $\eats{X}(q)=\set{r}$);
	we have $qX\alpha\approx_k r\alpha$, so we should take $\classC_r$ as the value of $\Cons_{k,n}$.
\item	Otherwise $(q,X)$ is not in $\epsilon$-mode.
	We compute the description $\Desc_k(qX\alpha)$ of the class $\classk{qX\alpha}{k}$ from definition,
	by listing all possible classes of $\approx_{k-1}$ that can be reached from $qX\alpha$ after reading a single letter.
	To this end, we consider all transitions starting in $(q,X)$:
	\begin{compactitem}
	\item	for transitions of the form $qX\rewrite{a}r$ we have $qX\alpha\weak{a}r\alpha$ and $\classk{r\alpha}{k-1}=\textsc{Up}_k(\classC_r)$;
	\item	for transitions of the form $qX\rewrite{a}sY$ we have $qX\alpha\weak{a}sY\alpha$ and
		\begin{align*}
			\classk{sY\alpha}{k-1}=\Cons_{k-1}(s,Y,\allowbreak(\textsc{Up}_k(\classC_r))_{r\in\eats{Y}(s)})
		\end{align*}
		(where $\eats{Y}(s)\subseteq\eats{X}(q)$);
	\item	for transitions of the form $qX\rewrite{a}tZY$ we have $qX\alpha\weak{a}tZY\alpha$ and
		\begin{align*}
			\classk{tZY\alpha}{k-1}=\Cons_{k-1}(t,Z,(\Cons_{k-1}(s,Y,(\textsc{Up}_k(\classC_r))_{r\in\eats{Y}(s)}))_{s\in\eats{Z}(t)})
		\end{align*}
		(where $\eats{Y}(s)\subseteq\eats{X}(q)$ for $s\in\eats{Z}(t)$).
	\end{compactitem}
\end{compactitem}

This finishes the description of how to compute $\Cons_{k,n}$ knowing $\Cons_{k,n-1}$ and $\Cons_{k-1}$.
Of course the values of $\Cons_{k,n}$ do not depend on $n$, only the domain of this function grows.
We can also see (cf.~\cref{lem-cons}) that if $\Cons_{k,n}$ has the same domain as $\Cons_{k,n-1}$, then the domain will not grow any more.
For such $n$ we have $\Cons_k=\Cons_{k,n}$; we can stop the computation.
Recall that we also stop the computation whenever we see more than $Z$ classes of $\approx_k$.
Thus, in the domain of $\Cons_{k,n}$ we may have at most $|Q|\cdot|\Gamma|\cdot Z^{|Q|}$ tuples,
and this is simultaneously a bound for $n$.
As already said, the considered values of $k$ are bounded by $Z$.
It follows that the running time is in $O(Z^{|\Pp|^{O(1)}})$, as declared.

\section{Lower bound}\label{sec:lower-full}

Our lower bound is shown via a reduction from the 
acceptance problem of exponentially
space-bounded alternating Turing machines. These are being introduced in Section~\ref{S ATM}.
Section~\ref{S Gadgets} recalls Defender's Forcing gadgets and introduces a gadget for checking if a certain prefix of the stack is in a given regular language.
Encodings of numbers and configurations are subject of Section~\ref{S Numbers} and Section~\ref{S Configurations},
respectively.
A gadget for verifying if a certain prefix of the stack consists of two consecutive configurations of
an exponentially space-bounded alternating Turing machine is introduced in Section~\ref{S Successor},
whereas Section~\ref{S Pushing} provides a gadget for pushing a successor configuration
on top of a configuration that can be found as the prefix of the stack.
The actual simulation of an exponentially space-bounded alternating Turing machine
is given in Section~\ref{S Simulation}. Building
upon these gadgets, Section~\ref{S Initial}
provides the final reduction from the acceptance problem
of exponentially space-bounded alternating Turing machines to 
the weak bisimulation finiteness problem for \PDS.

\subsection{Alternating Turing machines.}\label{S ATM}
An {\em alternating Turing machine} is a tuple $\M=(Q,\Upsilon,\Sigma,T,q_0,\Box)$,
where
\begin{compactitem}
	\item $Q=Q_\forall\uplus Q_\exists$ is a finite set of {\em states} that is partitioned into
		{\em universal states $Q_\forall$} and existential states $Q_\exists$,
	\item $\Upsilon$ is a finite {\em tape alphabet},
	\item $\Sigma\subseteq\Upsilon$ is an {\em input alphabet},
	\item $T\subseteq Q\times\Upsilon\times Q\times\Upsilon\times\{-1,1\}$ is
		a set of {\em transitions},
	\item $q_0\in Q$ is an initial state, and
	\item $\Box\in\Upsilon\setminus\Sigma$ is a {\em blank symbol}.
\end{compactitem}
A {\em configuration} of $\M$ is a word of the form $u(p,U)v$, where $u,v\in\Upsilon^*$
and $(p,U)\in Q\times\Upsilon$.
It is {\em universal} if $p\in Q_\forall$ and {\em existential}
if $p\in Q_\exists$.
For all $(p,U)\in Q\times\Upsilon$ we define
$T_{(p,U)}=\{(p,U,q,Y,d)\in T\mid q\in Q_\M, Y\in\Upsilon, d\in\{-1,1\}\}$.
By $\Configurations_\M$ we denote the set of configurations of $\M$.
For each pair $\mu=(t,Z)\in T\times\Upsilon$, where $t=(p,U,q,Y,d)$
we define the relation $\vdash_\mu\subseteq\Configurations_\M\times\Configurations_\M$,
where $w,w'\in\Configurations_\M$:
		\begin{eqnarray*}
			w\vdash_{\mu} w' &\Longleftrightarrow&
			\exists u,v\in\Upsilon^*.\, \begin{cases}w=uZ(p,U)v \text{ and } w'=u(q,Z)Yv & \text{if $d=-1$,}\\
				w=u(p,U)Zv \text{ and } w'=uY(q,Z)v& \text{if $d=1$.}
			\end{cases}
		\end{eqnarray*}
Note that if $w\vdash_{\mu} w$ and $w\vdash_{\mu}w''$,
then $w'=w''$.
We say that $w'$ a {\em successor configuration of $w$} if $w\vdash_\mu w'$
for some $\mu\in T\times\Upsilon$.
Since we use them for showing 2-\EXPTIME-hardness,
we assume without loss of generality that our alternating
Turing machines do not contain any cyclic computations, that is,
computations of the form
$w_{(0)}\vdash w_{(1)}\dots \vdash w_{(n)}$, where $n\geq 1$ and $w_{(0)}=w_{(n)}$.
This implies that we can do induction on the length of the longest computation starting in a given configuration.
We define
the set of {\em accepting configurations} as
the smallest set $C\subseteq\Configurations_\M$ that satisfies
the following properties:
\begin{compactitem}
\item For every existential configuration $w\in\Upsilon^*(p,U)\Upsilon^*$
	for which there exists some $t\in T_{(p,U)}$, some $Z\in\Upsilon$ and some configuration $w'\in C$
	with $w\vdash_{(t,Z)}w'$, we have $w\in C$.
\item For every universal configuration $w\in\Upsilon^*(p,U)\Upsilon^*$
	such that for all $t\in T_{(p,U)}$ there exists some $Z\in\Upsilon$ and
	some configuration $w'\in C$ with $w\vdash_{(t,Z)}w'$, we have $w\in C$.
\end{compactitem}
We note that a configuration $w\in\Upsilon^*(p,U)\Upsilon^*$ with $T_{(p,U)}=\emptyset$ is accepting if, and only if, it is universal.

Given a function $f\colon\N\rightarrow\N$ with $f(n)\geq n$ for all $n\in\N$,
the language of an $f$-space-bounded alternating Turing machine $\M$ is given by
$$
L(\M)=\{x_0\dots x_{n-1}\in\Sigma^n\mid n\geq 1, (q_0,x_0)x_1\dots x_{n-1}\Box^{f(n)-n}\text{ is accepting}\}.
$$

\begin{theorem}
	\problem is 2-\EXPTIME-hard under
	polynomial time reductions.
\end{theorem}

Fix any language $L$ in 2-\EXPTIME.
Fix a $2^{\ppoly(n)}$-space-bounded alternating Turing machine
$\M=(Q_\M,\Upsilon,\Sigma_\M,T,q_0,\Box)$
such that $L=L(\M)$, where $\ppoly$ is a polynomial (it is folklore that such a machine exists for every language in 2-\EXPTIME).
Moreover fix some input
$x=x_0\dots x_{n-1}\in\Sigma_\M^n$.
We construct in time polynomial
in $n$
an \PDS\ $\Pp=(Q,\Gamma,\Act,\Delta)$ and a configuration
$q_\updownarrow\#\in Q\Gamma$ of
$\Pp$ such that $x\in L(\M)$ if, and only if,
$(\Ll(\Pp),q_\updownarrow\#)$ is {\em not} weakly bisimulation finite.
This is sufficient since 2-\EXPTIME is closed under complement.

In the following let us identify $\ppoly$ with $\ppoly(n)$
and let $N=2^{\ppoly}$.
We assume without loss of generality that $\ell\geq 2$ and hence $N-1>n$.
We set
$\Gamma=\{\#,0,1\}\cup\Upsilon\cup(Q_\M\times\Upsilon)\cup\{\#_i\mid
i\in\interval{0,\ppoly-1}\}$.

Note that configurations of $\M$ needed for accepting inputs of length $n$ are words from the language
$
\bigcup_{k=0}^{N-1}\Upsilon^k (Q_\M\times\Upsilon) \Upsilon^{N-1-k}.
$
For every such configuration $w=w_0\dots w_{N-1}$
we write $\pos(w)$ to denote the unique $k\in\interval{0,N-1}$
such that $w_k\in Q_\M\times\Upsilon$,
hereby denoting the position of the read/write head of the configuration $w$.

The following lemma is an immediate consequence of the definition of $\vdash_\mu$:

\begin{lemma}\label{lemma vdash}
	For all $\mu=(t,Z)\in T\times\Upsilon$ with $t=(p,U,q,Y,d)$ and
	all configurations $w$ and $w'$ of $\M$
	we have
	$w\vdash_\mu w'$ if, and only if,
	$\Pos(w')=\Pos(w)+d$ and $h_\mu(w)=h_\mu'(w')$,
	where $h_\mu,h_\mu'\colon(\Upsilon\cup(Q_\M\times\Upsilon))^*\rightarrow(\Upsilon\cup(Q_\M\times\Upsilon))^*$
	are the letter-to-letter morphisms such that for all $X\in\Upsilon\cup(Q_\M\times\Upsilon)$
	we have
	\begin{align*}
	h_\mu(X)=
	\begin{cases}
		Y & \text{if $X=(p,U)$,}\\
		X & \text{otherwise,}
	\end{cases}
	\qquad\text{and}\qquad
	h_\mu'(X)=
	\begin{cases}
		Z & \text{if $X=(q,Z)$,}\\
		X & \text{otherwise.}
	\end{cases}
	\raisebox{-2ex}{\qedhere}
	\end{align*}
\end{lemma}

Throughout Section~\ref{sec:lower-full} we constantly add fresh rules; they implicitly
come with fresh control states and action symbols.
Whenever we introduce such rules, we prove
that certain configurations, whose control states appear in those
rules, have certain properties.
These properties however are not influenced by rules that are
introduced later.
In particular, if we describe rules starting in some pair $(q,x)\in Q\times\Gamma$,
then the reader may assume that no other rules starting in this pair are introduced later in the text.
Before we introduce these rules we will introduce some gadgets
in the next section.

\subsection{Gadgets for the lower bound construction.}\label{S Gadgets}

Inspired by the ``Defender's forcing'' technique of Jan\v{c}ar and Srba~\cite{JS08},
given a finite set of control states
$\{\Source,\Source'\}\cup\{\Target_i,\Target_i'\mid i\in I\}$ and a stack symbol $X$,
we introduce a notation
$$\langle \Source,\Source'\rangle X \rewrite{\Defender}
\{\langle \Target_i,\Target_i'\rangle X\mid i\in I\}$$
to denote the following set of rules,
where all action symbols and control states that are not in
$\{\Source,\allowbreak\Source'\}\allowbreak\cup\{\Target_i,\allowbreak\Target_i'\mid i\in I\}$
are newly introduced and
where $i,j\in I$:
\begin{align*}
	&\Source X\rewrite{a}\Wait X\,,&&\Source'X\rewrite{a}
	\Choose_i'X\,,\\
	&\Source X\rewrite{a}\Choose_i'X\,,&&\\
	&\Wait X\rewrite{a_i}\Target_i X\,,&&
	\Choose_i'X\rewrite{a_i}\Target_i'X\,,\\
	&
	&&\Choose_i' X\rewrite{a_j}\Target_j X\quad\text{if $i\not= j$.}
\end{align*}
The wit of the gadget, justified in details in Jan\v{c}ar and Srba~\cite[Section 4]{JS08}, is that for all stack contents
$\gamma\in\Gamma^*$ we have
$$
\Source X\gamma\approx\Source'X\gamma
\quad\Longleftrightarrow\quad \exists i\in I.\,
\Target_iX\gamma\approx\Target_i'X\gamma\,.
$$

Analogously (by adding appropriate push rules at the
end), given a set of control states $\{\Source,\Source'\}\cup\{\Target,\Target'\}$,
a set of stack symbols $\{X\}\cup\{Y_i\in\Gamma\mid i\in I\}$,
one can construct a gadget
$$\langle \Source,\Source'\rangle X\rewrite{\Defender}
\langle \Target,\Target'\rangle\ \{Y_iX\mid i\in I\}$$
such that for all stack contents $\gamma\in\Gamma^*$ we have
$\Source X\gamma\approx\Source'X\gamma$ if, and only if,
$\Target Y_iX\gamma\approx\Target'Y_iX\gamma$
for some $i\in I$.
In any case, both gadgets have size $O(|I|)$.

Going further, given a finite language $L\subseteq\Gamma^*$
of the form
$L=\Sigma_1\dots\Sigma_m$, where $\Sigma_j\subseteq\Gamma$ for all $j\in\interval{1,m}$,
by suitably cascading the above gadget $m$ times, one can
construct a gadget
$$\langle \Source,\Source'\rangle X \rewrite{\Defender}
\langle \Target,\Target'\rangle\ L$$
such that for all stack contents $\gamma\in\Gamma^*$ we have
$\Source X\gamma\approx\Source'X\gamma$ if, and only if,
there exists some $w\in L$ such that
$\Target w\gamma\approx\Target'w\gamma$.

A {\em deterministic finite automaton (DFA)}
is a tuple $A=(S,\Sigma,\delta_A,s_0,F)$,
where $S$ is a finite set of {\em states}, $\Sigma$ is a finite
{\em alphabet}, $\delta_A:S\times\Sigma\rightarrow S$
is the {\em transition function}, $s_0\in S$ is the
{\em initial state}, and $F\subseteq S$ is a set of {\em final states}.
The function $\delta_A$ is naturally extended to
a function from $S\times\Sigma^*$
inductively via $\delta_A(s,\varepsilon)=s$
and $\delta_A(s,aw)=\delta_A(\delta_A(s,a),w)$ for all $a\in\Sigma$
and all $w\in\Sigma^*$.
A language $L\subseteq\Sigma^*$ is {\em regular} if
$L=L(A)=\{w\in\Sigma^*\mid \delta_A(s_0,w)\in F\}$
for some DFA $A=(S,\Sigma,\delta_A,s_0,F)$.
Finally, given a regular language $L\subseteq\Sigma^*$,
where $\Sigma\subseteq\Gamma\setminus\{\#\}$,
and two control
states $p$ and $p'$ we introduce a gadget
$$\langle p,p'\rangle \#\rewrite{\textsc{Prefix-Check}\# L\#}$$
such that for all stack contents of the form
$\gamma=\#\theta\#\beta$, where $\theta\in\Sigma^*$ and
$\beta\in\Gamma^*$,
we have
$p\gamma\approx p'\gamma$ if, and only if,
$\theta\in L$.
Assuming some DFA $A=(S,\Sigma,\delta_A,s_0,F)$ such that
$L(A)=L$, we construct
the gadget as follows, where all symbols and control states
(except for $p$ and $p'$) are
freshly introduced, where $s\in S$,
where $s'$ is a copy of every such $s\in S$,
and where $a\in \Sigma$:
\begin{align*}
	&p\#\rewrite{1}s_0\,,&&p'\#\rewrite{1}s_0'\,,\\
	&sa\rewrite{\varepsilon}\delta_A(s,a)\,,&&
	s'a\rewrite{\varepsilon}\delta_A(s,a)'\,,\\
	&s\#\rewrite{2}t\quad\text{if $s\in F$,}&&
	s'\#\rewrite{2}t'\quad\text{if $s\in F$,}\\
	&s\#\rewrite{3}t\quad\text{if $s\not\in F$.}&&
\end{align*}
Since the above rules only allow to execute
runs reading at most two action symbols,
note that, by construction, both $\{\class{p\gamma}\mid \gamma\in\Gamma^*\}$
and $\{\class{p'\gamma}\mid \gamma\in\Gamma^*\}$ are
finite sets of classes all of which are weakly bisimulation finite.
For all of the prefix checking rules that we will introduce below where, say, $L$ is an involved
regular language, it holds implicitly that one can compute in polynomial time in $n=|x|$
a DFA $A$ with $L(A)=L$.

\subsection{Encoding numbers.}\label{S Numbers}

The {\em binary presentation} of a number $k\in\interval{0,N-1}$ is
defined as
$$
\Bin_k=\Bin_{k,\ppoly-1}\dots\Bin_{k,0},
\text{ where }\Bin_{k,i}\in\{0,1\} \text{ are such that }
k=\sum_{i=0}^{\ppoly-1}2^i\cdot\Bin_{k,i}\,.$$
In particular, note that $\Bin_0=0^{\ppoly}$ and
$\Bin_{N-1}=\Bin_{2^{\ppoly}-1}=1^{\ppoly}$.
Conversely, for each $u\in\{0,1\}^\ell$ let $\overline{u}\in\N$ be the unique $k\in\interval{0,N-1}$ such that
$\Bin(k)=u$.
For any two $k,k'\in\interval{0,N-1}$ with $k\not= k'$, we define $
\MSB(k,k')=\max\{i\in\interval{0,\ppoly-1}\mid \Bin_{k,i}\not=\Bin_{k',i}\}$ to be
the most significant bit in which the binary
presentations of $k$ and $k'$ differ.

In this subsection we present rules allowing us to verify whether one number is a successor of another number.
We start by adding the following set of rules, where $i,j\in\interval{0,\ppoly-1}$
and $b\in\{0,1\}$:
\begin{align*}
	&\Binsucc_{i}b\rewrite{s}\Binsucc_{i,\ppoly-1} b\,,  & &
	\widetilde{\Binsucc} b\rewrite{s}\widetilde{\Binsucc}_{\ppoly-1} b\,,\\
	& \Binsucc_{i,j} b\rewrite{b}\Binsucc_{i,j-1}\quad\text{if $j>i$,}&&\\
	& \Binsucc_{i,j} 0\rewrite{1}
	\Binsucc_{i,j-1}\quad\text{if $j=i$,}&&
	\widetilde{\Binsucc}_{j} b\rewrite{b}
	\widetilde{\Binsucc}_{j-1}\,.\\
	& \Binsucc_{i,j} 1\rewrite{0}\Binsucc_{i,j-1}\quad\text{if $j<i$,}&&
\end{align*}

Observe that there is no rule involving $\Binsucc_{i,-1}$ nor
$\widetilde{\Binsucc}_{-1}$ on its left-hand side.
We have the following characterization:

\begin{lemma}\label{lemma succ i}
	For all $i\in\interval{0,\ppoly-1}$, all $k,k'\in\interval{0,N-1}$, and all stack contents $\gamma,\gamma'\in\Gamma^*$ we have
	$\Binsucc_{i}\Bin_k\gamma\approx \widetilde{\Binsucc}\Bin_{k'}\gamma'$
	if, and only if, $k'=k+1$ and
	$i=\MSB(k,k')$.
	Moreover, both
	$\{\class{\Binsucc_{i}\gamma}\mid \gamma\in\Gamma^*\}$ and
	$\{\class{\widetilde{\Binsucc}\gamma}\mid \gamma\in\Gamma^*\}$
	are finite sets of classes all of which are weakly bisimulation finite.
\end{lemma}

\begin{proof}
	The first statement of the lemma follows
	immediately from inspection of the above rules
	and the equivalence of the following two statements:
	\begin{compactitem}
	\item $k'=k+1$ and $i=\MSB(k,k+1)$;
	\item the following three conditions hold
		for all $j\in\interval{0,\ell-1}$:
		\begin{compactitem}
		\item $\Bin_{k,j}=\Bin_{k',j}$
			if $j>i$,
		\item ($\Bin_{k',j}=1$ and $\Bin_{k,j}=0$) if $j=i$, and
		\item ($\Bin_{k,j}=1$ and $\Bin_{k',j}=0$)
			if $j<i$.
		\end{compactitem}
	\end{compactitem}
	The second statement follows from inspection of the
	above rules, which imply that the weak bisimulation class both
	of $\Binsucc\gamma$
	and of $\widetilde{\Binsucc}\gamma'$ is determined by the first $\ell$
	letters of $\gamma$.
\end{proof}

We now add a few more rules.
Assuming that the top of the stack is
of the form $\Bin_{k}\#_iX\Bin_{k'}$
for some $X\in\Upsilon\cup(Q_\M\times\Upsilon)$,
these rules
allow us to verify,
starting from two control states $\Binsucc$ and $\Binsucc'$,
whether it holds that $k'=k+1$
and $i$ is the most significant bit position in which
$k$ and $k'$ differ.
In the following rules we have
$X\in\Upsilon\cup(Q_\M\times\Upsilon)$, $i\in \interval{0,\ppoly-1}$,
and $b\in\{0,1\}$:
\begin{align*}
	&\Binsucc b\rewrite{a}\widetilde{\Binsucc}b\,,&&\Binsucc'b\rewrite{a}\Temp\,,\\
	& &&\Temp b\rewrite{\varepsilon}\Temp\,,\\
	& &&\Temp \#_i\rewrite{\varepsilon}\Temp_i\,,\\
	&&&\Temp_i X\rewrite{\varepsilon}\Binsucc_i\,.
\end{align*}

\begin{lemma}\label{lemma succ}
	For all $k,k'\in\interval{0,N-1}$, all $i\in\interval{0,\ppoly-1}$, all
	$X\in\Upsilon\cup(Q_\M\times\Upsilon)$, and all $\gamma\in\Gamma^*$
	we have
	$\Binsucc\Bin_k\#_iX\Bin_{k'}\gamma\approx
	\Binsucc'\Bin_k\#_iX\Bin_{k'}\gamma$
	if, and only if, $k'=k+1\text{ and }\MSB(k,k')=i$.
	Moreover, both $\{\class{\Binsucc\gamma}\mid\gamma\in\Gamma^*\}$
	and $\{\class{\Binsucc'\gamma}\mid\gamma\in\Gamma^*\}$ are finite sets
	of classes all of which are weakly bisimulation finite.
\end{lemma}

\begin{proof}
	Looking at the first rules what will be applied (yielding
	an $a$-labeled transition,
	possibly followed by a sequence of $\varepsilon$-transitions) we see that
	$\Binsucc\Bin_k\#_iX\Bin_{k'}\gamma\approx
	\Binsucc'\Bin_k\#_iX\Bin_{k'}\gamma$
	if, and only if,
	$\widetilde{\Binsucc}\Bin_{k}\#_iX\Bin_{k'}\gamma\approx
	\Binsucc_i\Bin_{k'}\gamma$,
	which by \cref{lemma succ i} holds
	if, and only if,
	$k'=k+1$ and $i=\MSB(k,k')$.
	That the sets
	$\{\class{\Binsucc\gamma}\mid\gamma\in\Gamma^*\}$
	and
	$\{\class{\Binsucc'\gamma}\mid\gamma\in\Gamma^*\}$
	are both finite sets of classes all of which are weakly
	bisimulation finite
	follows immediately from inspection of
	the above rules and \cref{lemma succ i}.
\end{proof}

\subsection{Encoding configurations.}\label{S Configurations}

The encoding of a configuration $w=w_0\dots w_{N-1}$ of $\M$ is defined to be
the following word $\Enc_w\in\Gamma^*$:
$$
\Enc_w=w_0\Bin_0\left(\prod_{i=1}^{N-1}\#_{\MSB(i-1,i)}w_i\Bin_{i}\right).
$$

\begin{Example}
	For $\ppoly=3$ (hence for $N=2^\ppoly=8$) and $w=ab(q,c)abcbb$
	we have
	$$
	\Enc_w=a000\#_0b001\#_1(q,c)010\#_0a011\#_2b100\#_0c101\#_1b110
	\#_0b111\,.
	$$
\end{Example}

We remark that $\Enc_w$ contains precisely $2^{\ppoly-1-i}$ occurrences of
$\#_i$, for every $i\in\interval{0,\ppoly-1}$.

\begin{Definition}\label{definition alpha w}
	For all configurations $w=w_0\dots w_{N-1}$ of $\M$ and
	for all $u\in\{0,1\}^{\leq \ppoly}$ let us define,
	by induction on $|u|$, the infix $\Enc^{(u)}_w$ of $\Enc_w$:
	\begin{compactitem}
	\item $\Enc^{(\varepsilon)}_{w}=\Enc_w$, and
	\item if $u\in\{0,1\}^i$ for some $i\in\interval{0,\ppoly-1}$,
		then $\Enc^{(u0)}_w$ and $\Enc^{(u1)}$ are the unique
		infixes of $\Enc^{(u)}$ satisfying
		$\Enc^{(u)}_w=\Enc^{(u0)}_w\#_{\ppoly-1-i}\Enc^{(u1)}_w$.
	\end{compactitem}
\end{Definition}

Recalling that $\overline{u}$ denotes the unique number that the binary string
$u\in\{0,1\}^\ppoly$ encodes, for all $u\in\{0,1\}^\ppoly$ we have
$\Enc^{(u)}_w=w_{\overline{u}}\Bin_{\overline{u}}$.
A simple induction yields that for all $i\in\interval{0,\ppoly}$ and
all $u\in\{0,1\}^i$, the infix $\Enc^{(u)}_w$
appears exactly once in $\Enc_w$ and,
for all $j\in\interval{0,\ppoly-1}$ we have that $\Enc^{(u)}_w$ contains precisely
$\lfloor 2^{\ppoly-1-i-j}\rfloor$ occurrences of $\#_{j}$.

\subsection{Gadget for checking the successor relation among configurations.}\label{S Successor}

In this subsection we show how to check whether the successor relation holds between two configurations written on the top of the stack.
First, we construct a gadget verifying whether head positions agree;
more precisely, given $d\in\{-1,1\}$, and assuming that the stack starts with $\#\Enc_{w'}\#\Enc_w$, the gadget allows us to test
whether $\Pos(w')=\Pos(w)+d$.
To this end, for all $d\in\{-1,1\}$, all $i\in\interval{0,\ell-1}$, and all $X\in\Gamma$ we add the following rules:
\begin{align}
	&&&\hspace{-13em}\langle \Pos_d,\Pos_d'\rangle\# \rewrite{\Defender}
	\{\langle \Pos_{d,j}, \Pos_{d,j}'\rangle \#\mid j\in\interval{0,\ppoly-1}\}\,,\label{Pos 1}\\
	&\Pos_{d,i}X\rewrite{\varepsilon}\Pos_{d,i}\quad\text{if $X\not\in Q_\M\times\Upsilon$,}&&
	\Pos_{d,i}'X\rewrite{\varepsilon}\Pos_{d,i}'
	\quad\text{if $X\not\in Q_\M\times\Upsilon$,}\label{Pos 2}\\
	&\hspace{0.1em}&&\Pos_{d,i}'X\rewrite{\varepsilon}\Pos_{d,i}''\quad\text{if $X\in Q_\M\times\Upsilon$,}\label{Pos 3}\\
	&\hspace{0.1em}&&\Pos_{d,i}''X\rewrite{\varepsilon}
	\Pos_{d,i}''\quad\text{if $X\not\in Q_\M\times\Upsilon$,}\label{Pos 4}\\
	&\Pos_{d,i}X\rewrite{\varepsilon}\Binsucc_{i}\quad\text{if $X\in Q_\M\times\Upsilon,d=\!-1$,}&&
	\Pos_{d,i}''X\rewrite{\varepsilon}\widetilde{\Binsucc}\quad\text{if $X\in Q_\M\times\Upsilon,d=\!-1$,}\label{Pos 5}\\
	&\Pos_{d,i}X\rewrite{\varepsilon}\widetilde{\Binsucc}\quad\text{if $X\in Q_\M\times\Upsilon,d=1$,}&&
	\Pos_{d,i}''X\rewrite{\varepsilon}\Binsucc_{i}\quad
	\text{if $X\in Q_\M\times\Upsilon,d=1$.}\label{Pos 6}
\end{align}

\begin{lemma}\label{lemma pos}
	For all $d\in\{-1,1\}$, all length-$N$ configurations $w,w'$ of $\M$, and all $\gamma\in\Gamma^*$ we have
	$$
		\Pos_d\#\Enc_{w'}\#\Enc_{w}\gamma\approx
		\Pos_d'\#\Enc_{w'}\#\Enc_{w}\gamma
			\quad\Longleftrightarrow\quad
		\pos(w')=\pos(w)+d\,.$$
	Moreover, both
	$\{\class{\Pos_d\gamma}\mid\gamma\in\Gamma^*\}$
	and
	$\{\class{\Pos_d'\gamma}\mid\gamma\in\Gamma^*\}$
	are finite sets of classes all of which are weakly
	bisimulation finite.
\end{lemma}

\begin{proof}
	Let $d\in\{-1,1\}$, let $w'=w_0'\dots w_{N-1}'$, $w=w_0\dots w_{N-1}$,
	$k'=\Pos(w')$, and $k=\Pos(w)$.
	We then have $w'_{k'},w_{k}\in Q_\M\times\Upsilon$.
	By the definition on an encoding we have
	$$\Enc_{w'}=w_0'\Bin_0\left(\prod_{i=1}^{N-1}\#_{\MSB(i-1,i)}w_{i}'\Bin_{i}\right)\quad\text{and}\quad
	\Enc_{w}=w_0\Bin_0\left(\prod_{i=1}^{N-1}\#_{\MSB(i-1,i)}w_{i}\Bin_{i}\right).$$
	Let $\gamma\in\Gamma^*$ and let us define $\delta=\#\Enc_{w'}\#\Enc_{w}\gamma$.
	We provide a proof only in the case of $d=-1$; the case of $d=1$ is completely analogous.
	By inspecting rules in \crefrange{Pos 2}{Pos 6}, for all
	$i\in\interval{0,\ppoly-1}$ we have
	$\Pos_{d,i}\delta\gamma\approx\Pos_{d,i}'\delta\gamma$
	if, and only if,
	$$\Binsucc_i \Bin_{k'}\#_{\MSB(k',k'+1)}\dots w_{N-1}'\Bin_{N-1}\#\Enc_{w}\gamma\approx
	\widetilde{\Binsucc} \Bin_{k}\#_{\MSB(k,k+1)}\dots w_{N-1}\Bin_{N-1}\gamma\,.
	$$
	By \cref{lemma succ i} the latter equivalence holds if, and only if,
	$k'=k+1$ and $i=\MSB(k,k')$.
	By properties of the Defender's forcing gadget used in \cref{Pos 1} we have
	$\Pos_d\delta\gamma\approx\Pos_d'\delta\gamma$ if, and only if,
	there exists some $i\in\interval{0,\ppoly-1}$ such that
	$\Pos_{d,i}\delta\gamma\approx\Pos_{d,i}'\delta\gamma$.
	Thus,
	$\Pos_d\delta\gamma\approx\Pos_d'\delta\gamma$ if, and only if, there exists
	some $i\in\interval{0,\ppoly-1}$ such that $k'=k+1$ and $i=\MSB(k,k')$,
	that is, if, and only if, $\Pos(w')=\Pos(w)+d$, as required.

	Let us finally prove that $\{\class{\Pos_d\gamma}\mid\gamma\in\Gamma^*\}$
	is a finite set of classes all of which are weakly bisimulation finite;
	the same can analogously be proven for
	$\{\class{\Pos_d'\gamma}\mid\gamma\in\Gamma^*\}$.
	By inspecting the above rules, there is a constant $c\in\N$
	(essentially depending on the Defender's forcing gadget)
	such that for all stack contents $\gamma\in\Gamma^*$, when applying the rules
	from a configuration of the form
	$\Pos_d\gamma$
	one can only execute a run reading at most $c$ action symbols and
	for some $\gamma'\in\Gamma^*$ either
	get stuck in a dead end of the form $\Pos_{d,i}\gamma'$,
	or reach a configuration of the form $\Binsucc_i\gamma'$
	or $\widetilde{\Binsucc}\gamma'$.
	In the former case, finiteness is clear,
	in the latter two cases finiteness
	immediately follows from the fact that
	$\{\class{\Binsucc_i\gamma}\mid\gamma\in\Gamma^*\}$
	is a finite set of classes all of which are weakly bisimulation
	finite
	for all $i\in\interval{0,\ppoly-1}$
	and the same holds for the set
	$\{\class{\widetilde{\Binsucc}\gamma}\mid\gamma\in\Gamma^*\}$
	according to \cref{lemma succ i}.
\end{proof}

Once head positions are verified, the remaining part of the successor relation may be checked with a help of letter-to-letter morphisms, as explained in \cref{lemma vdash}.
This is realized by our next gadget, which
for each $\mu\in T\times\Upsilon$ allows to verify,
assuming that the top of the stack is $\#\Enc_{w'}\#\Enc_w$,
whether $w\vdash_\mu w'$ holds.
In order to construct this gadget,
for all $\mu=(t,Z)\in T\times\Gamma$ with $t=(p,U,q,Y,d)$,
all $X\in\Gamma$, and all $i\in\interval{0,\ppoly-1}$
we add the following rules, where rules involving the control states
$\Pos_d$ and $\Pos_d'$ have already been introduced above
and where the morphisms $h_{\mu}$ and $h_{\mu}'$ are
defined in \cref{lemma vdash}:
\begin{align}
	&\Succ_{\mu}\#\rewrite{0}\Pos_{d}\#\,,  & &
	\Succ_{\mu}'\#\rewrite{0}\Pos_{d}'\#\,,\label{rule 7}\\
	&\Succ_{\mu}\#\rewrite{1}\Desc_{\mu,\ppoly}\,,  & &
	\Succ_{\mu}'\#\rewrite{1}\Desc_{\mu,\ppoly-1}'\,,\label{rule 8}\\
	&\Desc_{\mu,\ppoly}X\rewrite{\varepsilon}\Desc_{\mu,\ppoly}
	\quad\text{if $X\not=\#$,}
	&&\label{rule 9}\\
	&\Desc_{\mu,\ppoly}\#\rewrite{\varepsilon}\Desc_{\mu,\ppoly-1}\,,
	&&\label{rule 10}\\
	&\Desc_{\mu,i}X\rewrite{0}\Desc_{\mu,i-1}X\,,
	& &\Desc_{\mu,i}'X\rewrite{0}\Desc_{\mu,i-1}'X\,,\label{rule 11}\\
	&\Desc_{\mu,i}X\rewrite{1}\Pop_{\mu,i}X\,,
	& &\Desc_{\mu,i}'X\rewrite{1}\Pop_{\mu,i}'X\,,\label{rule 12}\\
	&\Pop_{\mu,i}X\rewrite{\varepsilon}\Pop_{\mu,i}
	\quad\text{if $X\not=\#_i$,}
	&&
	\Pop_{\mu,i}'X\rewrite{\varepsilon}\Pop_{\mu,i}'
	\quad\text{if $X\not=\#_i$,}
	\label{rule 13}\\
	&\Pop_{\mu,i}\#_i\rewrite{\varepsilon}\Desc_{\mu,i-1}\,,
	&&
	\Pop_{\mu,i}'\#_i\rewrite{\varepsilon}\Desc_{\mu,i-1}'\,,\label{rule 14}\\
	&\Desc_{\mu,-1}X\rewrite{h_{\mu}(X)}\Final_\mu X\,,&&
	\Desc_{\mu,-1}'X\rewrite{h_{\mu}'(X)}\Final_\mu'X\,.\label{rule 15}
\end{align}

\begin{lemma}\label{lemma Succ mu}
	For all length-$N$ configurations $w,w'$ of $\M$, all $\mu\in T\times\Upsilon$, and
	all $\gamma\in\Gamma^*$ we have
	$$
		\Succ_{\mu}\#\Enc_{w'}\#\Enc_{w}\gamma\approx
		\Succ_{\mu}'\#\Enc_{w'}\#\Enc_{w}\gamma
		\quad\Longleftrightarrow\quad
		w\vdash_{\mu} w'\,.$$
	Moreover, both $\{\class{\Succ_{\mu}\gamma}\mid
	\gamma\in\Gamma^*\}$ and
	$\{\class{\Succ_{\mu}'\gamma}\mid
	\gamma\in\Gamma^*\}$
	are finite sets of classes all of which are weakly bisimulation finite.
\end{lemma}

\begin{proofnoqed}
	Let us first prove the second statement of the lemma.
	Consider the following (smallest) partial order $\succ$
	relating the above control states as follows:
	\begin{align*}
	&\Succ_\mu\hspace{-0.5em}
		\makebox[0pt][l]{\,\raisebox{-1.8ex}{\rotatebox{-45}{$\succ$}}\raisebox{-4ex}{$\Pos_\mu$}} % below
		\raisebox{1.5ex}{\rotatebox{45}{$\succ$}}\!\raisebox{3.4ex}{$\Desc_{\mu,\ppoly}$}\raisebox{2.5ex}{\rotatebox{-45}{$\succ$}} % above
		\hspace{-0.2em}
		\Desc_{\mu,\ppoly-1}\succ\Pop_{\mu,\ppoly-1}\succ\Desc_{\mu,\ppoly-2}\succ\Pop_{\mu,\ppoly-2}\succ
		\dots\succ\Desc_{\mu,-1}\succ\Final_\mu\,,\displaybreak[0]\\[1.5ex]
	&\Succ_\mu'
		\makebox[0pt][l]{\raisebox{-2ex}{\rotatebox{-45}{$\succ$}}\raisebox{-4.4ex}{$\Pos_\mu'$}}
		\succ\Desc_{\mu,\ppoly-1}'\succ\Pop_{\mu,\ppoly-1}'\succ\Desc_{\mu,\ppoly-2}'\succ\Pop_{\mu,\ppoly-2}'\succ
		\dots\succ\Desc_{\mu,-1}'\succ\Final_\mu'\,.
	\end{align*}
	For all rules that appear in \crefrange{rule 7}{rule 15}
	in which, say, $p_1$ is the control state on the left-hand
	side of the rule and $p_2$ is the control state on the right-hand side,
	observe that either $p_1=p_2$ or $p_1\succ p_2$.
	Furthermore, if the rule is a reading rule
	(i.e., is not an $\varepsilon$-rule) we have $\ell\succ r$.
	Since moreover any configuration having control state $\Final_\mu$
	or $\Final_\mu'$ is a dead end,
	and recalling that both $\{\class{\Pos_\mu\gamma}\mid \gamma\in\Gamma^*\}$
	and $\{\class{\Pos_\mu'\gamma}\mid \gamma\in\Gamma^*\}$
	are finite sets of classes all of which are weakly
	bisimulation finite by \cref{lemma pos},
	it follows that both $\{\class{\Succ_\mu\gamma}\mid\gamma\in\Gamma^*\}$
	and $\{\class{\Succ_\mu'\gamma}\mid\gamma\in\Gamma^*\}$
	are indeed finite sets of classes all of which are weakly bisimulation
	finite.

	Let us now prove the first statement of the lemma.
	Let $\mu=(t,Z)$ with $t=(p,U,q,Y,d)$ and let us introduce the abbreviation
	$\delta=\#\Enc_{w'}\#\Enc_{w}$.
	By inspecting \cref{rule 7,rule 8} we have
	$\Succ_\mu\delta\gamma\approx\Succ_\mu'\delta\gamma$ if, and only
	if, $\Pos_d\delta\gamma\approx\Pos_d'\delta\gamma$
	and $\Desc_{\mu,\ppoly}\Enc_{w'}\#\Enc_{w}\gamma\approx
	\Desc_{\mu,\ppoly-1}'\Enc_{w'}\#\Enc_{w}\gamma$.
	Since, on the one hand, $\Pos_d\delta\gamma\approx\Pos_d'\delta\gamma$ if, and only
	if $\Pos(w')=\Pos(w)+d$ by \cref{lemma pos}
	and, on the other hand,
	$w\vdash_\mu w'$ if, and only
	if, $\Pos(w')=\Pos(w)+d$ and $h_{\mu}(w)=h_{\mu}'(w')$ by \cref{lemma vdash},
	for proving the first statement of the lemma
	it is enough to prove that $\Desc_{\mu,\ppoly}\Enc_{w'}\#\Enc_{w}\gamma\approx\allowbreak
	\Desc_{\mu,\ppoly-1}'\Enc_{w'}\#\Enc_{w}\gamma$
	if, and only if, $h_{\mu}(w)=h_{\mu}'(w')$.
	Hence, by inspection of \cref{rule 9,rule 10},
	it is sufficient to prove that
	\begin{equivalences}
		\Desc_{\mu,\ppoly-1}\Enc_{w}\gamma\approx
		\Desc_{\mu,\ppoly-1}'\Enc_{w'}\#\Enc_{w}\gamma\quad\Longleftrightarrow\quad
		h_\mu(w)=h_\mu'(w')\,.\label{L sufficient}
	\end{equivalences}

	Let $w=w_0\dots w_{N-1}$ and $w'=w_0'\dots w_{N-1}'$.
	For all $u\in\{0,1\}^{\leq\ppoly}$
	we have defined an infix $\Enc^{(u)}_{w}$ of $\Enc_{w}$
	and an infix $\Enc^{(u)}_{w'}$ of $\Enc_{w'}$
	(cf.\@ \cref{definition alpha w}); recall that
	$\Enc^{(u)}_{w}$ (and $\Enc^{(u)}_{w'}$) appears precisely once in
	$\Enc_{w}$ (in $\Enc_{w'}$, respectively).
	Let $\delta^{(u)}_{w}$ be the unique suffix of $\Enc_{w}$
	starting just after $\Enc^{(u)}_{w}$ (then $\Enc^{(u)}_{w}\delta^{(u)}_{w}$ is a suffix of $\Enc_{w}$),
	and analogously, let $\delta^{(u)}_{w'}$ be the unique suffix of
	$\Enc_{w'}$ starting just after $\Enc^{(u)}_{w'}$, (then $\Enc^{(u)}_{w'}\delta^{(u)}_{w'}$ is a suffix of $\Enc_{w'}$).
	For all $i\in\interval{0,\ppoly-1}$ and
	all $u\in\{0,1\}^i$ we have
	that both $\Enc^{(u)}_{w}$ and $\Enc^{(u)}_{w'}$
	contain precisely one occurrence of $\#_{\ell-1-i}$; we have
	$\Enc^{(u)}_{w}=\Enc^{(u0)}_{w}\#_{\ell-1-i} \Enc^{(u1)}_{w}$,
	and analogously
	$\Enc^{(u)}_{w'}=\Enc^{(u0)}_{w'}\#_{\ell-i-1} \Enc^{(u1)}_{w'}$
	by \cref{definition alpha w}.
	As a consequence, by inspection of the rules from \crefrange{rule 11}{rule 14},
	for all $i\in\interval{0,\ppoly-1}$, all $u\in\{0,1\}^i$, and
	all $b\in\{0,1\}$ we have
	$\Desc_{\mu,\ell-1-i}\Enc^{(u)}_{w}\delta^{(u)}_{w}\gamma\weak{b}\Desc_{\mu,\ell-2-i}
	\Enc^{(ub)}_{w}\delta^{(ub)}_{w}\gamma$
	and
	$\Desc_{\mu,\ell-1-i}'\Enc^{(u)}_{w'}\delta^{(u)}_{w'}\#\Enc_w\gamma\weak{b}\Desc_{\mu,\ell-2-i}'
	\Enc^{(ub)}_{w'}\delta^{(ub)}_{w'}\#\Enc_w\gamma$.
	We also have $\Enc_w=\Enc^{(\epsilon)}_{w}\delta^{(\epsilon)}_{w}$ and $\Enc_{w'}=\Enc^{(\epsilon)}_{w'}\delta^{(\epsilon)}_{w'}$.
	By a direct induction on $i\in\interval{0,\ppoly}$ this implies that for all $u\in\set{0,1}^i$ we have
	$\Desc_{\mu,\ell-1}\Enc_w\gamma\weak{u}\Desc_{\mu,\ell-1-i}\Enc^{(u)}_{w}\delta^{(u)}_{w}\gamma$
	and
	$\Desc'_{\mu,\ell-1}\Enc_{w'}\#\Enc_w\gamma\weak{u}\Desc'_{\mu,\ell-1-i}\Enc^{(u)}_{w'}\delta^{(u)}_{w'}\#\Enc_w\gamma$.

	For all $u\in\{0,1\}^\ppoly$ we have
	$\Enc^{(u)}_{w}=w_{\overline{u}}\Bin_{\overline{u}}$
	and
	$\Enc^{(u)}_{w'}=w_{\overline{u}}'\Bin_{\overline{u}}$
	by \cref{definition alpha w}.
	As a consequence, by inspection of the rules in \cref{rule 15},
	for all $u\in\{0,1\}^\ppoly$ we have
	$$
	\Desc_{\mu,\ppoly-1}\Enc_w\gamma\weak{u}\Desc_{\mu,-1}\Enc^{(u)}_w\delta^{(u)}_{w}\gamma\xRightarrow{h_{\mu} (w_{\overline{u}})}
	\Final_{\mu}\Enc^{(u)}_w\delta^{(u)}_{w}\gamma
	$$
	and analogously
	$$
	\Desc'_{\mu,\ppoly-1}\Enc_{w'}\#\Enc_w\gamma\weak{u}\Desc_{\mu,-1}'\Enc^{(u)}_{w'}\delta^{(u)}_{w'}\#\Enc_{w}\gamma\xRightarrow{h_{\mu}'
	(w_{\overline{u}}')}
	\Final_{\mu}'\Enc^{(u)}_{w'}\delta^{(u)}_{w'}\#\Enc_{w}\gamma\,.
	$$

	Inspecting the rules once again, we see that they are all deterministic,
	and that only symbols from $\set{0,1}$ can be read, until the state becomes $\Desc_{\mu,-1}$ or $\Desc_{\mu,-1}'$.
	After reaching such a state, it is only possible to read the symbol $h_\mu(X)$ or $h'_\mu(X)$
	(depending on whether the state is $\Desc_{\mu,-1}$ or $\Desc_{\mu,-1}'$), where $X$ is the topmost stack symbol.
	No further transitions are possible after reading this symbol (and reaching the state $\Final_\mu$ or $\Final'_\mu$).
	It follows that
	all runs that one can execute
	from the configuration $\Desc_{\mu,\ppoly-1}\Enc_{w}\gamma$
	are deterministic and the maximal
	such runs are precisely the runs reading a word from the set
	$\{uh_{\mu}(w_{\overline{u}})\mid u\in\{0,1\}^\ppoly\}$.
	Analogously,
	all runs that one can execute
	from the configuration $\Desc_{\mu,\ppoly-1}'\Enc_{w'}\#\Enc_{w}\gamma$
	are deterministic and the maximal
	such runs are precisely the runs reading a word from the set
	$\{uh_{\mu}'(w_{\overline{{u}}}')\mid u\in\{0,1\}^\ppoly\}$.
	Thus, \cref{L sufficient} and
	hence the lemma hold by the following equivalences:
	\begin{align*}
		&\Desc_{\mu,\ppoly-1}\Enc_{w}\gamma\approx
		\Desc_{\mu,\ppoly-1}'\Enc_{w'}\#\Enc_{w}\gamma \\
		\Longleftrightarrow\qquad&
		\{uh_{\mu}(w_{\overline{u}})\mid u\in\{0,1\}^\ppoly\}=
		\{uh_{\mu}'(w_{\overline{u}}')\mid u\in\{0,1\}^\ppoly\}\\
		\Longleftrightarrow\qquad& \forall u\in\{0,1\}^\ppoly.\,
		h_{\mu}(w_{\overline{u}})=
		h_{\mu}'(w_{\overline{u}}')\\
		\Longleftrightarrow\qquad& h_{\mu}(w)=h_{\mu}'(w')\,.
	\qedhere\end{align*}
\end{proofnoqed}

\subsection{A gadget for pushing successor configurations.}\label{S Pushing}

Before listing quite involved rules, the following lemma
states that one can design a gadget that allows to push a
successor configuration of a configuration whose encoding
is assumed to be on the top of the stack:

\begin{lemma}\label{lemma push defender}
	For all $t=(p,U,q,Y,d)\in T$,
	when adding the rules \crefrange{push 16A}{push 26} below,
	for all length-$N$ configurations $w$
	of $\M$ and all $\gamma\in\Gamma^*$ we have
	$$
	\Push_{t}\#\Enc_{w}\gamma\approx
	\Push_{t}'\#\Enc_{w}\gamma
	\quad\Longleftrightarrow\quad\!\!\begin{array}{l}
	\exists Z\in\Upsilon, w'\in\Upsilon^*(q,Z)\Upsilon^*.\,w\vdash_{(t,Z)}w'\text{ and }\\\quad
 	\Play_{(q,Z)}\#\Enc_{w'}\#\Enc_{w}\gamma\approx
	\Play_{(q,Z)}'\#\Enc_{w'}\#\Enc_{w}\gamma,\end{array}$$
\end{lemma}

\renewcommand{\Check}{\textup{\textsc{Check}}}
\newcommand{\HeadChk}{\textup{\textsc{HeadChk}}}
\newcommand{\FinChk}{\textup{\textsc{FinChk}}}
\newcommand{\Proceed}{\textup{\textsc{Ctd}}}
\newcommand{\Next}{\textup{\textsc{Next}}}
\noindent
For all $t=(p,U,q,Y,d)\in T$ and all $X\in\Gamma$ we add the rule
\begin{align}
	\langle \Push_t,\Push_t'\rangle X\rewrite{\textsc{Def}}
	\{\langle \Push_{(t,Z)},\Push_{(t,Z)}\rangle X\mid
	Z\in\Upsilon\}\,,
	\label{push 16A}
\end{align}
and for all $\mu=(t,Z)\in\{t\}\times\Upsilon$
we add the following rules, where $\Theta=\Upsilon\cup(Q_\M\times\Upsilon)$
and where $\Omega=\{\#_j\mid j\in\interval{0,\ppoly-1}\}$:
\begin{align}
	&\langle \Push_\mu,\Push_\mu'\rangle X\rewrite{\textsc{Def}}
	\langle \Check_\mu,\Check_\mu'\rangle\ \Theta\{0,1\}^\ppoly\#_0\Theta 1^\ppoly X\,,\hspace{-20em}
	\label{push 17}\displaybreak[0]\\
	&\langle \Proceed_\mu,\Proceed_\mu'\rangle X\rewrite{\textsc{Def}}
	\set{\langle \Next_\mu,\Next_\mu'\rangle X, \langle \Done_\mu,\Done_\mu'\rangle X}\,,\hspace{-20em}
	\label{push 18-new}\displaybreak[0]\\
	&\langle \Next_\mu,\Next_\mu'\rangle X\rewrite{\textsc{Def}}
	\langle \Check_\mu,\Check_\mu'\rangle\ \Theta\{0,1\}^\ppoly\Omega X\,,\hspace{-20em}
	\label{push 18}\displaybreak[0]\\
	&
	\Check_{\mu}X \rewrite{1}{\Binsucc}\,,&&
	\Check_{\mu}'X \rewrite{1}\Binsucc'\,,
	\label{push 20}\displaybreak[0]\\
	&
	\Check_{\mu}X \rewrite{2}{\Proceed}_\mu X\,,&&
	\Check_{\mu}'X \rewrite{2}{\Proceed}_\mu'X\,,
	\label{push 21}\displaybreak[0]\\
	&
	\Done_{\mu}X \rewrite{1}{\HeadChk}_\mu\#X\,,&&
	\Done_{\mu}'X \rewrite{1}\HeadChk_{\mu}'\#X\,,
	\label{push 22-new}\displaybreak[0]\\
	&\langle \HeadChk_{\mu},\HeadChk_\mu'\rangle \#
	\rewrite{\textsc{Prefix-Check} \#L^*(Q_\M\times\Upsilon)L^*\#}\,,\quad
	\text{where $L=\Upsilon\cup\Omega\cup\set{0,1}$,}\hspace{-18em}
	\label{push 22}\displaybreak[0]\\
	&
	\Done_{\mu}X \rewrite{2}{\FinChk}_\mu\#X\,,&&
	\Done_{\mu}'X \rewrite{2}\FinChk_{\mu}'\#X\,,
	\label{push 23-new}\displaybreak[0]\\
	&
	\langle \FinChk_{\mu},\FinChk_\mu'\rangle \#
	\rewrite{\textsc{Prefix-Check} \#\Theta 0^\ppoly(\Gamma\setminus\set{\#})^*\#}\,,\hspace{-20em}
	\label{push 23}\displaybreak[0]\\
	&
	\Done_{\mu}X \rewrite{3}{\Succ}_\mu\# X\,,&&
	\Done_{\mu}'X \rewrite{3}\Succ_{\mu}'\# X\,,
	\label{push 24}\displaybreak[0]\\
	&
	\Done_{\mu}X
	\rewrite{4}\Play_{(q,Z)}\# X\,,&
	&\Done_{\mu}'X\rewrite{4}\Play_{(q,Z)}'\# X\,.&\label{push 26}
\end{align}
We remark that the final control states of the form $\Play_{(q,Z)}$
and $\Play_{(q,Z)}'$ will be connected to a gadget further below.
\begin{proofnoqedof}{\cref{lemma push defender}}
	Let $t=(p,U,q,Y,d)\in T$,
	 let $w\in\Upsilon^*(Q_\M\times\Upsilon)\Upsilon^*$ be some configuration of length $N$,
	and let $\gamma\in\Gamma^*$ be some stack content.
	For any stack content $\delta\in\Gamma^*$ and control states
	$r$ and $r'$ we introduce the notation
	$\langle r\approx r'\rangle\delta$ and
	$\langle r\not\approx r'\rangle\delta$
	as an abbreviation for
	$r\delta\approx r'\delta$
	and $r\delta\not\approx r'\delta$, respectively.
	For all $i\in\interval{0,N-1}$ and all
	$Y_i,\dots,Y_{N-1}\in\Theta$,
	let us also introduce the notation
	$$\delta(Y_i,\dots,Y_{N-1})=Y_i\Bin_i\#_{\MSB(i,i+1)}Y_{i+1}\Bin_{i+1}\dots\#_{\MSB(N-2,N-1)}Y_{N-1}\Bin_{N-1}\#\Enc_{w}\gamma\,.$$

	We have the following claim:

	\begin{claim}\label{claim1}
		For all $\mu\in\{t\}\times\Upsilon$, all $i\in\interval{1,N-2}$,
		and all $Y_i,\dots,Y_{N-1}\in\Theta$
		we have
		$$\langle \Proceed_{\mu}\approx\Proceed_{\mu}'\rangle\delta(Y_i,\dots,Y_{N-1})\quad\Longleftrightarrow\quad
		\exists Y_{i-1}\in\Theta.\, \langle\Proceed_{\mu}\approx\Proceed_{\mu}'\rangle
		\delta(Y_{i-1},Y_i\dots,Y_{N-1})\,.$$
	\end{claim}

	\begin{proofof}{\cref{claim1}}
		First we claim that
		$\langle\Done_{\mu}\not\approx\Done_{\mu}'\rangle
		\delta(Y_i,\dots,Y_{N-1})$:
		indeed, the prefix
		$$\#Y_i\Bin_i\#_{\MSB(i,i+1)}Y_{i+1}\Bin_{i+1}\dots\#_{\MSB(N-2,N-1)}Y_{N-1}\Bin_{N-1}\#$$
		of $\#\delta(Y_i,\dots,Y_{N-1})$
		is not in the regular language
		$\#\Theta 0^\ppoly(\Gamma\setminus\set{\#})^*\#$,
		simply because $\Bin_i\not=0^\ell$ due to $i\not=0$;
		thus by the rules in \cref{push 23-new,push 23}
		it follows that $\langle\Done_{\mu}\not\approx\Done_{\mu}'\rangle
		\delta(Y_i,\dots,Y_{N-1})$.

		Let us fix any $i\in\interval{1,N-2}$. The claim follows from the
		following equivalences, where the second one follows from
		the just proven
		$\langle\Done_{\mu}\not\approx\Done_{\mu}'\rangle
		\delta(Y_i,\dots,Y_{N-1})$
		and the penultimate equivalence follows
		from the equivalence
		$\langle \Binsucc\approx \Binsucc'\rangle y\#_j\delta(Y_i,\dots,Y_{N-1})\Longleftrightarrow y=\Bin_{i-1}$ and
		$j=\MSB(i-1,i)$, which holds for $y\in\{0,1\}^\ell$ by \cref{lemma succ}:
		\begin{align*}
			&\langle\Proceed_{\mu}\approx\Proceed_{\mu}'\rangle
			\delta(Y_i,\dots,Y_{N-1})\\
		\stackrel{\text{\hspace{-2em}\labelcref{push 18-new},\labelcref{push 18}\hspace{-2em}}}
		{\Longleftrightarrow}\quad&
			\left(\exists Y_{i-1}\in\Theta,y\in\{0,1\}^\ell,\#_j\in\Omega.\,
			\langle\Check_{\mu}\approx\Check_{\mu}'\rangle
			Y_{i-1}y\#_j\delta(Y_i,\dots,Y_{N-1})\right)\text{ or}\\
			&\hspace{12em}\langle \Done_{\mu}\approx\Done_{\mu}'\rangle
				\delta(Y_i,\dots,Y_{N-1})\\
		\Longleftrightarrow\quad&
			\exists Y_{i-1}\in\Theta,
		y\in\{0,1\}^\ell,\#_j\in\Omega.\,
			\langle \Check_{\mu}\approx\Check_{\mu}'\rangle Y_{i-1}
			y\#_j\delta(Y_i,\dots,Y_{N-1})\\
		\stackrel{\text{\hspace{-2em}\labelcref{push 20},\labelcref{push 21}\hspace{-2em}}}
		{\Longleftrightarrow}\quad&
			\exists Y_{i-1}\in\Theta, y\in\{0,1\}^\ell,\#_j\in\Omega.\\
			&
			\langle\Binsucc\approx\Binsucc'\rangle
			y\#_j\delta(Y_i,\dots,Y_{N-1})\text{ and }
			\langle\Proceed_{\mu}\approx\Proceed_{\mu}'\rangle
			Y_{i-1}y\#_j\delta(Y_i,\dots,Y_{N-1})\\
		\Longleftrightarrow\quad&
		\exists Y_{i-1}\in\Theta.\,
			\langle\Proceed_{\mu}\approx\Proceed_{\mu}'\rangle
			Y_{i-1}\Bin_{i-1}\#_{\MSB(i-1,i)}\delta(Y_i,\dots,Y_{N-1})\\
			\Longleftrightarrow\quad&
			\exists Y_{i-1}\in\Theta.\,
			\langle\Proceed_{\mu}\approx\Proceed_{\mu}'\rangle
			\delta(Y_{i-1},Y\dots,Y_{N-1})\,.
		\end{align*}
		This completes the proof of \cref{claim1}.
	\end{proofof}

	Next we have the following claim, where we recall that $t=(p,U,q,Y,d)$ and $\mu=(t,Z)$:

	\begin{claim}\label{claim2}
		For all $\mu=(t,Z)\in\{t\}\times\Upsilon$ and all
		$w'=Y_0\dots Y_{N-1}\in\Theta^N$ we have
		$
		\langle\Proceed_{\mu}\approx\Proceed_{\mu}'\rangle\delta(Y_0,\dots,Y_{N-1})
		$
		if, and only if, $w\vdash_\mu w'$ and
		$\langle \Play_{(q,Z)}\approx\Play_{(q,Z)}'\rangle
		\#\Enc_{w'}\#\Enc_{w}\gamma$.
	\end{claim}

	\begin{proofnoqedof}{\cref{claim2}}
		Let $\mu=(t,Z)\in\{t\}\times\Upsilon$
		and let $w'=Y_0\dots Y_{N-1}\in\Theta^N$.
		Recall that $\delta(Y_0,\dots,Y_{N-1})=\Enc_{w'}\#\Enc_w\gamma$.
		Firstly, we claim that for all $Y\in\Theta$, all $y\in\{0,1\}^\ppoly$,
		and all $\#_j\in\Omega$ we have
		\begin{properties}
			&\langle \Check_{\mu}\not\approx\Check_{\mu}'\rangle
			Yy\#_j\Enc_{w'}\#\Enc_w\gamma\,.
			\label{approx ineq}
		\end{properties}
		Indeed, $\langle \Binsucc\not\approx\Binsucc'\rangle y\#_j\Enc_{w'}\#\Enc_w\gamma$
		follows directly from \cref{lemma succ}, hence by the rules in \cref{push 20}
		we obtain \cref{approx ineq} as a consequence.

		Secondly, we remark that the presence of the prefix checking
		rule from \cref{push 23} (reachable due to the rules in \cref{push 23-new}) does not impact
		the equivalence
		$\langle\Done_{\mu}\approx\Done_{\mu}\rangle
		\Enc_{w'}\#\Enc_w\gamma$,
		simply as the unique prefix of $\#\Enc_{w'}\#\Enc_w\gamma$
		lying in $\#(\Gamma\setminus\{\#\})^*\#$,
		namely $\#\Enc_{w'}\#$, clearly lies in
		the regular language $\#\Theta 0^\ppoly(\Gamma\setminus\set{\#})^*\#$ appearing in the rule.
		Hence, it follows that the relevant equivalence
		$\langle\Done_{\mu}\approx\Done_{\mu}\rangle
		\Enc_{w'}\#\Enc_w\gamma$
		only depends on the application of the
		rules appearing in \cref{push 22,push 22-new,push 24,push 26}.
		The claim now follows from the following equivalences:
		\begin{align*}
			&\langle\Proceed_{\mu},\Proceed_{\mu}'\rangle\delta(Y_0,\dots,Y_{N-1})\\
			\Longleftrightarrow\qquad\quad&\langle\Proceed_{\mu},\Proceed_{\mu}'\rangle\Enc_{w'}\#\Enc_w\gamma\\
			\stackrel{\text{\hspace{-2em}\labelcref{push 18-new},\labelcref{push 18}\hspace{-2em}}}{\Longleftrightarrow}\qquad\quad & \exists Y\in\Theta,y\in\{0,1\}^\ppoly,\#_j\in\Omega:
			\langle\Check_{\mu}\approx\Check_{\mu}'\rangle Yy\#_j
			\Enc_{w'}\#\Enc_w\gamma\\
			&\text{or } \langle\Done_{\mu}\approx\Done_{\mu}'\rangle
			\Enc_{w'}\#\Enc_w\gamma\\
			\stackrel{\text{\labelcref{approx ineq}}}{\Longleftrightarrow}\qquad\quad &
			\langle\Done_{\mu}\approx\Done_{\mu}'\rangle
			\Enc_{w'}\#\Enc_w\gamma\\
			\stackrel{\text{\hspace{-3em}\labelcref{push 22-new},\labelcref{push 22},\labelcref{push 24},\labelcref{push 26}\hspace{-3em}}}{\Longleftrightarrow}\qquad\quad &
			w'\in\Upsilon^*(Q_\M\times\Upsilon)\Upsilon^*\text{, }
			\langle\Succ_{\mu}\approx\Succ_{\mu}'\rangle\#\Enc_{w'}\#\Enc_w\gamma,\\
			&\text{and }
			\langle\Play_{(q,Z)}\approx\Play_{(q,Z)}'\rangle\#\Enc_{w'}\#\Enc_w\gamma\\
			\Longleftrightarrow\qquad\quad &
			w'\text{ is a configuration, }
			\langle\Succ_{\mu}\approx\Succ_{\mu}'\rangle\#\Enc_{w'}\#\Enc_w\gamma,\\
			&\text{and }
			\langle\Play_{(q,Z)}\approx\Play_{(q,Z)}'\rangle\#\Enc_{w'}\#\Enc_w\gamma)\\
			\stackrel{\text{\hspace{-2em}\cref{lemma Succ mu}}\hspace{-2em}}{\Longleftrightarrow}\qquad\quad &
			w'\text{ is a configuration, }w\vdash_\mu w'\text{, and }\langle\Play_{(q,Z)}\approx\Play_{(q,Z)}'\rangle\#\Enc_{w'}\#\Enc_w\gamma\\
			{\Longleftrightarrow}\qquad\quad &
			w\vdash_\mu w'\text{ and }\langle \Play_{(q,Z)}\approx\Play_{(q,Z)}'\rangle
			\#\Enc_{w'}\#\Enc_{w}\gamma\,.
		\qedhere\end{align*}
	\end{proofnoqedof}

	The lemma now follows from the following equivalences:
	\begin{align*}
		&\langle \Push_t,\Push_t'\rangle\#\Enc_{w}\gamma\\
		\stackrel{\labelcref{push 16A}}{\Longleftrightarrow}\qquad
		&\exists \mu\in\{t\}\times\Upsilon:
		\langle \Push_{\mu}\approx\Push_{\mu}'\rangle \#\Enc_{w}\gamma\\
		\stackrel{\text{\labelcref{push 17}}}{\Longleftrightarrow}\qquad &
		\exists \mu\in\{t\}\times\Upsilon,Y_{N-2},Y_{N-1}\in\Theta, y\in\{0,1\}^\ppoly.\\
		&
		\hspace{8em}\langle\Check_{\mu}\approx\Check_{\mu}'\rangle Y_{N-2}y\#_0Y_{N-1}1^\ell\#
		\Enc_{w}\gamma\\
		\stackrel{\text{\hspace{-2em}\labelcref{push 20},\labelcref{push 21}\hspace{-2em}}}{\Longleftrightarrow}\qquad &
		\exists \mu\in\{t\}\times\Upsilon, Y_{N-2},Y_{N-1}\in\Theta, y\in\{0,1\}^\ppoly.\\
		&	\Big(\biggl.\langle\Binsucc\approx
		\Binsucc'\rangle y\#_0Y_{N-1}1^\ell\#\Enc_{w}\gamma\text{ and }
		\langle\Proceed_{\mu}\approx\Proceed_{\mu}'\rangle Y_{N-2}y\#_0Y_{N-1}1^\ell\#\Enc_{w}\gamma\biggl.\Big)\\
		\stackrel{\text{\hspace{-2em}\cref{lemma succ}\hspace{-2em}}}{\Longleftrightarrow}\qquad &
		\exists \mu\in\{t\}\times\Upsilon, Y_{N-2},Y_{N-1}\in\Theta,\,
		\langle\Proceed_{\mu}\approx\Proceed_{\mu}'\rangle Y_{N-2}\Bin_{N-2}\#_0
		Y_{N-1}\Bin_{N-1}\#\Enc_{w}\gamma\\
		\Longleftrightarrow\qquad&\exists \mu\in\{t\}\times\Upsilon, Y_{N-2},Y_{N-1}\in\Theta.\,
		\langle\Proceed_{\mu}\approx\Proceed_{\mu}'\rangle \delta(Y_{N-2},Y_{N-1})\\
		\stackrel{\text{\hspace{-2em}\cref{claim1}\hspace{-2em}}}{\Longleftrightarrow}\qquad&
		\exists \mu\in\{t\}\times\Upsilon,Y_{N-3},Y_{N-2},Y_{N-1}\in\Theta.\,
		\langle\Proceed_{\mu}\approx\Proceed_{\mu}'\rangle
		\delta(Y_{N-3},Y_{N-2},Y_{N-1})\\
		\cdots&\\
		\stackrel{\text{\hspace{-2em}\cref{claim1}\hspace{-2em}}}{\Longleftrightarrow}\qquad&
		\exists \mu\in\{t\}\times\Upsilon,Y_0,\dots,Y_{N-1}\in\Theta.\,
		\langle\Proceed_{\mu}\approx\Proceed_{\mu}'\rangle
		\delta(Y_0,\dots,Y_{N-1})\\
		\stackrel{\text{\hspace{-2em}\cref{claim2}\hspace{-2em}}}{\Longleftrightarrow}\qquad&
		\exists \mu=(t,Z)\in\{t\}\times\Upsilon,w'\in\Theta^N.\,w\vdash_\mu w'\text{ and }
		\langle\Play_{(q,Z)}\approx\Play_{(q,Z)}'\rangle
		\#\Enc_{w'}\#\Enc_{w}\gamma\\
		\Longleftrightarrow\qquad&
		\exists Z\in\Upsilon, w'\in\Upsilon^*(q,Z)\Upsilon^*.\,w\vdash_\mu w'\text{ and }\langle\Play_{(q,Z)}\approx\Play_{(q,Z)}'\rangle
		\#\Enc_{w'}\#\Enc_{w}\gamma\,.
	\qedhere\end{align*}
\end{proofnoqedof}

\subsection{Simulating \texorpdfstring{$\M$}{M}.}\label{S Simulation}

The next gadget allows us to decide, assuming that the topmost stack content is
of the form $\#\alpha_w$, whether the configuration $w$ is indeed accepting.
\begin{samepage}
For all $(p,U)\in Q_\M\times\Upsilon$ and all $t\in T_{(p,U)}$
we add the following rules:
\begin{align}
	&\Play_{(p,U)}\#\rewrite{t}\Push_{t}\#\quad
	\text{if $p\in Q_\exists$,}\hspace{5em}
	\Play_{(p,U)}'\#\rewrite{t}\Push_{t}'\# \quad\text{if $p\in Q_\exists$,}\label{Play 2}\\
	&
	\langle\Play_{(p,U)},\Play_{(p,U)}'\rangle \#\rewrite{\Defender}
	\{\langle \Push_{t'},\Push_{t'}'\rangle \text\#\mid
	t'\in T_{(p,U)}\}\quad
	\text{if $p\in Q_\forall$.}\label{Play 3}
\end{align}
\end{samepage}

\begin{samepage}
We have the following accompanying lemma:
\begin{lemma}\label{lemma play}
	For all $(p,U)\in Q_\M\times\Upsilon$,
	all length-$N$ configurations $w\in\Upsilon^*(p,U)\Upsilon^*$ of $\M$,
	and all $\gamma\in\Gamma^*$ we have
	$$
	\Play_{(p,U)}\#\Enc_{w}\gamma\not\approx \Play_{(p,U)}'
	\#\Enc_{w}\gamma\quad
	\Longleftrightarrow\quad \text{$w$ is accepting}
	\,.
	$$
	Moreover, both
	$\{\class{\Play_{(p,U)}\gamma}\mid\gamma\in\Gamma^*\}$
	and
	$\{\class{\Play_{(p,U)}'\gamma}\mid\gamma\in\Gamma^*\}$
	are finite sets of classes all of which are weakly
	bisimulation finite.
\end{lemma}
\end{samepage}

\begin{proofnoqed}
	Let us first prove the second statement of the lemma.
	Since the rules in \cref{Play 2,Play 3}
	mutually depend on the rules \crefrange{push 16A}{push 26} we analyze them together.
	First of all, observe that, immediately by the definition of the gadget,
	neither of the prefix checking
	rules in \cref{push 22,push 23} can contribute
	to an infinity of classes.
	An important consequence of \cref{equiv-because-equiv-below} is
	that the class of every configuration
	$qX\eta$ is determined by the tuple $(q,X,(\class{r\eta})_{r\in\eats{X}(q)})$.

	By inspection of the rules
	in \crefrange{push 16A}{push 26}, \labelcref{Play 2}, and \labelcref{Play 3}
	one realizes that the only rule that {\em decreases}
	the stack height are the rules in \cref{push 20},
	leading to the control states $\Binsucc$ and $\Binsucc'$, respectively.
	Hence the only such classes of the above form
	$\class{r\eta}$ are the classes of the form
	$\class{\Binsucc\gamma}$ or $\class{\Binsucc'\gamma}$.
	But the sets
	$\{\class{\Binsucc\gamma}\mid\gamma\in\Gamma^*\}$
	and
	$\{\class{\Binsucc'\gamma}\mid\gamma\in\Gamma^*\}$
	are finite sets of classes all of which are weakly bisimulation
	finite by \cref{lemma succ}.
	Thus, it follows that
	$\{\class{\Play_{(p,U)}\gamma}\mid\gamma\in\Gamma^*\}$
	and
	$\{\class{\Play_{(p,U)}'\gamma}\mid\gamma\in\Gamma^*\}$
	are finite sets of classes all of which are weakly bisimulation finite.

	Let us now prove the first statement of the lemma.
	Let $(p,U)\in Q_\M\times\Upsilon$, let $w=\Upsilon^*(p,U)\Upsilon^*$ be a
	configuration of $\M$ and let $\gamma\in\Gamma^*$ be any stack content.
	We prove the statement by induction on the
	length of the longest computation starting in $w$.
	We make a case distinction whether $w$ is universal or existential.

	If $w$ is existential we have the following
	equivalences:
	\begin{align*}
		&w\text{ is accepting}\\
	    \Longleftrightarrow\qquad &
		\exists t=(p,U,q,Y,d)\in T_{(p,U)}, Z\in\Upsilon, w'\in\Upsilon^*(q,Z)\Upsilon^*.\,
		w\vdash_{(t,Z)} w'\text{ and $w'$ is accepting}\\
	    \stackrel{\text{IH}}{\Longleftrightarrow}\qquad&
		\exists t=(p,U,q,Y,d)\in T_{(p,U)}, Z\in\Upsilon, w'\in\Upsilon^*(q,Z)\Upsilon^*.\,
		w\vdash_{(t,Z)} w'\mbox{ and}\\
	    &\hspace{12em}
		\Play_{(q,Z)}\#\Enc_{w'}\#\Enc_{w}\gamma
		\not\approx
		\Play_{(q,Z)}'\#\Enc_{w'}\#\Enc_{w}\gamma\\
	    \hspace{1em}\stackrel{\hspace{-2em}\text{\cref{lemma push defender}\hspace{-2em}}}{\Longleftrightarrow}\qquad &
		\exists t=(p,U,q,Y,d)\in T_{(p,U)}.\,
		\Push_{t}\#\Enc_{w}\gamma\not\approx
		\Push_{t}\#\Enc_{w}\gamma\\
	    \stackrel{\hspace{-2em}p\in Q_\exists,\labelcref{Play 2}\hspace{-2em}}{\Longleftrightarrow}\qquad&
		\Play_{(p,U)}\#\Enc_{w}\gamma\not\approx
		\Play_{(p,U)}'\#\Enc_{w}\gamma\,.
	\end{align*}

	If $w$ is universal
	we have the following equivalences:
	\begin{align*}
		&w\text{ is accepting}\\
	    \Longleftrightarrow\qquad &
		\forall t=(p,U,q,Y,d)\in T_{(p,U)}.\,\exists Z\in\Upsilon, w'\in\Upsilon^*(q,Z)\Upsilon^*.\,
		w\vdash_{(t,Z)} w'\text{ and $w'$ is accepting}\\
	    \stackrel{\text{IH}}{\Longleftrightarrow}\qquad&
		\forall t=(p,U,q,Y,d)\in T_{(p,U)}.\,\exists Z\in\Upsilon, w'\in\Upsilon^*(q,Z)\Upsilon^*.\,
		w\vdash_{(t,Z)} w'\mbox{ and}\\
	    &\hspace{12em}
		\Play_{(q,Z)}\#\Enc_{w'}\#\Enc_{w}\gamma
		\not\approx
		\Play_{(q,Z)}'\#\Enc_{w'}\#\Enc_{w}\gamma\\
	    \hspace{1em}\stackrel{\text{\hspace{-2em}\cref{lemma push defender}\hspace{-2em}}}{\Longleftrightarrow}\qquad &
		\forall t=(p,U,q,Y,d)\in T_{(p,U)}.\,
		\Push_{t}\#\Enc_{w}\gamma\not\approx
		\Push_{t}\#\Enc_{w}\gamma\\
	    \stackrel{\hspace{-2em}p\in Q_\forall,\labelcref{Play 3}\hspace{-2em}}{\Longleftrightarrow}\qquad&
		\Play_{(p,U)}\#\Enc_{w}\gamma\not\approx
		\Play_{(p,U)}'\#\Enc_{w}\gamma\,.
	\qedhere\end{align*}
\end{proofnoqed}

\subsection{Setting up the initial configuration.}\label{S Initial}

Recall that $\Bin_k$ denotes the binary encoding for every $k\in\interval{0,2^\ppoly-1}$,
that $x=x_0\dots x_{n-1}$ is the input to $\M$ and that
$q_0\in Q$ is the initial state of $\M$.
We add the following rules, where $\Theta_0=\Upsilon\cup\{(q_0,x_0)\}$ and
$\Omega=\{\#_j\mid j\in\interval{0,\ppoly-1}\}$:
\begin{align}
	&\langle \Init\#,\Init'\rangle \# \rewrite{\textsc{Def}}
	\langle \Check,\Check' \rangle\
	\Box\alpha_{N-2}\#_0\Box\alpha_{N-1}\#\,,\hspace{-20em}\label{init 1}\displaybreak[0]\\
	&\langle \Proceed,\Proceed'\rangle X\rewrite{\textsc{Def}}
	\set{\langle \Next,\Next'\rangle X, \langle \Done,\Done'\rangle X}\,,\hspace{-20em}
	\label{init 1-new}\displaybreak[0]\\
	&\langle \Next,\Next'\rangle X \rewrite{\textsc{Def}}
	\langle \Check,\Check' \rangle\  \Theta_0\{0,1\}^\ppoly\Omega X\,,\hspace{-20em}
	\label{init 2}\displaybreak[0]\\
	&
	\Check X\ \rewrite{1}{\Binsucc}\,,&&
	\Check'X\ \rewrite{1}{\Binsucc}'\,,
	\label{init 4}\displaybreak[0]\\
	&
	\Check X\ \rewrite{2}{\Proceed} X\,,&&
	\Check'X\ \rewrite{2}{\Proceed}'X\,,
	\label{init 4-2}\displaybreak[0]\\
	&
	\Done X
	\rewrite{1}\IniChk \# X\,,&
	&\Done'X\rewrite{1}\IniChk'\# X\,.\label{init 5-new}\displaybreak[0]\\
	&
	\langle \IniChk,\IniChk' \rangle\#
	\rewrite{\textsc{Prefix-Check} \#(q_0,x_0)\alpha_0\#_0x_1\alpha_1\dots
	\#_{\MSB(n-2,n-1)}x_{n-1}\alpha_{n-1}(\Omega\Box\{0,1\}^\ppoly)^*\#}\,,\hspace{-18em}\label{init 5}\displaybreak[0]\\
	&
	\Done X
	\rewrite{1}\Play_{(q_0,x_0)}\# X\,,&
	&\Done'X\rewrite{1}\Play_{(q_0,x_0)}'\# X\,.&\label{init 6}
\end{align}

\begin{lemma}\label{lemma init}
	For all $\gamma\in\Gamma^*$ we have
	$\Init\#\gamma\not\approx\Init'\#\gamma$
	if, and only if, $x\in L(\M)$.
	Moreover, $\{\class{\Init\gamma}\mid\gamma\in\Gamma^*\}$
	and $\{\class{\Init'\gamma}\mid\gamma\in\Gamma^*\}$
	are both finite sets of classes all of which are weakly bisimulation
	finite.
\end{lemma}

\begin{proof}
	Let us first prove the first statement of the lemma.
	Since the proof is very closely related to the proof
	of \cref{lemma push defender}, so
	we only sketch it here.

	Let $w_x=(q_0,x_0)x_1\dots x_{n-1}\Box^{N-n}$.
	We have the following equivalences, where the first
	equivalence can be proven analogously as \cref{lemma push defender}:
	\begin{align*}
		\Init\#\gamma\not\approx\Init'\#\gamma\qquad\Longleftrightarrow\qquad
		&\Play_{(q_0,x_0)}\#\Enc_{w_x}\#\gamma\not\approx
		\Play_{(q_0,x_0)}'\#\Enc_{w_x}\#\gamma\\
		\stackrel{\hspace{-2em}\text{\cref{lemma play}}\hspace{-2em}}{\Longleftrightarrow}\qquad
		&\text{$w_x$ is accepting}\\
		\Longleftrightarrow\qquad &
		x\in L(\M)\,.
	\end{align*}

	The second statement of the lemma can be shown in the same way as the analogous statement of \cref{lemma play}.
\end{proof}

We conclude our construction by adding the following initial rules:
\begin{align}
	&q_\updownarrow\#\rewrite{a}q_\updownarrow A\#\,,\label{final 1}\\
	&q_\updownarrow A\rewrite{a}q_\updownarrow AA\,,\label{final 2}\\
	&q_\updownarrow A\rewrite{a}q_\updownarrow\,,\label{final 3}\\
	&q_\updownarrow A\rewrite{\$}q_\varepsilon\,,\label{final 4}\\
	&q_\varepsilon A\rewrite{\varepsilon}q_\varepsilon\,,\label{final 5}\\
	&q_\varepsilon\#\rewrite{\varepsilon}\Init\#\,,\label{final 6}\\
	&q_\updownarrow\#\rewrite{\$}\Init' \#\,.\label{final 7}
\end{align}

The following lemma provides the final desired reduction.

\begin{lemma}
	We have $x\in L(\M)$ if, and only if,
	$(\Ll(\Pp),q_\updownarrow\#)$ is not weakly bisimulation finite.
\end{lemma}

\begin{proof}
	Let us first assume $x\in L(\M)$.
	By \cref{lemma init} we have $\Init\#\not\approx\Init'\#$.
	For $n\geq 1$ we have $q_\updownarrow A^n\#\weak{\$}\Init\#$,
	while reading $\$$ from $q_\updownarrow\#$ necessarily leads to $\Init'\#\not\approx\Init'\#$.
	It follows that $q_\updownarrow A^n\#\not\approx q_\updownarrow\#$ for all $n\geq 1$.
	In consequence $q_\updownarrow\#$ is the only reachable configuration in its class (once we enter the configuration $\Init\#$ or $\Init'\#$, we cannot read $\$$ any more).
	Observe now that if $q_\updownarrow A^n\#\weak{w}q_\updownarrow\#$ for some $n\in\Nat$, then $|w|\geq n$;
	on the other hand $q_\updownarrow A^n\#\weak{a^n}q_\updownarrow\#$.
	Then $\dist(\class{q_\updownarrow A^n\#},\class{q_\updownarrow\#})=n$, which
	implies $q_\updownarrow A^n\#\not\approx q_\updownarrow A^m\#$
	for all $n,m\in\N$ with $n\not=m$.
	Moreover, we have $q_\updownarrow\#\to^* q_\updownarrow A^n\#$ for all $n\in\Nat$.
	Hence, $q_\updownarrow\#$ is not weakly	bisimulation finite since infinitely
	many configurations, that are pairwise not weakly bisimilar, are reachable from it.

	Conversely, assume $x\not\in L(\M)$. Then $\Init\#\approx\Init'\#$ by \cref{lemma init}.
	Moreover, $\Init\#$ is weakly bisimulation finite by \cref{lemma init}.
	From this and by inspection of the rules in \crefrange{final 1}{final 7}
	one easily sees that $q_\updownarrow A^n\#\approx q_\updownarrow A^m\#$
	for all $n,m\in\N$. In fact, $q_\updownarrow\#$ is weakly bisimilar to a finite
	pointed $\varepsilon$-LTS $(\Ll,c)$,
	where $\Ll$ has the following shape: the configuration $c$ has an $a$-loop plus
	a $\$$-labeled transition to the weakly bisimulation finite
	$\class{\Init\#}=\class{\Init'\#}$.
\end{proof}

\section{Conclusion}\label{sec:conclusion}

In this paper we have shown that weak bisimulation finiteness
is 2-\EXPTIME-complete for pushdown systems with deterministic
$\varepsilon$-transitions.
This improves a previously known \ACKERMANN upper bound of
the problem and improves the previously best known 6-\EXPSPACE
upper bound when $\varepsilon$-transitions are not present.
It also generalizes the 2-\EXPTIME upper bound for regularity
of deterministic pushdown automata and tightens
a previously known \EXPTIME lower bound for the problem.
The more general case with unrestricted $\varepsilon$-transitions
is indeed challenging, since our upper bound proof heavily relied
on the fact that the underlying transition system is finitely branching.

\bibliographystyle{abbrv}

\bibliography{bib}

\end{document}